\documentclass[aps,pra,twocolumn,superscriptaddress,floatfix,nofootinbib,showpacs,longbibliography]{revtex4-2}

\usepackage[utf8]{inputenc}  
\usepackage[T1]{fontenc}     
\usepackage[table]{xcolor}
\usepackage[british]{babel}  
\usepackage[sc,osf]{mathpazo}
\usepackage[scaled=0.86]{berasans}  
\usepackage[colorlinks=true, citecolor=blue, urlcolor=blue]{hyperref}  
\usepackage{graphicx} 
\usepackage[babel]{microtype}  
\usepackage{amsmath,amssymb,amsthm,bm,amsfonts,mathrsfs,bbm} 

\usepackage{xspace}  
\usepackage{pgf,tikz}
\usepackage{xcolor}
\usepackage{multirow}
\usepackage{array}
\usepackage{bigstrut}
\usepackage{braket}
\usepackage{color}
\usepackage{natbib}
\usepackage{multirow}
\usepackage{mathtools}
\usepackage{float}
\usepackage[caption = false]{subfig}
\usepackage{xcolor,colortbl}
\usepackage{color}
\usepackage{rotating}
\usepackage{tikz}
\usepackage{ulem}
\usetikzlibrary{quantikz}
\usepackage{diagbox}
\usepackage{soul}

\newcommand{\be}{\begin{equation}}
\newcommand{\ee}{\end{equation}}
\newcommand{\ba}{\begin{eqnarray}}
\newcommand{\ea}{\end{eqnarray}}
\newcommand{\ketbra}[2]{|#1\rangle \langle #2|}
\newcommand{\tr}{\operatorname{Tr}}

\newtheorem{definition}{Definition}
\newtheorem{proposition}{Proposition}
\newtheorem{observation}{Observation}

\newtheorem{numconjecture}{Numerical Conjecture}

\def\>{\rangle}
\def\<{\langle}

\newtheorem{obj}{Objective}
\newtheorem{sce}{Scenario}

\usepackage{centernot}
\usepackage{subfig}

\begin{document}

\title{
Convex Optimization Approaches to Optimal Teleportation Fidelity in Linear Three-Party Networks }

\author{Arkaprabha Ghosal}
\email{a.ghosal1993@gmail.com}
\affiliation{Mathematics Department, University of New Orleans, 2000 Lakeshore Drive, New Orleans, Louisiana 70148, USA}

\author{Jatin Ghai}
\email{jghai98@gmail.com}
\affiliation{Optics \& Quantum Information Group, The Institute of Mathematical Sciences, CIT Campus, Taramani, Chennai 600 113, India}
\affiliation{Homi Bhabha National Institute, Training School complex, Anushakti Nagar, Mumbai 400085, India}

\author{Tanmay Saha}

\affiliation{Optics \& Quantum Information Group, The Institute of Mathematical Sciences, CIT Campus, Taramani, Chennai 600 113, India}
\affiliation{Homi Bhabha National Institute, Training School complex, Anushakti Nagar, Mumbai 400085, India}

\author{Mir Alimuddin}   
\affiliation{ICFO-Institut de Ciencies Fotoniques, The Barcelona Institute of Science and Technology, Av. Carl Friedrich Gauss 3, 08860 Castelldefels (Barcelona), Spain.}

\author{Sibasish Ghosh}
\email{sibasish@imsc.res.in}
\affiliation{Optics \& Quantum Information Group, The Institute of Mathematical Sciences, CIT Campus, Taramani, Chennai 600 113, India}
\affiliation{Homi Bhabha National Institute, Training School complex, Anushakti Nagar, Mumbai 400085, India}

\begin{abstract} 

We study the maximum achievable quantum teleportation fidelity between two distant parties, Alice and Charlie, where each of them share a bipartite quantum state only with a common intermediary, Bob, and all parties are allowed to perform {\it Local Operations and Classical Communication} (LOCC). As the structure of LOCC is complicated, we relax the set of free operations to separable (SEP) operations and formulate a convex optimization problem that provides upper bounds on the LOCC achievable fidelity value. We observe that the complexity of such optimization problem reduces significantly if we restrict ourselves to a subclass of SEP operations, where the Kraus operators of either Alice or Charlie are proportional to unitary operators, leading to a simplified convex optimization that matches the general LOCC limit for certain two-qubit states. Through explicit examples, we show that protocols initiated by Bob by performing measurements in a maximally entangled basis are not necessarily optimal, and alternative strategies can outperform them. Finally, we extend our analysis to linear networks and demonstrate that different LOCC strategies can achieve the same optimal fidelity while consuming different amounts of entanglement content.
 
\end{abstract}

\maketitle	

\section{Introduction}

Quantum Teleportation is a task in which a sender sends quantum information of an unknown quantum state to a distant receiver without physically transmitting the unknown state itself. Quantum Teleportation requires shared entanglement between the sender and receiver and the freedom of {\it Local Operations and Classical Communication} (LOCC) \cite{PhysRevLett.70.1895, PhysRevA.60.1888}. Faithful or perfect qubit teleportation specifically requires one {\it ebit} of shared entanglement (or equivalently, a two-qubit maximally entangled state) and two cbits of classical communication from the sender to the receiver \cite{PhysRevLett.70.1895}. Practically, it is hard to distribute a maximally entangled state between distant parties due to the effects of noise in the environment \cite{PhysRevA.62.012311, im2021optimal, hu2023progress, PhysRevLett.134.160803}. Consequently, we have to deal mainly with noisy entangled states that cause imperfect quantum teleportation, which is more practical than faithful teleportation \cite{PhysRevLett.90.097901}. In such a case, our first concern is to quantify the efficacy of imperfect Quantum Teleportation. A most well-established figure of merit for Quantum Teleportation is the average teleportation fidelity \cite{Horodecki1996, PhysRevA.60.1888,PhysRevA.62.012311} that measures the average closeness of the teleported state (a density matrix) to the input state and such average is taken over all possible input states with uniform weights. 

However, it is generally hard to compute the optimal value of the average teleportation fidelity over all LOCC protocols for an arbitrary shared entangled state \cite{PhysRevA.62.012311, PhysRevA.60.1888,PhysRevA.64.060301, PhysRevA.60.1888, PhysRevA.108.032617}. The reason behind this is that the structure of LOCC is complicated as there can be arbitrary number of rounds between the sender and receiver. In the literature, there exists a classical upper bound on the average teleportation fidelity while the sender and receiver have a local quantum dimension $d\geq 2$, which is 
\[
f_{classical} = \dfrac{2}{d+1}.
\]
Here, by classical, we mean that such fidelity value can be achieved only using LOCC while there is no shared entanglement. However, even if there is a preshared entanglement, some LOCC may provide same or even less fidelity value than the classical fidelity bound.  For instance, in Ref. \cite{Horodecki1996}, it was shown that if the sender and receiver share an arbitrary two-qubit state, then under a specific LOCC protocol (also known as the standard protocol) the maximal average teleportation fidelity can be analytically computed. If the computed value is greater than $\frac{2}{3}$ we say that the shared two-qubit state is useful for Quantum Teleportation \cite{Horodecki1996}. However, such a protocol is not the optimal LOCC protocol in general.

In Ref. \cite{PhysRevA.60.1888}, it has been shown that the optimal teleportation fidelity is a linear function of the optimal {\it fully entangled fraction} of the shared state. For any bipartite density matrix with local dimension $d$, {\it fully entangled fraction } measures the largest overlap of the density matrix with a maximally entangled state. In Ref. \cite{PhysRevLett.90.097901}, the authors had shown that for any two-qubit entangled state, there exists a one-way LOCC via which the optimal teleportation fidelity can be achieved. In such LOCC, both parties perform a local post-processing before performing teleportation. More precisely, either the sender or the receiver applies a state-dependent filtering operation \cite{PhysRevA.64.010101, PhysRevA.102.052419} and if the filter passes, they perform the standard protocol; otherwise, they discard the existing shared state and locally prepare a product state \cite{PhysRevLett.90.097901}. Thus, by taking an average over the filter outcomes, one deterministically obtains the {\it fully entangled fraction} value greater than the classical limit $\frac{1}{2}$ {\it iff} the preshared state is entangled \cite{PhysRevLett.90.097901, PhysRevA.62.012311,PhysRevA.97.032322}.   

An interesting fact to point out here is that the optimal teleportation fidelity of a two-qubit state does not have a one-to-one correspondence with the established entanglement measures, like the concurrence measure \cite{wootters1998entanglement}.  For instance, there exists a pair of two-qubit states in which one state has a higher concurrence value but a lower optimal teleportation fidelity value compared to the other \cite{PhysRevA.86.020304, PhysRevA.86.020304}. On the other hand, the optimal teleportation fidelity for an arbitrary two-qubit state is upper bounded by a linear function of concurrence \cite{Zhao_2010,PhysRevA.66.022307}.

In recent years, the efficient distribution of entanglement and other quantum correlations over quantum networks has been extensively studied~\cite{PhysRevA.77.022308, PhysRevA.59.169, PhysRevA.57.822, PhysRevLett.78.3031, Konrad2007, PhysRevA.89.014303, PhysRevA.104.022425, PhysRevLett.81.5932, PhysRevA.80.032301, PhysRevA.86.020302, PhysRevLett.126.170501, PhysRevResearch.3.023045}.  
Depending on the setting, different kinds of free resources are considered. Some distribution tasks allow {\it Local Operations and Shared Randomness} (LOSR) as a free resource, where classical communication is forbidden~\cite{Schmid2023understanding, PhysRevLett.108.200401, PRXQuantum.2.020301}. Others permit classical communication as a free resource.  
In this work, we focus on the framework of —{\it Local Operations and Classical Communication} (LOCC).

Once entanglement is distributed between non-neighboring parties in a quantum network, it can be used to perform various entanglement-assisted information-processing tasks~\cite{hermans2022qubit, PhysRevA.78.032321}.  
Here, we address the following fundamental question:

\medskip
{\it 
Alice and Charlie are spatially separated and wish to perform quantum teleportation. They initially share no entangled state with each other but each shares entanglement with a common party, Bob. All three parties are allowed to perform LOCC. Then, what is the maximum teleportation fidelity (optimal fidelity) achievable between Alice and Charlie under such three-party LOCC operations? Furthermore, does there exist an unique LOCC which provides the optimal fidelity irrespective of the preshared states?
}
\medskip

Such a question is a basic yet important one in the context of three-party system. Our interest in quantum teleportation stems from its clear operational meaning in entanglement theory.  
Recent results have shown that, in a sender–receiver setting, quantum dense coding is dually related to quantum teleportation \cite{10315956}. Hence, achieving a quantum advantage in teleportation also implies some quantum advantage in dense coding. The optimal teleportation fidelity in this setting acts as a {\it three-party LOCC monotone}: once this fidelity is achieved, it cannot be increased further by any three-party local post-processing.

Our question is important as in the case of a bipartite system with a shared two-qubit state, the optimal fidelity can be computed via solving a semi-definite programming (SDP) under {\it PPT} class of operations, which is a larger class of operations than LOCC although the former achieves the same teleportation fidelity as that of the latter \cite{PhysRevLett.90.097901}. It was shown that the optimal feasible solution of the SDP can always be obtained through a one-way LOCC. So we are concerned with the same problem but in a three-party case, where Bob has the opportunity to assist Alice and Charlie.

Intuitively, for Alice and Charlie to achieve a quantum advantage, the pairs of parties—Alice–Bob and Bob–Charlie—must initially share distillable entangled states \cite{PhysRevA.60.1888}. A natural suggestion is that the optimal LOCC protocol might be realized by first distributing the maximum possible amount of distillable entanglement between Alice and Charlie, followed by appropriate local post-processing. This idea is motivated by the fact that the analytical upper bound on the fully entangled fraction of a shared two-qudit state depends linearly on its entanglement content \cite{Zhao_2010}. However, even if this intuition is correct, identifying the corresponding LOCC protocol(s) in general is highly nontrivial. Moreover, it remains unclear whether there exists a sufficiently general class of LOCC operations that can always maximize the teleportation fidelity between Alice and Charlie for arbitrary preshared states.

In this paper, we address these questions in the context of a three-party scenario involving Alice, Bob, and Charlie. Specifically, Alice–Bob and Bob–Charlie share bipartite quantum states with finite local dimension 
$d\geq 2$, and all parties are allowed to perform LOCC. Before presenting the detailed analysis, we summarize our main findings below.

\subsection{Summary of results}

\begin{itemize}
    \item Since the structure of LOCC operations is notoriously complex, we focus on the set of three-party separable operations (SEP), which forms a superset of LOCC and is more amenable to analytical treatment. Within this framework, we only consider the class of SEP operations, which provide a teleportation fidelity between Alice and Charlie at least $\frac{1}{d}$ (classical upper bound) or more.

   \item Considering such class of SEP, we formulate a convex optimization problem that by theory yields an upper bound on the fully entangled fraction achievable under three-party separable (SEP), and consequently, also serves as an upper bound over all three-party LOCC protocols. However, the optimization we propose involves variables represented by very large hermitian matrices—of size 
   $256 \times 256$ or larger—which poses significant computational challenges and may considerably increase the time required to achieve numerical convergence. For this reason, we find it challenging to obtain the optimal structure of the SEP for a given preshared states.

    \item The computational complexity of the optimization problem can be significantly reduced by restricting to the class of SEP operations, where the Kraus operators of one of the parties (either Alice or Charlie) are unitary operators. Within this restricted class, we formulate a convex optimization problem that requires variable hermitian matrices with dimension at least $16 \times 16$, which much lesser than the dimensional requirement for the general convex optimization problem {\it i.e.}, where there is no such restrictions on the SEP operations. Such reduction simplifies the optimization problem. The optimal feasible solutions of such a restricted optimization problem provide upper bounds on a subclass of LOCC protocols, where it is sufficient to consider that Bob engages in multiple rounds of LOCC with Alice, and Alice-Bob performs a one-way LOCC with Charlie. Interestingly, we identify examples of preshared two-qubit states for which the optimal solutions of the simplified optimization problem under restricted SEP coincide with the best achievable performance within LOCC, indicating that, for these cases, no LOCC protocol can yield a higher value.

\item We find instances of preshared two-qubit states in which a type of LOCC protocol (contained within the restricted SEP operations described above), where Bob starts with a projection valued measurement (PVM), and followed by one-way LOCC from Alice to Charlie, is optimal, as it attains the upper bound achievable by LOCC. However, we also identify an example where the fully entangled fraction obtained from the protocol in which Bob performs a PVM in a maximally entangled basis is not optimal for distributing the teleportation channel between Alice and Charlie. In fact, there exist alternative protocols—such as those where Bob performs a PVM in a non-maximally entangled basis or where Alice initiates the protocol with a local operation—that outperform it.

\item We then extend our analysis and consider a linear quantum network with every pair of nodes share a two-qubit state, and the two end nodes aim to perform optimal quantum teleportation. We present an explicit example, where the first and second nodes share a noisy entangled state, and rest of the pairs of nodes share the same pure entangled state. We find two distinct LOCC strategies that achieve the optimal teleportation fidelity between the end nodes over all LOCC. In Strategy I, the intermediate nodes perform PVMs in a rotated maximally entangled basis, whereas in Strategy II, they perform Bell basis measurements. Notably, we find that Strategy I requires a significantly smaller amount of entanglement for the shared pure states than Strategy II to establish the same level of teleportation fidelity. Such an observation shows that changing the local basis of an entangling measurement can be beneficial with the consumption of entanglement. 
  
\end{itemize}
 
\section{Basic Notations and Preliminaries} \label{II}
If $X$ represents any quantum system, then a pure quantum state of $X$ can be represented by a vector $\ket{\psi}_X $ in the Hilbert space $\mathcal{H}_d$ of dimension $d$. Any pure or mixed state of $X$ can be represented as a density matrix $\rho_X \in \mathcal{D}(\mathcal{H}_d)$, where $\mathcal{D}(\mathcal{H}_d)$ denotes the set of trace-one, positive and linear Hermitian operators. We denote the set of all positive semi-definite operators acting on the Hilbert space $\mathcal{H}_d$ as $\mathcal{P}(\mathcal{H}_d)$. For simplicity, we denote any positive semi-definite operator $M\in \mathcal{P}(\mathcal{H}_d)$ as $M\geq \mathbf{0}$, where we denote $\mathbf{0}$ as the {\it null operator}. The set of all Hermitian operators in $\mathcal{H}_d$ is denoted by $Herm(\mathcal{H}_d)$.

\subsection{Average teleportation fidelity}
Let Alice ($\mathcal{A}$) and Bob ($\mathcal{B}$) share {\color{red}a} bipartite quantum state $\rho_{AB}$ with the same local dimension $d\geq 2$. Let $\ket{\psi}_{in}$ be an unknown state in $\mathcal{C}^d$ that Alice wants to teleport. So Alice-Bob perform a general Quantum Teleportation protocol (a LOCC protocol), say $\mathcal{T}$ on the composite state $\ket{\psi}_{in}\bra{\psi}\otimes \rho_{AB}$. Finally, Alice's system is discarded and Bob ends up with an output state,  
\begin{align}
    \rho_{\psi,out}=Tr_{A}\left( \mathcal{T}(\ket{\psi}_{in}\bra{\psi}\otimes \rho_{AB})\right). \nonumber
\end{align}
So this implies a mapping between the output state at Bob's lab with the input state of Alice or equivalently a quantum channel $\Lambda_{\mathcal{T},~\rho_{AB}}$ such that $\Lambda_{\mathcal{T},~\rho_{AB}}(\ket{\psi}_{in}\bra{\psi})=\rho_{\psi,out}$.  Note that the channel $\Lambda_{\mathcal{T},~\rho_{AB}}$ depends only on the LOCC protocol $\mathcal{T}$ and the shared state $\rho_{AB}$. The quality of the channel can be measured by computing the average fidelity between input and output state as
\begin{align}
  f(\mathcal{T},~\rho_{AB})=\int d\psi ~\bra{\psi} \rho_{\psi,out}\ket{\psi}, \label{1}
\end{align}
where the integration denotes a Haar measure over all input states because we assume that we are choosing the input state from a uniform distribution. The average fidelity $ f(\mathcal{T},~\rho_{AB})$ denotes the efficacy of the share state $\rho_{AB}$ for the Quantum Teleportation protocol $\mathcal{T}$. Generally the optimal teleportation fidelity denotes maximizing  $ f(\mathcal{T},~\rho_{AB})$ over all possible $\mathcal{T}$ as 
\begin{align}
    f^{*}(\rho_{AB}) = \max _{\mathcal{T}} f(\mathcal{T},~\rho_{AB}), \label{2}
\end{align}
which becomes an optimization problem that generally has no closed form expression.

\subsection{Fully entangled fraction}

{\it Fully entangled fraction} of a bipartite state $\rho_{AB}$ with local dimension $d\geq 2$ is a measure of closeness of $\rho_{AB}$ with its closest maximally entangled state (MES) and it is expressed as 
\begin{align}
    F(\rho_{AB}) = \max_{\ket{\Phi}\in \mathcal{C}^d \otimes \mathcal{C}^d} \bra{\Phi}\rho_{AB} \ket{\Phi}, \label{3}
\end{align}
where  $\ket{\Phi}$ is a MES. It is worth mentioning that $F(\rho_{AB})$ is not a LOCC monotone, hence, may increase under LOCC \cite{PhysRevA.62.012311} and the optimal {\it fully entangled fraction} of $\rho_{AB}$ is the maximized {\it fully entangled fraction} of $\rho_{AB}$ over all LOCC and it is defined as 
\begin{align}
    F^*(\rho_{AB})= \max_{\Lambda \in LOCC} \max_{\ket{\Phi}} \bra{\Phi}\Lambda(\rho_{AB}) \ket{\Phi} \nonumber
\end{align}
In Ref. \cite{PhysRevA.60.1888} it has been shown that the optimal teleportation fidelity $f^{*}(\rho_{AB})$ is a linear function of the optimal {\it fully entangled fraction} of the shared state $\rho_{AB}$. So for any $\rho_{AB}$ with local dimension $d$, such linear dependence can be written as 
\begin{align}
    f^{*}(\rho_{AB})=\dfrac{d~F^{*}(\rho_{AB})+1}{d+1}, \label{4}
\end{align}
Note here that the optimal LOCC protocol $\mathcal{T}^*$ that relates $F^*(\rho_{AB})$ with $f^*(\rho_{AB})$ can be written as $\mathcal{T}^*=\Lambda_{twirl}\circ \Lambda^*$. Here $\Lambda^*$ is the optimal protocol that maximizes {\it fully entangled fraction} of $\rho_{AB}$ and then both parties perform the twirling operation $\Lambda_{twirl}$ which transforms the state $\Lambda^*(\rho_{AB})$ into an isotropic state $\rho_{iso}$ without changing the {\it fully entangled fraction} of $\Lambda^*(\rho_{AB})$ \cite{PhysRevA.60.1888}. After this, Alice and Bob  perform the standard teleportation protocol $\mathcal{T}_{st}$ \cite{Horodecki1996, PhysRevA.62.012311}. So the optimal teleportation protocol can be expressed as $\mathcal{T}^*_{tel}=\mathcal{T}_{st} \circ \mathcal{T}^*$.      
 
Throughout this work, we deal with two-qubit states. For all entangled two-qubit states, the optimal {\it fully entangled fraction} is strictly greater than the classical bound $\frac{1}{2}$. Furthermore, the upper bounds on optimal {\it fully entangled fraction} can be expressed as 
\begin{align}
     F^{*}(\rho_{AB}) \leq \dfrac{1+N(\rho_{AB})}{2}\leq \dfrac{1+C(\rho_{AB})}{2},  \label{5}
\end{align}
where $N(\rho_{AB})$ is the negativity \cite{RevModPhys.81.865} and $C(\rho_{AB})$ is the concurrence \cite{PhysRevLett.80.2245} of the two-qubit state $\rho_{AB}$. Generally, the upper bounds are not always achievable, and in those cases, $F^{*}(\rho_{AB})$ has no correspondence with the measures $N(\rho_{AB})$ and $C(\rho_{AB})$. The upper bound is saturated {\it iff} the following condition is satisfied,
\begin{align}
    \rho_{AB}^{\Gamma_B} \ket{\Psi}= \lambda_{min} \ket{\Psi}, \label{6}
\end{align}
where $\Gamma_B$ is the partial transpose with respect to Bob's system and $\lambda_{min}$ is the smallest eigenvalue of $\rho_{AB}^{\Gamma_B}$ and the corresponding eigenvector $\ket{\Psi}$ is a maximally entangled state. If such a condition is not satisfied, there is no closed-form expression to compute the optimal {\it fully entangled fraction}. On that occasion, one needs to find the feasible solution of the following semi-definite programming (SDP) \cite{PhysRevLett.90.097901}:

\begin{align}
    maximize \quad  & F(X,\rho)= \dfrac{1}{2}-\tr(X\rho^{\Gamma}), \nonumber \\
     s.t. \quad \quad & \mathbf{0}\leq X \leq \mathbb{I}, \nonumber \\
      and \quad \quad   -&\dfrac{\mathbb{I}}{2} \leq X^{\Gamma} \leq \dfrac{\mathbb{I}}{2}, \label{7}
\end{align}
where the optimal feasible solution, {\it i.e.,} $F(X_{opt}, \rho)$, will give us the optimal {\it fully entangled fraction} value that is strictly greater than $\frac{1}{2}$ if the given state $\rho$ is entangled. Here, the operator $X_{opt}$ must be a rank one positive operator. 

\begin{figure}
    \includegraphics[height=160px,width=253px]{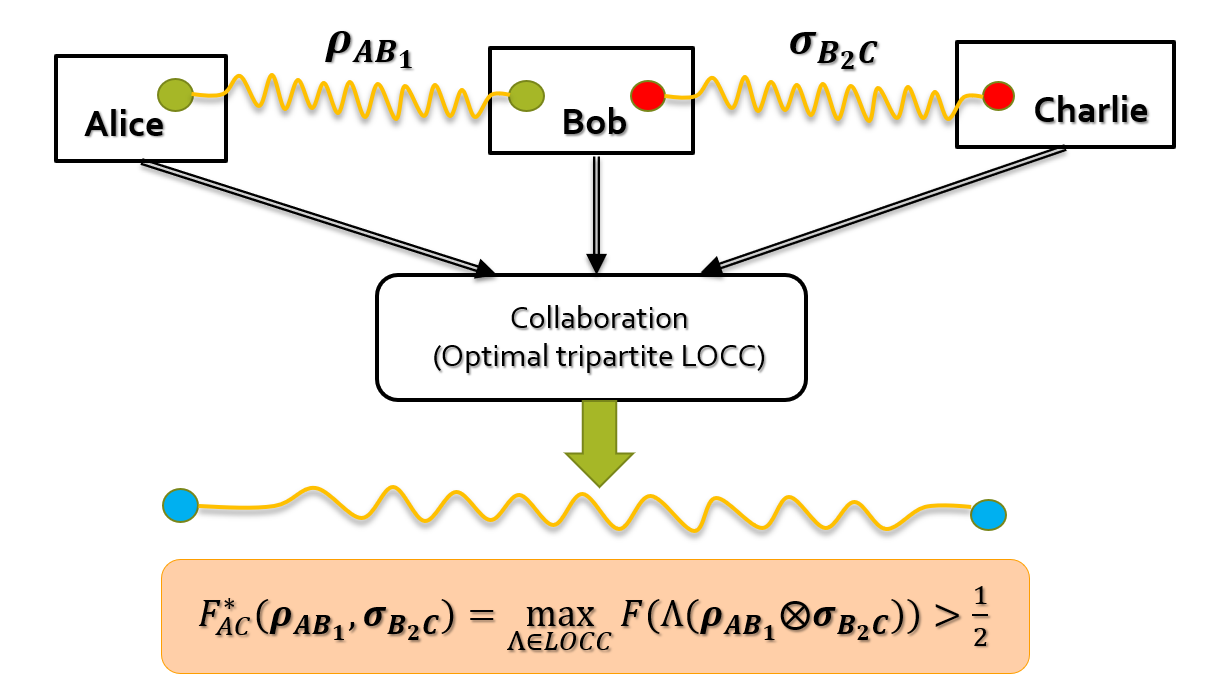}
    \caption{ (Color online) Alice-Bob share a two-qubit state $\rho_{AB_1}$, whereas Bob-Charlie share a two-qubit state $\sigma_{B_2 C}$. The task is to obtain the maximum possible teleportation fidelity, or equivalently the optimal fully entangled fraction between Alice and Charlie, where the maximization is over all three-party LOCC. The optimal fully entangled fraction must be a LOCC monotone and should give a higher value than the classical upper bound $\frac{1}{2}$ {\it iff} some amount of distillable entanglement is distributed between Alice and Charlie.  }
    \label{fig1: Collaboration three party}
\end{figure}

\section{Optimal teleportation in a three-party network} \label{III}

\begin{obj}
 Consider three spatially separated parties, namely, Alice ($\mathcal{A}$), Bob ($\mathcal{B}$) and Charlie ($\mathcal{C}$) such that $\mathcal{AB}$ share an arbitrary state $\rho_{AB_1}$ with local dimensions $d_A=d_{B_1}=d\geq 2$ and $\mathcal{BC}$ share another arbitrary state $\sigma_{B_2C}$ with local dimensions $d_{B_2}=d_C=d\geq 2$. All parties have the freedom to perform LOCC. The task is to distribute entanglement between Alice and Charlie such that its fully entangled fraction is maximized over all three-party LOCC and it is strictly greater than the classical bound $\frac{1}{d}$. 
 \end{obj}
 Our objective becomes trivial when we aim to perform a perfect quantum teleportation between Alice and Charlie, where the fully entangled fraction value reaches one. This is possible {\it iff} we can distribute $log_2~d$ \textit{ebit} of entanglement between Alice and Charlie. This is always possible by the swapping of entanglement operation \cite{werner2001all, mor1999teleportation, PhysRevA.64.064304} (see Appendix {\bf A}), where $\mathcal{AB}$ and $\mathcal{BC}$ both share $log_2d$ \textit{ebit} and Bob performs PVM in a basis of MES. Bob communicates the outcome to Alice and Charlie through noiseless classical channels and they perform further unitary corrections.

Now, to answer our generic question, we have to assume that Alice, Bob, and Charlie perform the most general LOCC. Unfortunately, the structure of the most general three-party LOCC in not very clear \cite{chitambar2014everything} as there could be many rounds. In Fig. \ref{fig1: Collaboration three party}, we provide a schematic diagram to visualize this.  
 
 It is well established that any three-party LOCC can be represented as a three-party separable operation (SEP), which has a simpler mathematical structure. However, the converse does not hold in general \cite{bennett1999quantum} i.e., any SEP cannot be realized as a LOCC. Let us first understand the structure of a most general three-party SEP operation. If $\eta_{AB_1B_2C}=\rho_{AB_1}\otimes \sigma_{B_2C}$ be the initial three-party state, then the general three-party SEP operation $\Lambda_{SEP}$ has the following effect,
\begin{align}
    &\Lambda_{SEP}(\eta_{AB_1B_2C})= \sum_i K_i~\eta_{AB_1B_2C}~K_i^{\dagger} \nonumber \\
   &= \sum_{i}\left( m^A_{i} \otimes n^B_{i} \otimes o^C_{i}\right)~\eta_{AB_1B_2C}~(m^A_{i}\otimes n^B_{i}\otimes o^C_{i})^{\dagger},
    \label{8}
\end{align}
where  $\lbrace K_i =m^A_i \otimes n^B_i \otimes o^C_i\rbrace$ is the set of Kraus operators, where $m^A_i,~n^B_i, ~o^C_i$ are complex matrices satisfying the conditions,  
\begin{align}
    \sum_i K_i ^{\dagger}K_i = \mathbb{I}_{d^4}, \quad \quad \mathbf{0}\leq K_i^{\dagger}K_i \leq \mathbb{I}_{d^4}
\end{align}
 After the operation $\Lambda\in SEP$, Bob discards his system, and finally, Alice-Charlie end up with a bipartite state of local dimensions $d_A=d_C=d$.

\subsection{Maximizing fully entangled fraction under SEP operation}
\begin{definition}
    For a given choice of preshared states $\rho_{AB_1}$ and $ \sigma_{B_2C}$ we define the maximum fully entangled fraction achievable under separable quantum operations as $F^*_{SEP}(\rho_{AB_1},\sigma_{B_2C})$ and maximum fully entangled fraction achievable under LOCC as $F^*(\rho_{AB_1},\sigma_{B_2C})$, where the condition $F^*_{SEP}(\rho_{AB_1},\sigma_{B_2C}) \geq F^*(\rho_{AB_1},\sigma_{B_2C})$ is always satisfied.  
\end{definition}
\vskip 0.3cm
Since $\Lambda_{SEP}$ has a simplified mathematical structure, unlike LOCC, under all three-party SEP operations, we can express the optimal fully entangled fraction as

\begin{small}
\begin{align}
    F^*_{SEP}(\rho_{AB_1},\sigma_{B_2C})
     = \max_{\ket{\Phi}} \max_{\Lambda_{SEP}} \bra{\Phi_{AC}}Tr_{B}\left(\Lambda_{SEP}(\eta_{AB_1B_2C})\right)\ket{\Phi_{AC}},
\end{align}
\end{small}

where the inner product is maximized over all maximally entangled state $\ket{\Phi}$, and over all trace preserving separable operations  $\Lambda_{SEP}$. We should always expect that value of $F^*_{SEP}(\rho_{AB_1},\sigma_{B_2C})$ cannot be lower than $\frac{1}{d}$. This is because for a given choice of $\eta_{AB_1B_2C}$ if there does not exists any such operation that gives a higher value than $\frac{1}{d}$, we can simply discard the existing state and replace it by a product state which gives the value $\frac{1}{d}$. So the optimal SEP operation can be realized as follows: 
\vskip 0.3cm 
\paragraph{{\it Optimal SEP operation:}}  {\it For a given initial quantum three-party state $\eta_{AB_1B_2C} = \rho_{AB_1} \otimes \sigma_{B_2C}$, the optimal separable operation $\Lambda^*_{SEP}$ can be realized as a post-selection process, where with some nonzero success probability $p_{succ}$ the three parties implement an optimal state-dependent SEP operation which provides a fully entangled fraction value larger than $\frac{1}{d}$, otherwise with a probability of failure, $1-p_{succ}$ Alice and Charlie replace their existing state by a product state.}
\vskip 0.3cm
Hence, without any ambiguity, we can claim that $F^*_{SEP}(\rho_{AB_1},\sigma_{B_2C})$ cannot go below the classical bound. This means if there does not exist any trace nonincreasing separable map giving a fully entangled fraction value more than $\frac{1}{d}$, then we can discard the existing state by a product state and making $p_{succ} =0$, which saturates the classical upper bound. So we can write the expression of $F^*_{SEP}(\rho_{AB_1},\sigma_{B_2C})$ as  

    \begin{align}
      F^*_{SEP}(\rho_{AB_1},\sigma_{B_2C})
     &= \max_{\ket{\Phi}\in MES} ~\max_{\tilde{\Lambda}_{SEP}}~~\Big( \dfrac{1-p_{succ}}{d} \nonumber \\
     &+  \bra{\Phi_{AC}}Tr_{B}\left(\tilde{\Lambda}_{SEP}(\eta_{AB_1B_2C})\right)\ket{\Phi_{AC}} \Big), \label{optp} 
\end{align}

where $\tilde{\Lambda}_{SEP}$ is a trace non-increasing SEP operation and we define the success probability $p_{succ}$ as 
\begin{align}
p_{succ} = \tr \left[ \tilde{\Lambda}_{SEP}(\eta_{AB_1B_2C})\right], \label{psucc}
\end{align}
where $\tilde{\Lambda}_{SEP}$ has an operator-sum representation, 
\begin{align}
    \tilde{\Lambda}_{SEP}(\cdot) = \sum_{i=1}^K (m^A_i \otimes n^{B}_i \otimes o_i^C )~(\cdot) ~(m^A_i \otimes n^{B}_i \otimes o_i^C )^{\dagger} \label{tnsep}
\end{align}
such that 
\[
\sum_{i=1}^K (m^A_i)^{\dagger}~m^A_i \otimes (n^B_i)^{\dagger}~n^B_i \otimes (o^C_i)^{\dagger}~o^C_i \leq \mathbb{I}_{AB_1B_2C}. 
\]
Thus, the expression of the fully entangled fraction can be written as
\begin{widetext}
\begin{align}
      F^*_{SEP}(\rho_{AB_1},\sigma_{B_2C})
     &= \max_{\ket{\Phi}} \max_{\{m_i^A,n_i^B,o_i^C\}} \Big(\dfrac{1-p_{succ}}{d} + \bra{\Phi_{AC}}\tr_{B}\Big[\sum_{i=1}^K (m^A_i \otimes n^{B}_i \otimes o_i^C )(\eta_{AB_1B_2C}) ~(m^A_i \otimes n^{B}_i \otimes o_i^C )^{\dagger}\Big]\ket{\Phi_{AC}}\Big),\nonumber\\
      &= \max_{U} \max_{\{m_i^A,n_i^B,o_i^C\}} \Big(\dfrac{1-p_{succ}}{d} \nonumber\\
      &\qquad\qquad\qquad\qquad+ \bra{\Phi_{0}}(U^{\dagger}\otimes \mathbb{I})\tr_{B}\Big[\sum_{i=1}^K (m^A_i \otimes n^{B}_i \otimes o_i^C )(\eta_{AB_1B_2C}) ~(m^A_i \otimes n^{B}_i \otimes o_i^C )^{\dagger}\Big](U\otimes \mathbb{I})\ket{\Phi_{0}}\Big),\nonumber\\
      &= \max_{\{m_i^A,n_i^B,o_i^C\}} \Big(\dfrac{1-p_{succ}}{d} + \bra{\Phi_{0}}\tr_{B}\Big[\sum_{i=1}^K (m^A_i \otimes n^{B}_i \otimes o_i^C )(\eta_{AB_1B_2C}) ~(m^A_i \otimes n^{B}_i \otimes o_i^C )^{\dagger}\Big]\ket{\Phi_{0}}\Big),\nonumber\\
      &= \max_{\{m_i^A,n_i^B,o_i^C\}}
\Big(\frac{1-p_{succ}}{2}
+ \sum_{i=1}^K \operatorname{Tr} \Big[
\big( P_{\Phi_0}^{A C} \otimes \mathbb{1}_B \big)\big( m^A_i \otimes n^{B}_i \otimes o_i^C  \big)\,
(\eta_{AB_1B_2C}) \,
\big( m^A_i \otimes n^{B}_i \otimes o_i^C  \big)^\dagger
\Big]\Big)\label{optp} 
\end{align}
Noting that 
\[
p_{succ} = \sum_{i=1}^k \operatorname{Tr} \Big[
\big( (m_i^A)^{\dagger}m_i^A \otimes (n_i^B)^{\dagger}n_i^B \otimes (o_i^C)^{\dagger}o_i^C \big) (\eta_{AB_1B_2C})
\Big]
\]
We can write
\begin{align}
    F^*_{SEP}(\rho_{AB_1},\sigma_{B_2C})= \frac{1}{2}
+ \max_{\{m_i^A,n_i^B,o_i^C\}}~\operatorname{Tr} \Big[
\Delta \sum_{i=1}^K
\big( m^A_i \otimes n^{B}_i \otimes o_i^C  \big)\,
(\eta_{AB_1B_2C}) \,
\big( m^A_i \otimes n^{B}_i \otimes o_i^C  \big)^\dagger
\Big]
\end{align}
where 
\begin{align}
\Delta=\Big( P_{\Phi_0}^{A C} \otimes \mathbb{1}^{B_1B_2} - \tfrac{\mathbb{1}_{AB_1B_2C}}{2} \Big)
\end{align}
Let
\begin{align*}
\mathcal{E}(\eta_{AB_1B_2C})
&= \sum_{i=1}^K K_i \, (\rho_{AB_1} \otimes \sigma_{B_2C}) \, K_i^\dagger,
\quad \text{with } K_i = m^A_i \otimes n^{B}_i \otimes o_i^C, \\
&\text{s.t.} \quad \operatorname{Tr}_{B} \big[ \mathcal{E}(\omega) \big] \leq \mathbb{1}_{AC}.
\end{align*}
Then,
\begin{align*}
F^*_{SEP}(\rho_{AB_1},\sigma_{B_2C})
= \max_{\mathcal{E}} \Big( \frac{1}{2}
+ \operatorname{Tr} \Big[ \Delta \mathcal{E}(\eta_{AB_1B_2C}) \Big]\Big).
\end{align*}
Using Choi-Jamiolkowski isomorphism, we can write \cite{nielsen2010quantum}
\begin{align}
\mathcal{E}(\eta_{AB_1B_2C})
&= \operatorname{Tr}_{AB_1B_2C} \Big[
\big( \mathbb{1}_{A'B'_1B'_2C'} \otimes \eta_{AB_1B_2C}^{T} \big) J(\mathcal{E})
\Big],
\end{align}
where $J(\mathcal{E}$ is the Choi matrix for the given operation $\mathcal{E}.$

Thus, we can write the expression for the fully entangled fraction as the following convex optimization problem
\begin{align}
F^*_{SEP}(\rho_{AB_1},\sigma_{B_2C})
\leq& \max_{J(\mathcal{E})}  \Big(\frac{1}{2}
+ \operatorname{Tr} \Big[ \Delta  \operatorname{Tr}_{AB_1B_2C} \Big[
\big( \mathbb{1}_{A'B'_1B'_2C'} \otimes \eta_{AB_1B_2C}^{T} \big) J(\mathcal{E})
\Big] \Big]\Big)\nonumber\\
=& \max_{J(\mathcal{E})}  \Big(\frac{1}{2}
+ \operatorname{Tr} \Big[    
\big( \Delta \otimes \eta_{AB_1B_2C}^{T} \big) J(\mathcal{E})
 \Big]\Big),\label{SDPgen}
\end{align}
where $J(\mathcal{E})\in\mathcal{L}(\mathcal{H}_{A'}\otimes\mathcal{H}_{B'_1}\otimes\mathcal{H}_{B'_2}\otimes\mathcal{H}_{C'}\otimes\mathcal{H}_{A}\otimes\mathcal{H}_{B_1}\otimes\mathcal{H}_{B_2}\otimes\mathcal{H}_{C})$ under the constraints
\begin{align*}
    J(\mathcal{E})&\geq0,\\
    J(\mathcal{E})^{\Gamma_{B'_1B'_2B_1B_2}}&\geq0,\\
    \tr_{AB_1B_2C}[J(\mathcal{E})]&\leq\mathbb{1}_{A'B'_1B'_2C'}
\end{align*}
\end{widetext}

 Here the constraint $J(\mathcal{E})^{\Gamma_{B'_1B'_2B_1B_2}}\geq0$ implies that $J(\mathcal{E})$ is PPT across the $B'_1B'_2B_1B_2|A'C'AC$ cut. This automatically includes the Choi matrices for SEP operations, as SEP operations are a subset of PPT operations. This gives us the upper bound on the achievable fully entangled fraction in terms of a convex optimization problem. But we observe that finding the optimal fully entangled fraction involves optimization over a variable which is a $d^8 \times d^8$ matrix. Hence, solving this problem in its full generality is not that straightforward. It might be possible to explore some symmetries of the form of the cost function to simplify the problem, but we leave it as an avenue to be explored for future studies.

 However, the objective function of the above convex optimization problem emerges through the expression of the fully entangled fraction defined in Eq.$~$(\ref{optp}). In the expression of Eq.$~$(\ref{optp}), there is an inner product term $\bra{\Phi_{AC}}Tr_{B}\left(\tilde{\Lambda}_{SEP}(\eta_{AB_1B_2C})\right)\ket{\Phi_{AC}}$, where $\tilde{\Lambda}_{SEP}$ denotes any trace non-increasing SEP. Now one could ask whether we can manipulate the local Kraus operators of $\tilde{\Lambda}_{SEP}$ inside such inner product in such a way that it remains invariant, and consequently, we reduce the variable size of our optimization problem. Below we will show that such reduction is not possible in general.

 So let us consider an arbitrary trace non-increasing SEP operation $\tilde{\Lambda}_{SEP}$ given in Eq.$~$(\ref{tnsep}).  We now manipulate the inner product as follows:

\begin{widetext}
\begin{small}
\begin{align}    
      &\bra{\Phi_{AC}}Tr_{B}\left(\tilde{\Lambda}_{SEP}(\eta_{AB_1B_2C})\right)\ket{\Phi_{AC}} = \max_{\ket{\Phi}} \max_{\lbrace m_i^A,n_i^B,o_i^C \rbrace} \bra{\Phi_{AC}} Tr_{B}\left(\sum_{i}  (m^A_i \otimes n^B_i \otimes o^C_i)~\eta_{AB_1B_2C} ~(m^A_i\otimes n^B_i \otimes o^C_i)^{\dagger} \right)\ket{\Phi_{AC}}    \nonumber  \\
    &=  \max_{U} \max_{\lbrace m_i^A,~n_i^B,~o_i^C \rbrace} \bra{\Phi_0} (U^{\dagger}\otimes \mathbb{I}) \left(\sum_{i}  (m^A_i \otimes o^C_i) ~ Tr_{B}\left( (\mathbb{I} \otimes (n^B_i)^{\dagger}n^B_i \otimes \mathbb{I})~\eta_{AB_1B_2C}\right) (m^A_i\otimes o^C_i)^{\dagger} \right) (U\otimes \mathbb{I})\ket{\Phi_{0}} \nonumber \\
    &= \max_{\lbrace m_i^A,~n_i^B,~o_i^C \rbrace} \sum_{i} \bra{\Phi_0}  (m^A_i \otimes o_i^C) ~ Tr_{B}\left( (\mathbb{I} \otimes (n^B_i)^{\dagger}n^B_i \otimes \mathbb{I})~\eta_{AB_1B_2C}\right) (m^A_i\otimes o^C_i)^{\dagger}) \ket{\Phi_0},  \label{9} \\
    &= \max_{\lbrace m_i^A,~N_i^B,~o_i^C \rbrace } \sum_i~ \bra{\Phi_0}  ((o_i^C)^T~m_i^A \otimes\mathbb{I}) ~\tilde{\chi}_{AC}(N^B_i)~ ((m_i^A)^{\dagger}~(o_i^C)^*\otimes \mathbb{I}) \ket{\Phi_0}, \quad \quad \quad \quad \quad \quad \text{where} \quad N_i^B=(n_i^B)^{\dagger}n_i^B,  \label{10} \\
    &= \max_{\lbrace m_i^A,~N_i^B\rbrace}~~ \bra{\Phi_0}\left(\sum_i  ~ (m^A_i \otimes \mathbb{I}) ~\tilde{\chi}_{AC}(N_i) ~ ((m^A_i)^{\dagger} \otimes \mathbb{I}) \right) \ket{\Phi_0},  \label{11} \\
    &= \max_{\lbrace m_i^A,~N_i^B\rbrace} ~\tr\left( \sum_i  ~(m^A_i\otimes \mathbb{I})~\tilde{\chi}_{AC}(N_i)~((m^A_i)^{\dagger}\otimes \mathbb{I})~P^{AC}_{\Phi_0}\right) \nonumber \\ 
    &= \max_{\lbrace m_i^A,~N_i^B\rbrace} \sum_i~\tr\left( \tilde{\chi}_{AC}(N^B_i)~X^{AC}_i(m_i^A)\right), \label{12}
\end{align}
\end{small}
\end{widetext}

where $P^{AC}_{\Phi_0}=\ket{\Phi_0}\bra{\Phi_0}$ is the projector, where $\sqrt{d}\ket{\Phi_0}=\sum_{i=0}^{d-1}\ket{ii}$ is the MES in $\mathbb{C}^d \otimes \mathbb{C}^d$. In second and fifth line the product unitary operator $U\otimes \mathbb{I}$ and the Kraus operator of Charlie $(o_i^C)^T \otimes \mathbb{I}$ are absorbed in every $m^A_i \otimes \mathbb{I}$ respectively, since there is already a maximization over all such set of operators. In Eq.$~$(\ref{10}) we used the following property, 
\begin{align}
    (\mathbb{I}\otimes Y) \ket{\Phi_0} = (Y^T \otimes \mathbb{I})\ket{\Phi_0}, \nonumber
\end{align}
where $Y$ is an arbitrary complex matrix. Also in Eq.$~$(\ref{12}) the operator $X^{AC}_i(m_i^A)$ denotes an un-normalized rank one operator defined as
\begin{align}
    X^{AC}_i(m_i^A)&= ((m^A_i)^{\dagger}\otimes \mathbb{I})~P_{\Phi_0} ~(m^A_i\otimes \mathbb{I}),
\end{align}
and the operator $\tilde{\chi}_{AC}(N_i^B)$ in Eq.$~$(\ref{10}) is defined as  
\begin{align}
    \tilde{\chi}_{AC}(N_i^B)= Tr_{B}\left( (\mathbb{I} \otimes (n^B_i)^{\dagger}n^B_i \otimes \mathbb{I})~\eta_{AB_1B_2C}\right), \label{13}
\end{align}
which is un-normalized density matrix.
\vskip 0.3cm
Therefore, from Eq.$~$(\ref{11}) it is clear that if we start with any trace non-increasing SEP operation, and manipulate the local Kraus operators, then the inner product remains invariant while Charlie performs a trivial operation and we only maximize over Alice's modified Kraus operator $\lbrace (o_i^C)^T~m_i^A\rbrace$, and over Bob's Kraus operator $\lbrace n_i^B \rbrace$. However, if we again consider the objective function given in Eq.$~$(\ref{optp}), then it can happen that the modified global operation does not remain physical after manipulating the local Kraus operators of $\tilde{\Lambda}_{SEP}$ as done in Eq.$~$(\ref{11}). 

Therefore, if $o_i^C$ are not unitary, then the following constraint inequality 
\begin{small}
\begin{align}
  & (i) ~~ \sum_i (m_i^A)^{\dagger}m_i^A \otimes (n_i^B)^{\dagger}(n_i^B) \otimes (o_i^C)^{\dagger}o_i^C \leq \mathbb{I}_{AB_1B_2C} \nonumber \\
  &(ii) ~~ \sum_i ((o_i^C)^T~m_i^A)^{\dagger}((o_i^C)^T~m_i^A) \otimes (n_i^B)^{\dagger}(n_i^B)\otimes\mathbb{I}_C \leq\mathbb{I}_{AB_1B_2C} \nonumber\\
\end{align}
\end{small}
does not satisfy simultaneously in general. For instance, consider for $d=2$
\begin{align}
    &m_1^A=\ket{0}\bra{0},~m_2^A=\ket{0}\bra{0}+\ket{0}\bra{1},\nonumber\\
    &n_1^B=n_2^B=\mathbb{I}_{4\times4},\nonumber\\
    &o_1^C=\frac{1}{\sqrt{2}}\ket{0}\bra{0},~o_2^C=\frac{1}{\sqrt{2}}\ket{0}\bra{1}.\nonumber
\end{align}
In this case, it can be easily verified that
\begin{align}
    \sum_{i=1}^2 (m^A_i)^{\dagger}~m^A_i \otimes (n^B_i)^{\dagger}~n^B_i \otimes (o^C_i)^{\dagger}~o^C_i < \mathbb{I}_{AB_1B_2C},
\end{align}
is satisfied and therefore, such Kraus operators represent a valid trace non-increasing SEP operation. However, the modified Kraus operators {\it i.e.,} $((o_i^C)^T~m_i^A)^{\dagger}((o_i^C)^T~m_i^A) \otimes (n_i^B)^{\dagger}(n_i^B)\otimes\mathbb{I}_C$ do not represent a trace non-increasing SEP as the following
\begin{align}
    \sum_i ((o_i^C)^T~m_i^A)^{\dagger}((o_i^C)^T~m_i^A) \otimes (n_i^B)^{\dagger}(n_i^B)\otimes\mathbb{I}_C >\mathbb{I}_{AB_1B_2C}. 
\end{align}
is true.

 \subsection{A restricted class of separable operations $(\Lambda_{rSEP} \subset \Lambda_{SEP})$}
 Let us restrict ourselves to the class of SEP operations in Eq.$~$(\ref{optp}), where in the trace non-increasing part we impose that all the Kraus operators of Charlie should be unitary operators \textit{i.e} $o_i^C$ satisfy the property $(o_{i}^C)^{\dagger}o_i^C =\mathbb{I}_C$~ $\forall i$. We denote the trace non-increasing part of the restricted SEP as $\tilde{\Lambda}_{rSEP}$ such that the deterministic version can be expressed as 
 \begin{align}
     \Lambda_{rSEP}(\rho) &= \sum_i (m_i^A \otimes n_i^B \otimes o_i^C) ~\rho ~(m_i^A \otimes n_i^B \otimes o_i^C)^{\dagger}    \nonumber \\
      & + (1- p_{succ}) ~\ket{\psi~\phi}_{AC}\bra{\psi~\phi}, \label{rSEP}
 \end{align}
  where $\rho = \rho_{AB_1} \otimes \sigma_{B_2C}$ and  $p_{succ} = Tr[\sum_i (m_i^A \otimes n_i^B \otimes o_i^C) ~\rho ~(m_i^A \otimes n_i^B \otimes o_i^C)^{\dagger}]$ denotes the probability of success of implementing $\tilde{\Lambda}_{rSEP}$. In the above expression of  $ \Lambda_{rSEP}(\rho) $, the product state $\ket{\psi~\phi}_{AC}\bra{\psi~\phi}$ denotes some post-selected state by Alice and Charlie for which the fully entangled fraction saturates the classical upper bound.  
 \vskip 0.3cm 
 Note here that the fully entangled fraction achievable through $\Lambda_{rSEP}$ cannot give a less value than $\frac{1}{d}$. Later we will show that such restriction reduces the complexity of the convex optimization problem introduced in Eq.$~$(\ref{SDPgen}). The optimal fully entangled fraction under such a restricted class of separable operations $\tilde{\Lambda}_{rSEP} \subset SEP$ can be described as
\begin{align}
     &F^*_{rSEP}(\rho_{AB_1},\sigma_{B_2C})= \nonumber \\ 
     &\dfrac{1-p_{succ}}{d}+  
     \bra{\Phi_0} Tr_{B}[\tilde{\Lambda}_{rSEP}(\eta_{AB_1B_2 C})]\ket{\Phi_0}
    \label{sufficient} \\
    s.t. \nonumber \\
     &\tilde{\Lambda}_{rSEP}= \lbrace \sum_i K_i^{\dagger}K_i = \sum_i (m^A_i)^{\dagger}m^A_i \otimes (n^B_i)^{\dagger}n^B_i \otimes \mathbb{I}_C \leq \mathbb{I} \rbrace
\end{align}
  Here $F^*_{rSEP}(\rho_{AB_1},\sigma_{B_2C})$ in Eq.$~$(\ref{sufficient}) represents the optimal teleportation fidelity achievable under the restricted SEP protocols (where Charlie's Kraus operators are unitary operators).  It is important to point out that $\tilde{\Lambda}_{rSEP}$ cannot realize any three-party LOCC, as we are imposing restrictions on Charlie's Kraus operators. We denote the class of LOCC as $\Lambda_{rLOCC}\subset LOCC$ that can be realized through $\Lambda_{rSEP}$. As $\Lambda_{rLOCC}$ is a subset of $\Lambda_{rSEP}$, the operator-sum representation of $\Lambda_{rLOCC}$ can be expressed exactly as Eq.$~$(\ref{rSEP}). The following types of LOCC can realize $\Lambda_{rLOCC}$.

\begin{small}
\begin{align}
   & (1) ~\Lambda_{B \rightarrow AC} \implies One~way~LOCC ~from~Bob ~to ~Alice-Charlie 
   \nonumber \\
   & (2)~ \Lambda_{A\rightarrow BC}\implies One~way~LOCC ~from~Alice ~to ~Bob-Charlie
   \nonumber \\
   & (4) ~ \Lambda_{A\leftrightarrow B\rightarrow C} \implies Both~way~LOCC~only ~between~Alice-Bob \label{loccstructure}
\end{align}
\end{small}

 where the arrow $\rightarrow$ indicates one-way classical communication from sender to receiver and $\leftrightarrow$ indicates both way communications. Therefore, $\Lambda_{B \rightarrow AC}$ operations represents LOCC protocols initiated by Bob, and $\Lambda_{A\leftrightarrow B \rightarrow C}$ represents the protocol initiated by either Alice or Bob.

\subsubsection{Exploring the optimization problem under restricted SEP operation} \label{secB}
For any given states $\rho_{AB_1}$ and $\sigma_{B_2C}$, one can propose an optimization problem under this restricted three-party SEP operation to obtain $F^*_{rSEP}(\rho_{AB_1},\sigma_{B_2C})$ ( see Eq.$~$(\ref{sufficient})) as 
\vspace{-5mm}
\begin{align}
     &(i) \quad \max_{\lbrace N_i,m_i\rbrace}~  \Big(\dfrac{1-p_{succ}}{d}+\sum_{i=1}^K \tr\left[X^{AC}_i(m_i^A)~\tilde{\chi}_{AC}(N^B_i)\right]\Big) \nonumber  \\
      &s.t. \quad  p_{succ} = \sum_{i=1}^K \tr[(M_i^A \otimes N_i^B\otimes \mathbb{I}_C)~\rho_{AB_1} \otimes \sigma_{B_2C}]  \nonumber \\
      & and\quad  \sum_{i=1}^K M_i^A \otimes N_i^B\leq  \mathbb{I}_{d^3},  \nonumber \\
      & \quad \quad ~~  M^A_i= (m^A_i)^{\dagger}m^A_i \geq \mathbf{0},~~N^B_i=(n^B_i)^{\dagger}n^B_i\geq \mathbf{0} \quad \forall i\nonumber \\ 
     &  \quad \quad ~ X^{AC}_i(m_i^A)=((m^A_i)^{\dagger}\otimes \mathbb{I}_d)\ket{\Phi_0}\bra{\Phi_0}(m^A_i \otimes \mathbb{I}_d) \quad \forall m^A_i \label{primalsep}
\end{align}
or, one can equivalently pose a slightly modified version of the same optimization problem as
\begin{align}
    &  \max_{\lbrace m_i, ~N_{i}\rbrace} ~~ \dfrac{1-p_{succ}}{d}+ \sum_{i=1}^K \tr\left[ (X^{AC}_i(m_i^A) \otimes N^B_{i})~\rho_{AB_1} \otimes \sigma_{B_2C}\right] \nonumber \\
    \nonumber \\
    &s.t. \nonumber\\
    & p_{succ} = d\sum_{i=1}^K \tr\left[ (Tr_C[X^{AC}_i(m_i^A)] \otimes N_i^B\otimes \mathbb{I}_C)~\rho_{AB_1} \otimes \sigma_{B_2C}\right] \nonumber \\
    &and \quad  X_i^{AC}(m_i^A)= (m_i^{\dagger} \otimes \mathbb{I}_C)P_{\Phi_0}(m^A_i \otimes \mathbb{I}_C) , \nonumber \\
    & \quad M^A_i=(m^A_i)^{\dagger}m^A_i \geq \mathbf{0} ,~~N^B_{i}\geq \mathbf{0}, \nonumber \\
    & \quad  \sum_i M_i^A \otimes N_i^B\otimes \mathbb{I}_C \leq \mathbb{I}_{d^4}, \label{primalalter}
\end{align}
where the following inequality condition
\[
\sum_i X^{AC}_i(m_i^A) \otimes N_i^B \leq  \sum_i M_i^A \otimes N_i^B\otimes \mathbb{I}_C
\]
holds in general, and the reason can be shown as follows, 

\begin{align}
    & M_i^A \otimes\mathbb{I}_C \otimes N_i^B= (m^A_i)^{\dagger}~m^A_i \otimes  \mathbb{I}_C \otimes N_i^B \nonumber \\
    & = ((m^A_i)^{\dagger} \otimes \mathbb{I}_C) \mathbb{I}_{AC}(m^A_i \otimes \mathbb{I}_C) \otimes N_i^B\nonumber \\
    & = X^{AC}_i(m_i^A) \otimes N_i^B  \nonumber \\
    &+ \sum_{j=1}^{d^2 -1}((m^A_i)^{\dagger} \otimes \mathbb{I}_C) \ket{\Phi_j} \bra{\Phi_j}(m^A_i \otimes \mathbb{I}_C) \otimes N_i^B\nonumber \\
    & \geq X^{AC}_i(m_i^A) \otimes N^B_i,
\end{align}
where the above inequality is strict as long as $M_i^A,~N_i^B$ are positive operators and we have assumed that the basis set, $\lbrace \ket{\Phi_j}\rbrace _{j=0}^{d^2 -1}$ forms a complete set such that every $\ket{\Phi_j}$ is a maximally entangled state. Therefore, every operator \[((m^A_i)^{\dagger} \otimes \mathbb{I}_C) \ket{\Phi_j} \bra{\Phi_j}(m^A_i \otimes \mathbb{I}_C)\] is a rank one positive operator.

which proves the existence of the inequality constraint in Eq.$~$(\ref{primalalter}). However, it is important to notice that the primal optimization problem in Eq.$~$(\ref{primalsep}) or Eq.$~$(\ref{primalalter}) is not a convex optimization problem. This is because the objective function of problem \hyperref[primalsep]{$(i)$} is not linearly dependent on the variables $\lbrace m_i^A \rbrace_{i=1}^K$ and equivalently the constraints of the modified problem in Eq.$~$(\ref{primalalter}) are not convex.   

We refer the reader to Appendix \ref{nonconvexity} for more details about he non-convexity of the optimization problem in \hyperref[primalsep]{$(i)$}. We proceed by convexifying it in the following sections so that we can obtain a tractable bound on the fully entangled fraction. 

\subsection{Obtaining bound on $F^*_{rSEP}(\rho_{AB_1},\sigma_{B_2C})$ via convex relaxation}
The importance of the convex relaxation method is to convert the non-convex optimization mentioned in problem~(i) in Eq.~(\ref{primalsep}) to a convex optimization problem. Here, relaxation means that we convert the actual problem to a convex optimization while we relax the set of variables given in Eq.~(\ref{primalsep}) to a larger set of variables that forms a convex set. The standard definition of a convex optimization problem can be given as follows \cite{boyd2004convex}:
\begin{align}
    &\max_{\mathbf{x}} \quad f(\mathbf{x}) \nonumber \\
    &s.t. \quad g_i(\mathbf{x}) \geq 0 \quad \forall i\in \lbrace 0, m-1 \rbrace \nonumber \\
    & \quad \quad ~~ h_j(\mathbf{x}) = 0 \quad \forall j \in \lbrace 0, n-1 \rbrace, 
\end{align}
where $\mathbf{x}$ denotes a finite set of variables, and the objective function $f(\mathbf{x})$ and every $g_i(\mathbf{x})$ should be a convex function of $\mathbf{x}$. Also every $h_j(\mathbf{x})$ should represent an affine transformation of the set $\mathbf{x}$ \cite{boyd2004convex}. 

In our case, the set of optimization variables for problem $(i)$ are in a tensor product form i.e., $\lbrace X_i^{AC} \otimes N_i^B\rbrace_{i=1}^K$ or in separable form across Alice-Charlie vs. Bob bipartition. In the convex relaxation, we enlarge the such set of operators to another set of operators, $\lbrace \mathcal{M}_i \rbrace_{i=1}^K$ which are separable across Alice-Charlie vs. Bob bipartition. We thus propose the following convex optimization problem, which provides an upper bound on the feasible solutions of problem $(i)$: 
\begin{align}
  (ii) \quad   &F_S^*(\rho_{AB_1},\sigma_{B_2C}) = \max_{\lbrace \mathcal{M}_i\rbrace}~~ \dfrac{1}{d} \nonumber\\
  &\qquad \qquad \qquad ~~~~~~- \sum_{i=1}^K \tr[(Tr_C[\mathcal{M}_i] \otimes \mathbb{I}_C)~\rho \otimes \sigma]\nonumber \\
  & \qquad \qquad \qquad ~~~~~~+\sum_{i=1}^K \tr\left[ \mathcal{M}_i ~\rho \otimes \sigma\right]
  \end{align}
  \begin{align}
    s.t. \quad & \sum_{i=1}^K d~Tr_C[\mathcal{M}_i] \otimes \mathbb{I}_C \leq \mathbb{I}_{AB_1B_2C} \nonumber \\
    & \mathbf{0} \leq \mathcal{M}_i \leq  \mathbb{I}_{AB_1B_2C} \nonumber \\
    & \mathcal{M}_i \in \mathcal{S}_{sep}(\mathcal{H}_{AC} \otimes \mathcal{H}_{B_1B_2}),  \label{sepconv}
\end{align}
where $\mathcal{S}_{sep}(\mathcal{H}_{AC} \otimes \mathcal{H}_{B_1B_2})$ denotes the set of positive Hermitian operators which are separable along the Alice-Charlie vs. Bob cut. Notice that the above optimization can be solved as a convex cone programming and will provide an optimal value depending on the state $\rho_{AB_1} \otimes \sigma_{B_2C}$. So we can claim that the following inequality,
\begin{align}
    F^*_S(\rho_{AB_1},\sigma_{B_2C})  \geq F^*_{rSEP}(\rho_{AB_1},\sigma_{B_2C}),
\end{align}
truly holds. This is because the set of operations $\mathcal{S}_{sep}(\mathcal{H}_{AC} \otimes \mathcal{H}_{B_1B_2})$ forms a super set of the set $\mathcal{S}$ defined as 

\begin{small}
\begin{align}
\mathcal{S} = &\left \lbrace X_i^{AC}\otimes N_i^B~|~X_i^{AC}\geq \mathbf{0},~rank(X_i^{AC})=1, ~N_i^B\geq \mathbf{0},~~ \right. \nonumber \\
 & \left. \sum_{i=1}^K X_i^{AC}\otimes N_i^B\leq \mathbb{I}_{AB_1B_2C} \right \rbrace \nonumber \\
 \subset &\left \lbrace \mathcal{M}_i ~| ~\mathcal{M}_i \geq \mathbf{0},~~\mathcal{M}_i \in \mathcal{S}_{sep}(\mathcal{H}_{AC} \otimes \mathcal{H}_{B_1B_2}),~ \right. \nonumber \\ 
 & \left. \sum_{i=1}^K \mathcal{M}_i \leq \mathbb{I}_{AB_1B_2C} \right \rbrace, \quad \quad \quad 
\end{align}
\end{small}

Thus, now we have an optimization problem in \hyperref[SDPp]{$(ii)$} that is convex. But there is one caveat. When Alice's and Charlie's local dimension $d$ exceeds three, then it is hard to certify whether the variable operator $\mathcal{M}_i$ is separable or entangled across Alice-Charlie vs. Bob's bipartition. In order to tackle this we further relax the variable set of operators $\lbrace \mathcal{M}_i\rbrace$ as mentioned in problem \hyperref[sepconv]{$(ii)$} to the set of PPT operators, where the partial transpose is taken with respect to Bob's two-qudit system. We define such a set as $\mathcal{S}_{PPT}$:
\[
\mathcal{S}_{PPT} (\mathcal{H_{AC}} \otimes \mathcal{H}_{B_1B_2})= \left \lbrace \mathcal{M}_i ~|~ \mathcal{M}_i \geq \mathbf{0},~~\mathcal{M}_i^{T_{B_1B_2}} \geq 0\right \rbrace
\]
where the following hierarchy for the set of operators truly holds, 
\[
\mathcal{S}_{sep} (\mathcal{H_{AC}} \otimes \mathcal{H}_{B_1B_2}) \subset
\mathcal{S}_{PPT} (\mathcal{H_{AC}} \otimes \mathcal{H}_{B_1B_2})  
\]
We now propose another optimization problem below: 
\begin{align}
 (iii) \quad    &F_P^*(\rho,\sigma) = ~~\max_{\lbrace \mathcal{M}_i\rbrace} ~\dfrac{1}{d} - \sum_{i=1}^K \tr[(Tr_C[\mathcal{M}_i] \otimes \mathbb{I}_C)~\rho \otimes \sigma]\nonumber \\
  & +\sum_{i=1}^K \tr\left[ \mathcal{M}_i ~\rho \otimes \sigma\right] \nonumber \\
    s.t. \quad & \sum_{i=1}^K d~Tr_C[\mathcal{M}_i] \otimes \mathbb{I}_C \leq  \mathbb{I}_{AB_1B_2C} \nonumber \\
    & \mathbf{0} \leq \mathcal{M}_i \leq \mathbb{I}_{AB_1B_2C},  \nonumber \\
    &\mathcal{M}_i^{T_{B_1B_2}} \geq \mathbf{0}, \label{SDP} 
\end{align}
where each of the operators $\mathcal{M}_i$ acts jointly on the composite state space of the three-party system and the partial transposition is taken with respect to Bob's entire system. The optimal feasible solutions of the aforementioned optimization problems, i.e. \hyperref[primalsep]{$(i)$}, \hyperref[SDPp]{$(ii)$}, become a strictly feasible solution of the proposed convex optimization problem $(iii)$ satisfying the following hierarchy:
\begin{align}
    F^*_P(\rho_{AB_1},\sigma_{B_2C}) \geq& F^*_S(\rho_{AB_1},\sigma_{B_2C}) \nonumber\\
    &\geq F^*_{rSEP}(\rho_{AB_1},\sigma_{B_2C}) \geq F_r^*(\rho_{AB_1},\sigma_{B_2C}).
\end{align}
Here $F_r^*(\rho_{AB_1},\sigma_{B_2C})$ is the fully entangled fraction achievable through restricted LOCC operation.
\vskip 0.3cm
We thus see that the optimization problem in \hyperref[SDP]{$(iii)$} under the restricted SEP operations is convex, and it involves optimization over $d^4\times d^4$ matrices, which is significantly easier to tackle than the optimization problem in Eqn.\eqref{SDPgen} for general SEP operations, which involves optimizing over matrices of the size $d^8\times d^8$. As mentioned earlier, we could try to utilize the symmetries of the problem to solve the optimization problem in  Eqn.\eqref{SDPgen}, but we leave it as a topic for future studies. In the remainder of this paper, our focus will be on whether there exist certain states for which the gap between the LOCC operations and the SEP operations vanishes from the context of optimizing the fully entangled fraction.

\subsection{General upper bounds on the optimal fully entangled fraction achievable under LOCC }
   Up to this point, we have formulated a convex optimization problem given in Eq.$~$(\ref{SDPgen}) under three-party SEP. Because of its complexity, we have chosen a restricted class of SEP and formulated another simpler convex optimization problem that could provide upper bounds on the fully entangled fraction achievable through a subclass of LOCC that we denote as $\Lambda_{rLOCC}$ (discussed in Eq.$~$(\ref{loccstructure})). The proposition given below shows a generic upper bounds on the fully entangled fraction that can be achievable through any three-party LOCC {\it i.e.} $F^*(\rho_{AB_1}, \sigma_{B_2C})$, and the generic upper bounds that can be achievable through $\Lambda_{rLOCC}$ {\it i.e.} $F^*_r(\rho_{AB_1}, \sigma_{B_2C})$.   
  
 \begin{proposition}
    In a three-party linear network if Alice-Bob share an arbitrary bipartite state $\rho_{AB_1}$ with local dimension $d\geq 2$ and Bob-Charlie share another bipartite state $\sigma_{B_2C}$ with local dimension $d\geq 2$, then the optimal fully entangled fraction $F^*(\rho_{AB_1}, \sigma_{B_2C})$ achievable through any LOCC, and the optimal fully entangled fraction $F^*_r(\rho_{AB_1},\sigma_{B_2C})$ that can be established under a subclass of LOCC defined in Eq.$~$(\ref{loccstructure}) has the following upper bounds:  \label{p1}
    \begin{align}
    &F^*(\rho_{AB_1},\sigma_{B_2C}) \leq \min \left \lbrace F_1^*, ~F_2^*,~F^*_{SEP}(\rho_{AB_1},\sigma_{B_2C})\right \rbrace  \\
    &F_r^*(\rho_{AB_1},\sigma_{B_2C}) \leq \min \left \lbrace F^*(\rho_{AB_1},\sigma_{B_2C}),~F_P^*(\rho_{AB_1},\sigma_{B_2C})\right \rbrace  \\
     & where, \nonumber \\
      &F_1^* = \min \left \lbrace F^*(\rho_{AB_1}),~ F^*(\sigma_{B_2C})\right \rbrace \label{F_1}
      \end{align}
      \begin{align}
    &F_2^* = \max_{LOCC} \dfrac{1}{d} \left(1+ \sqrt{\frac{d~(d-1)}{2}}~C(\rho_{\Lambda,~AC}) \right), \label{F_2} \\
    & \leq \dfrac{1}{d}+ \sqrt{\dfrac{d(d-1)}{2~d^2}}~min \lbrace C(\rho_{AB_1}),~C(\sigma_{B_2C}) \rbrace,  \nonumber \\
 \label{genupp}\nonumber
\end{align}
\end{proposition}

\begin{proof}
    The fact that $F^*(\rho_{AB_1},\sigma_{B_2C}) \leq F_1^*$ can be proposed directly from LOCC argument. If we only consider the bipartition of Alice vs. Bob-Charlie (Bob and Charlie treated as a single lab), then one can argue that the optimal fully entangled fraction in Alice vs. Bob-Charlie bipartition cannot surpass the value $F^*(\rho_{AB_1})$ that is the optimal value achievable by any two-party LOCC using $\rho_{AB_1}$. The same argument follows for Charlie vs. Alice-Bob bipartition. Therefore, it is also true that $F^*(\rho_{AB_1},\sigma_{B_2C})$ cannot surpass $F_1^*$. 
    \vskip 0.2cm 
    The fact that $F^*(\rho_{AB_1},\sigma_{B_2C}) \leq F_2^*$ holds because $F_2^*$ is a LOCC monotone and its functional form has been derived in Ref. \cite{Zhao_2010}. It has been shown that $F_2^*$ is a function of the concurrence in $d \otimes d$, and for a $d\otimes d$ bipartite state $\rho$, the lower bound on concurrence $C(\rho)$ is given in Ref.~ \cite{PhysRevLett.95.040504}
    Note that $F_2^*$ explicitly depends on the maximum achievable concurrence over all three-party LOCC. However, it is hard to find a closed form expression for $F_2^*$ as the maximum concurrence under LOCC has no known tighter upper bound for $d>2$, but we can propose a trivial upper bound as  
    \[
     \max_{LOCC} ~C(\rho_{\Lambda,AC}) \leq min \lbrace C(\rho_{AB_1}),~C(\sigma_{B_2C}) \rbrace
    \]

    The fact that $F^*(\rho_{AB_1},\sigma_{B_2C}) \leq F^*_{SEP}(\rho_{AB_1}, \sigma_{B_2C})$ holds because $F^*_{SEP}(\rho_{AB_1}, \sigma_{B_2C})$ is the optimal fully entangled fraction under any three-party SEP ( given in Eq.$~$(\ref{SDP})), which is by theory a superset of LOCC. This proves that $F^*F^*_{SEP}(\rho_{AB_1}, \sigma_{B_2C})$ is upper bounded by the minimum of $F_1^*$, $F_2^*$ and $F^*_{SEP}(\rho_{AB_1}, \sigma_{B_2C})$. 
\vskip 0.2cm
    Now, $F^*_{r}(\rho_{AB_1}, \sigma_{B_2C})$ is the optimal fully entangled fraction under a subclass of LOCC $(\Lambda_{rLOCC})$, and therefore we cannot claim that $F^*_{r}(\rho_{AB_1}, \sigma_{B_2C})$ cannot be increased by any LOCC other than $\Lambda_{rLOCC}$ {\it i.e.} $F^*_{r}(\rho_{AB_1}, \sigma_{B_2C})$ is not a LOCC monotone in general. Thus by definition, the inequality $F^*_{r}(\rho_{AB_1}, \sigma_{B_2C}) \leq F^*(\rho_{AB_1}, \sigma_{B_2C})$ must hold. Also,  $F^*_{r}(\rho_{AB_1}, \sigma_{B_2C})$ cannot surpass the quantity  $F^*_{P}(\rho_{AB_1}, \sigma_{B_2C})$, which is the optimal feasible solution of the convex optimization $(iii)$ given in Eq.$~$(\ref{SDP}). This is because $\Lambda_{rLOCC}$ is a subset of the variable set (the restricted SEP $\Lambda_{rSEP}$) of such convex optimization problem.  Therefore, $F^*_{r}(\rho_{AB_1}, \sigma_{B_2C})$ cannot surpass the minimum of $F^*(\rho_{AB_1}, \sigma_{B_2C})$ and $F^*_{P}(\rho_{AB_1}, \sigma_{B_2C})$. This proves the proposition. 

\end{proof}

\subsubsection{Simplified upper bound for $d=2$}
\begin{proposition} 
    If both $\rho_{AB_1}$ and $\sigma_{B_2C}$ are two-qubit states, then under all three-party LOCC, the optimal {\it fully entangled fraction} between Alice-Charlie {\it i.e.,} $F^*(\rho_{AB_1},~\sigma_{B_2C})$ is upper bounded as,  
    \begin{align}
           F^*(\rho_{AB_1},~\sigma_{B_2C})  \leq \min \lbrace F_1^*,~F_2^*,~F^*_{SEP}(\rho_{AB_1},~\sigma_{B_2C})\rbrace, \nonumber \\
        \text{where} \quad  F_{2}^{*} = \dfrac{1}{2}\left( 1+C(\rho_{AB_1})~C(\sigma_{B_2C})\right), \label{loccup2}   
    \end{align}
\end{proposition}
\begin{proof}
    The quantity $F_2^*$ in Eq.$~$(\ref{F_2}) can be modified to a more simple expression if both $\rho_{AB_1}$ and $\sigma_{B_2C}$ are two-qubit states. Any two-qubit state $\rho$ can have an unique convex decomposition (not necessarily the spectral decomposition) \cite{wootters1998entanglement} as 
\begin{align}
    \rho = \sum_{i=0}^3 r_i ~\ket{\phi_i}\bra{\phi_i}, \quad \quad r_i \geq 0, \quad \sum_i r_i =1, 
\end{align}
where the concurrence of $\rho$ satisfies the property, $C(\rho)=C(\ket{\phi_i})$ for all $i$. So let $\rho_{AC}$ be a two-qubit state prepared between Alice and Charlie after a deterministic LOCC. Let $\lbrace q_{i},~\ket{\phi_{i}}\rbrace$ be the optimal decomposition of $\rho_{AC}$. So without any loss of generality, we can write the optimal fully entangled fraction of the two-qubit state $\rho_{AC}$ as
\begin{align}
    F^*(\rho_{AC})&\leq \sum_{i=0}^3 q_{i}~F^*(\ket{\phi_{i}}) \nonumber \\
    &= \sum_{i=0}^3 q_{i}~\dfrac{1+C(\ket{\phi_{i}})}{2} \nonumber \\
    &= \dfrac{1+C(\rho_{AC})}{2}\leq \dfrac{1+C(\rho_{AB_1})~C(\sigma_{B_2C})}{2} =F_2^*\nonumber 
\end{align}
where we applied the condition that the optimal fully entangled fraction of any two-qubit pure state is upper bounded by a linear function of concurrence of the state, and the upper bound on the concurrence $C(\rho_{AC})$ under LOCC can be written as $C(\rho_{AB_1})~C(\sigma_{B_2C})$ from Ref. \cite{PhysRevLett.93.260501}. 

\end{proof}

\subsection{Some feasible solutions of the convex optimization problem given in Eq.$~$(\ref{SDPgen}) which are LOCC achievable}

 Up to this point, we have proposed a general convex optimization problems given in Eq.$~$(\ref{SDPgen}) that theoretically provides upper bounds on the LOCC achievable optimal fully entangled fraction $F^*(\rho_{AB_1}, \sigma_{B_2C})$. Then we have proposed a relatively simple convex optimization problem in Eq.$~$(\ref{SDP}) considering a restricted class of SEP operations. We will now show some explicit examples of preshared two-qubit states, $\rho_{AB_1}$ and $\sigma_{B_2C}$ for which we can find feasible solutions of the convex optimization problem $(iii)$ in Eq.$~$(\ref{SDP}) that are LOCC achievable. Hence, by theory such solutions are also feasible solutions for the general convex optimization in Eq.$~$(\ref{SDPgen}). Also we will some instances, when such feasible solutions saturate the upper bound achievable through any three-party LOCC.
\vskip 0.2cm 
\paragraph{{\bf Case 1:}} Suppose $\rho_{AB_1} = \ket{\Psi}\bra{\Psi}$, where $\ket{\Psi}= \sqrt{\alpha_0}~\ket{00} + \sqrt{ \alpha_1}~ \ket{11}$ is a two-qubit pure state, where $\alpha_{0,1} \in (0,1)$ and $\sigma_{B_2C}= \ket{\Phi}\bra{\Phi}$, where $\ket{\Phi}= \sqrt{\beta_0} \ket{00}+ \sqrt{\beta_1}\ket{11}$, where $\beta_{0,1} \in (0,1)$. 

Without any loss of generality, one can write the following identity, 
\begin{align}
    &\ket{\Psi}_{AB_1} \ket{\Phi}_{B_2C}=  \nonumber \\
    &\dfrac{1}{\sqrt{2}} (\sqrt{\alpha_0 ~\beta_0} \ket{00}+ \sqrt{\alpha_1 ~\beta_1}\ket{11})_{AC} \otimes \ket{\Phi_0}_{B_1B_2}  \nonumber \\
    & + \dfrac{1}{\sqrt{2}} (\sqrt{\alpha_0 ~\beta_0} \ket{00}- \sqrt{\alpha_1 ~\beta_1}\ket{11})_{AC} \otimes \ket{\Phi_1}_{B_1B_2} \nonumber \\
    & + \dfrac{1}{\sqrt{2}} (\sqrt{\alpha_0 ~\beta_1} \ket{01}+ \sqrt{\alpha_1 ~\beta_0}\ket{10})_{AC} \otimes \ket{\Phi_2}_{B_1B_2} \nonumber \\
    & +\dfrac{1}{\sqrt{2}} (\sqrt{\alpha_0 ~\beta_1} \ket{01}- \sqrt{\alpha_1 ~\beta_0}\ket{10})_{AC} \otimes \ket{\Phi_3}_{B_1B_2},
\end{align}
where each $\ket{\Phi_i}$ represents a Bell state defined as 
\begin{align}
    &\ket{\Phi_{0,1}} = \dfrac{1}{\sqrt{2}} (\ket{00} \pm \ket{11}), \nonumber  \\
    & \ket{\Phi_{2,3}} = \dfrac{1}{\sqrt{2}} (\ket{01} \pm \ket{10}). \label{bellb} 
\end{align}
Now we choose a separable operator, 
\begin{align}
    \mathcal{M} = \sum_{i=0}^3 \ket{\Phi_i}_{AC} \bra{\Phi_i} \otimes \ket{\Phi_i}_{B_1B_2} \bra{\Phi_i}. 
\end{align}
After this it is easy to check that the following expectation value gives 
\begin{align}
    \tr[\mathcal{M}~(\rho_{AB_1}\otimes \sigma_{B_2C})] &= \dfrac{1+ 4 ~\sqrt{\alpha_0~\alpha_1~\beta_0~\beta_1}}{2} \nonumber \\
    & = \dfrac{1+ C(\ket{\Psi})~C(\ket{\Phi})}{2}, \label{seppure}
\end{align}
where $C(\ket{\Psi}) = 2 \sqrt{\alpha_0~ \alpha_1}$ is the concurrence of $\ket{\Psi}$ and  $C(\ket{\Phi}) = 2 \sqrt{\beta_0~ \beta_1}$ is the concurrence of $\ket{\Phi}$.
\vskip 0.3cm
It is also easy to check that the operator $\mathcal{M}$ actually represents the separable operator $\sum_{i=1}^K X^{AC}_i(m_i^A) \otimes N_i$ in the primal optimization problem in Eq.$~$(\ref{primalsep}). Here every $X^{AC}_i(m_i^A) = \ket{\Phi_i}\bra{\Phi_i}$ and for that, the following constraint 
\[
\sum_{i=0}^3 M_i^A \otimes N_i^B= \mathbb{I}_8
\]
is satisfied because of the fact that every $M_i = 2~Tr_C[X^{AC}_i(m_i^A)]= \mathbb{I}_2$. Hence, the choice of $\mathcal{M}$ represents a feasible point of the optimization, and moreover, the value $ \tr[\mathcal{M}~(\rho_{AB_1}\otimes \sigma_{B_2C})]$ matches with the LOCC upper bound defined in Eq.$~$(\ref{loccup2}). So at least we can say the following, 
\begin{align}
    F^*(\ket{\Psi}_{AB_1}, \ket{\Psi}_{B_2C}) =  \tr[\mathcal{M}~(\rho_{AB_1}\otimes \sigma_{B_2C})] 
\end{align}
is true. 
\vskip 0.3cm 
\paragraph*{Remark:} Note that the operator $\mathcal{M}$ represents the Smolin state $\rho_S$ \cite{PhysRevLett.87.277902} after suitable normalization, where $\mathcal{M}=4 ~\rho_S$ and such a four-qubit state represents an unlockable bound entangled state if $B_1B_2$ becomes spatially separated.  
\vskip 0.3cm
\paragraph{{\bf Case 2: }} Suppose $\rho_{AB_1}$ is a rank two noisy state of the form, 
\begin{align}
&\rho_{AB_1}= \left( 1- \dfrac{p}{2}\right)~\ket{\chi}\bra{\chi} + \dfrac{p}{2} ~\ket{01}\bra{01}, \nonumber \\
where \nonumber \\
& \ket{\chi}= \sqrt{\dfrac{1}{2-p}}~(\ket{00} + \sqrt{1-p}~\ket{11}),
\end{align}
where $\dfrac{\sqrt{5}-1}{2}<p<1$. Whereas, the state $\sigma_{B_2C}$ is a pure two-qubit state of the form, $\ket{\Phi}_{B_2C}= \sqrt{\beta}~\ket{00}+ \sqrt{1-\beta}~ \ket{11}$ such that $\dfrac{1}{2}\leq \beta <1$. 

Previously we defined the complete set of Bell basis as $\lbrace \ket{\Phi_i} \rbrace$. Now we define four post-measurement states between Alice and Charlie as 
\begin{align}
    P_i~\rho_{AC}^{ii} = Tr_{B_1B_2} \left[ (\mathbb{I}_{AC} \otimes \ket{\Phi_i}_{B_1B_2} \bra{\Phi_i})~\rho_{AB_1}\otimes \sigma_{B_2C}\right]
\end{align}
for $i=0,1,2,3$, where \[P_i = \tr \left[ (\mathbb{I}_{AC} \otimes \ket{\Phi_i}_{B_1B_2} \bra{\Phi_i})~\rho_{AB_1}\otimes \sigma_{B_2C}\right] = \frac{1}{4}\] is the normalization factor. We can further represent them as   

\begin{align}
    & \rho_{AC}^{00}= \left( 1-p+ p~\beta\right)~\ket{\chi_+}\bra{\chi_+} + p~(1-\beta) ~\ket{01}\bra{01} \nonumber \\
    & \rho_{AC}^{11}= \left( 1-p+ p~\beta\right)~\ket{\chi_-}\bra{\chi_-} + p~(1-\beta) ~\ket{01}\bra{01} \nonumber \\
    &\rho_{AC}^{22}= \left( 1-p~\beta\right)~\ket{\chi'_+}\bra{\chi'_+} + p~\beta ~\ket{00}\bra{00} \nonumber \\
     &\rho_{AC}^{33}= \left( 1-p~\beta\right)~\ket{\chi'_-}\bra{\chi'_-} + p~\beta ~\ket{00}\bra{00},
\end{align}
where 
\begin{align}
    & \ket{\chi_{\pm}}= \dfrac{1}{\sqrt{1-p+p~\beta}}(\sqrt{\beta} \ket{00} \pm \sqrt{(1-\beta)(1-p)}\ket{11}) \nonumber \\
    & \ket{\chi'_{\pm}}= \dfrac{1}{\sqrt{1-p~\beta}}(\sqrt{1-\beta}~ \ket{01} \pm \sqrt{\beta~(1-p)}~\ket{10}) 
\end{align}
 
After this, we take $K=4$ and try to find a feasible solution of the SEP optimization problem \hyperref[SDPp]{$(iv)$}. Keeping this in mind, we take a separable operator of the form, 
\begin{align}
    \mathcal{M} = \sum_{i=0}^3 X_{i}^{AC}(m_i^A) \otimes N_i^{B}, \label{sepopcase2}
\end{align}
where $N_{0}^{B}= \ket{\Phi_0}\bra{\Phi_0}$, $N_{1}^{B}= \ket{\Phi_1}\bra{\Phi_1}$, $N_{2}^{B}= \ket{\Phi_2}\bra{\Phi_2}$ and $N_{3}^{B}= \ket{\Phi_3}\bra{\Phi_3}$ and 
\begin{align}
     X_{i}^{AC}(m_i^A)= ((m_i^A)^{\dagger} \otimes \mathbb{I}_2) ~\ket{\Phi_0} \bra{\Phi_0}~(m^A_i \otimes \mathbb{I}_2), 
\end{align}
where
\begin{align}
    & m^A_0 = \left( \begin{array}{cc}
       1  & 0 \\
        0 & \dfrac{\sqrt{\beta~(1-p)}}{p~\sqrt{1-\beta}}
    \end{array} 
    \right),  \nonumber \\
    &m^A_1 = \sigma_Z~m^A_0,  \nonumber \\
    & m^A_2= \sigma_X\left( \begin{array}{cc}
       1  & 0 \\
        0 & \dfrac{\sqrt{(1-\beta)~(1-p)}}{p~\sqrt{\beta}}
    \end{array} 
    \right),  \nonumber \\
    &m^A_3 = \sigma_Z~m^A_2.
\end{align}
So clearly we now have a trace non-increasing separable operator $\mathcal{M}$ which satisfies a strict feasibility condition of the problem \hyperref[SDPp]{$(iv)$}, since, we have the strict inequality condition  
\[
2~\sum_{i=0}^3 Tr_C[X_{i}^{AC}(m_i^A)] \otimes N_i^{B} \otimes \mathbb{I}_C < \mathbb{I}_{AB_1B_2C}
\]
holds for any $p\in \left( \frac{\sqrt{5}-1}{2},~1\right]$ and moreover, the success probability for implementing the operation is 
\[
p_{succ} = \tr[(2~\sum_{i=0}^3 Tr_C[X_{i}^{AC}(m_i^A)] \otimes N_i^{B} \otimes \mathbb{I}_C)~\rho_{AB_1} \otimes \sigma_{B_2C}]
\]

One important point to notice here is that after choosing the operators $\lbrace m^A_i \rbrace$, one can easily find that the following set of un-normalized density matrices, 
\begin{align}
    &\eta^{0}_{AC}= (\tilde{m}^A_0 \otimes \mathbb{I}_2 ) ~\rho_{AC}^{00} ~((\tilde{m}^A_0)^{\dagger} \otimes \mathbb{I}_2 ) \nonumber \\
     &\eta^{1}_{AC}= (\tilde{m}^A_1 \otimes \mathbb{I}_2 ) ~\rho_{AC}^{11} ~((\tilde{m}^A_1)^{\dagger} \otimes \mathbb{I}_2 ) \nonumber \\
     &\eta^{2}_{AC}= (\tilde{m}^A_2 \otimes \mathbb{I}_2 ) ~\rho_{AC}^{22} ~((\tilde{m}^A_2)^{\dagger} \otimes \mathbb{I}_2 ) \nonumber \\
      &\eta^{3}_{AC}= (\tilde{m}^A_3 \otimes \mathbb{I}_2 ) ~\rho_{AC}^{33} ~((\tilde{m}^A_3)^{\dagger} \otimes \mathbb{I}_2 ) \nonumber
\end{align}
are separable, where we define  
\[
\tilde{m}^A_i = \sqrt{\mathbb{I}_2 - (m^A_i)^{\dagger}~m^A_i}
\]
as the other four local operators. So as discussed earlier, we can replace each of this separable states with a product state so that the average fully entangled fraction does not go below $\frac{1}{2}$. So after this step, we can write the objective function of problem \hyperref[SDPp]{$(iv)$} as 
\begin{align}
    &F_{SEP}^*(\rho_{AB_1},\sigma_{B_2C})= \dfrac{1-p_{succ}}{2} \nonumber \\
    &+\sum_{i=0}^3 \tr\left[(X_{i}^{AC}(m_i^A) \otimes N_i^{B})~\rho_{AB_1} \otimes \sigma_{B_2C} \right] \nonumber \\
     =& \dfrac{1+p}{4~p} = \min \lbrace F^*(\rho_{AB_1}), F^*(\sigma_{B_2C})\rbrace, \nonumber \\ \label{sephorpure}
\end{align}
where we obtain $F_{SEP}^*(\rho_{AB_1},\sigma_{B_2C})= \min \lbrace F^*(\rho_{AB_1}), F^*(\sigma_{B_2C})\rbrace $ as long as we consider the following range, 
\[
\dfrac{\sqrt{5}-1}{2}<p<1, \quad \dfrac{1}{2}\leq \beta < \dfrac{p^2}{1-p+p^2}.
\]

Here, we can also prove that this value of $F_{SEP}^*(\rho_{AB_1},\sigma_{B_2C})$ can be achieved through an LOCC protocol, where Bob initially performs a complete Bell basis measurement. Depending on his measurement outcome $i=0,1,2,3$, Alice performs the local filtering operation $m^A_i$ with an average success probability $p_{succ}$, otherwise Alice and Charlie replace the existing state by the product state $\ket{00}_{AC}$. Furthermore, this is not the only protocol that can give this value; rather, there exists a class of non-Bell measurements performed by Bob which yield this value also This we will discuss in the next section.

\vskip 0.4cm
\paragraph{{\bf Case 3:}} 
 Consider $\rho_{AB_1}$ and $\sigma_{B_2C}$ as mixtures of Bell states i.e., $\rho_{AB_1}=\sum_{i=0}^3 p_i~ \ket{\Phi_i}\bra{\Phi_i}$ and $\sigma_{B_2C}= \sum_{i=0}^3~q_i ~\ket{\Phi_i}\bra{\Phi_i}$ such that $p_i\geq p_{i+1}$ and $q_i\geq q_{i+1}$ for $i=0,1,2,3$. In such a case, one can rewrite each of the two-qubit states as 
\begin{align}
    &\rho_{AB_1} = (\Lambda_p \otimes \mathbb{I}_2)~\ket{\Phi_0} _{AB_1}\bra{\Phi_0}, \nonumber\\
    & \sigma_{B_2C} = (\mathbb{I}_2 \otimes \Lambda_q)~\ket{\Phi_0}_{B_2C} \bra{\Phi_0}, \nonumber
\end{align}

where $\Lambda_p$ and $\Lambda_q$ are two Pauli qubit channels with Kraus operators, $\lbrace K_i = \sqrt{p_i}~\sigma_i\rbrace$ and $\lbrace L_j = \sqrt{q_j}~\sigma_j \rbrace$. In such a case, one can find a feasible point of the SEP optimization problem if we consider a separable operator of the form,
\[
\mathcal{M}= \sum_{i=0}^3 \ket{\Phi_i}_{AC}\bra{\Phi_i} \otimes \ket{\Phi_i}_{B_1B_2}\bra{\Phi_i},
\]
where $\lbrace \ket{\Phi_i}\rbrace$ is the complete set of Bell basis, then one can give a lower bound on $F^*_{SEP}(\rho_{AB_1}, \sigma_{B_1C})$ as 
\begin{align}
    F^*_{SEP}(\rho_{AB_1}, \sigma_{B_1C}) \geq \max \lbrace \dfrac{1}{2}, ~\sum_{i=0}^3p_iq_i\},
\end{align}
 where such lower bound is a feasible solution of the convex optimization problem in Eq.$~$(\ref{SDPgen}), and such value can be achieved through a Bob assisted LOCC. One can even prove that this value cannot be surpassed under any Bob-assisted LOCC, {\it i.e.}, where Bob initiates the LOCC and it is saturated when Bob performs a Bell measurement. The details of the proof calculation is given in {\bf Appendix} \ref{belldiag}.

\subsection{Comparison of different types of finite round LOCC for different choices of noisy states} \label{upb}

\begin{sce}
    Let us consider a scenario, where Alice-Bob ($\mathcal{AB}$) share a two-qubit state $\rho_{AB_1}$ and Bob-Charlie ($\mathcal{BC}$) share another two-qubit state $\sigma_{B_2C}$ and they perform an LOCC finite number of rounds belonging to $\Lambda_{B \leftrightarrow A \rightarrow C}$ type of operations.  
\end{sce}

Under {\bf Scenario 1} we consider different types of LOCC protocols which belongs to the class of protocols $\Lambda_{B \leftrightarrow A \rightarrow C}$. Generally it is extremely hard to estimate $F^*(\rho_{AB_1},\sigma_{B_2C})$ if an arbitrary $\rho_{AB_1}$ and $\sigma_{B_2C}$ is given except a trivial case, where at least one of the two-qubit state is separable. In that case, we know that $F^*(\rho_{AB_1},\sigma_{B_2C})$ cannot exceed the classical upper bound.

\vskip 0.2cm
So it is better if we consider some class of protocols $\Lambda_{B \leftrightarrow A \rightarrow C}$ and compare their efficiencies using specific class of noisy initial states which are entangled. Here by efficiency of a LOCC protocol means, we mean how close we can reach the value $F^*(\rho_{AB_1},\sigma_{B_2C})$. So let us classify three types of $\Lambda_{B \leftrightarrow A \rightarrow C}$ protocols as,   

\vskip 0.2cm 
\paragraph*{\textbf{Protocol 1}: Consider a class of $\Lambda_{B \leftrightarrow A \rightarrow C}$ protocols, where Bob first performs a projective measurement in the complete Bell basis and broadcasts the outcome to Alice and Charlie and then Bob is discarded. After hearing from Bob, Alice and Charlie performs a one-way LOCC from Alice to Charlie. \\}
\vskip 0.2cm

\vskip 0.2cm
\paragraph*{\textbf{Protocol 2}: Consider a class of $\Lambda_{B \leftrightarrow A \rightarrow C}$ protocols, where Bob first performs a projective measurement in a complete basis $\lbrace \ket{\eta_i}\rbrace_{i=0}^3$ of only non-maximally entangled states which are Bell like and broadcasts the outcome to Alice and Charlie. After that Bob's system is discarded and a one-way LOCC is performed from Alice to Charlie. Here Bob's measurement basis is restricted as} 
\begin{align}
    &\ket{\eta_0} = \alpha_0 \ket{00}+ \alpha_1 \ket{11}, \quad \ket{\eta_1} = \alpha_1 \ket{00}- \alpha_0 \ket{11}, \nonumber \\
    & \ket{\eta_2} = \beta_0 \ket{01}+ \beta_1 \ket{10}, \quad \ket{\eta_1} = \beta_1 \ket{01}- \beta_0 \ket{10}, \label{protocol2}
\end{align}
where $\alpha_0^2 > \alpha_1^2 >0$ and $\beta_0^2 > \beta_1^2 >0$ and $\alpha_0^2 + \alpha_1^2 = \beta_0^2 + \beta_1^2 =1$. 
\vskip 0.4cm 
\paragraph*{\textbf{Protocol 3}: Consider a class of LOCC protocol, where either of Alice or Bob first performs a local filtering operation and if the filter passes, then three of them continues \textbf{Protocol 1}; otherwise, if the filter fails, then Alice and Charlie replace their existing state with a pure product state.}

\begin{definition}
    For a given choice of preshared state $\rho_{AB_1} \otimes \sigma_{B_2C}$, we define the maximum fully entangled fraction achievable through a LOCC protocol $\mathcal{P}$ as $F^*_{\mathcal{P}}(\rho_{AB_1},\sigma_{B_2C})$,  which by definition satisfies the condition $F^*_{\mathcal{P}}(\rho_{AB_1},\sigma_{B_2C})\leq F^*(\rho_{AB_1},\sigma_{B_2C})$. 
\end{definition}

\subsubsection{Instances where {\bf Protocol 1} outperforms {\bf Protocol 2}}
$\bullet$  {\it Example 1:} {\it Consider $\rho_{AB_1}$ is arbitrary but an entangled two-qubit state and $\sigma_{B_2C}$ is a maximally entangled state $\ket{\Phi}_{B_2C}$ or vice versa.} 

\begin{proof}
    This proof is very simple. Since Bob-Charlie share a maximally entangled state, it acts like a noiseless teleportation channel. Hence, Bob and Charlie can simply apply the standard quantum teleportation protocol, where Bob performs measurement in the Bell basis as defined in Eq.$~$(\ref{bellb}) and Charlie performs a unitary correction. With this protocol, Alice and Charlie now share the state $\rho_{AB_1}$ that was initially shared between Alice and Bob. After this, Alice and Charlie perform an optimal one-way LOCC and obtain a fully entangled fraction 
    \begin{align*}
    F^*_{\mathcal{P}_1}(\rho_{AB_1},\sigma_{B_2C})=&F^*(\rho_{AB_1},\sigma_{B_2C}) \nonumber\\
    =&\min \lbrace  F^*(\rho_{AB_1}),~F^*(\sigma_{B_2C})\rbrace = F^*(\rho_{AB_1})
    \end{align*}
\end{proof}
\vskip 0.3cm 
$\bullet$ {\it Example 2:} {\it Consider $\rho_{AB_1}$ as a pure entangled state $\ket{\Psi}_{AB_1}= \sqrt{\alpha} \ket{00}+ \sqrt{1-\alpha} \ket{11}$, where $\alpha \in (0,1)$ and $\sigma_{B_2C}$ as another pure entangled state $\ket{\Phi}_{B_2C}= \sqrt{\beta} \ket{00}+ \sqrt{1-\beta} \ket{11}$, where $\beta \in (0,1)$. }
\begin{proof}
    This proof is also very easy to show and it is already mentioned in Eq.$~$(\ref{seppure}) of the previous section. The best fully entangled state that can be achieved here under LOCC cannot surpass the amount $\frac{1+ C(\ket{\Psi})~C(\ket{\Phi})}{2}$. 

    The optimal LOCC protocol can be realized via looking at the separable operator $\mathcal{M}$ in Eq.$~$(\ref{seppure}). In the protocol, Bob performs projective measurement in the Bell basis, and either Alice or Charlie performs unitary correction depending on Bob's measurement outcome, which is essentially \textbf{Protocol1}. 
\end{proof}

The numerical upper bound for any two given states of the form in \textit{Example 2} can be computed and plotted. The results are shown in Fig. (\ref{figpvp}). We can see that the numerical upper bound exactly matches $F^*_2$ and the \textbf{Protocol 1} exactly achieves these values. We have plotted it for a variety of pure states between Alice-Bob and Bob-Charlie and still haven't found an example that shows that \textbf{Protocol 1} is not optimal. For \textbf{Protocol 2}, we have considered the Schmidt coefficients of the non-maximally entangled basis chosen by Bob for the measurement $\alpha_0=\beta_0=\frac{2}{\sqrt{5}}$ in Eq. \eqref{protocol2}.
\begin{figure}[!h]
    \centering
    \includegraphics[height=150px, width =246px]{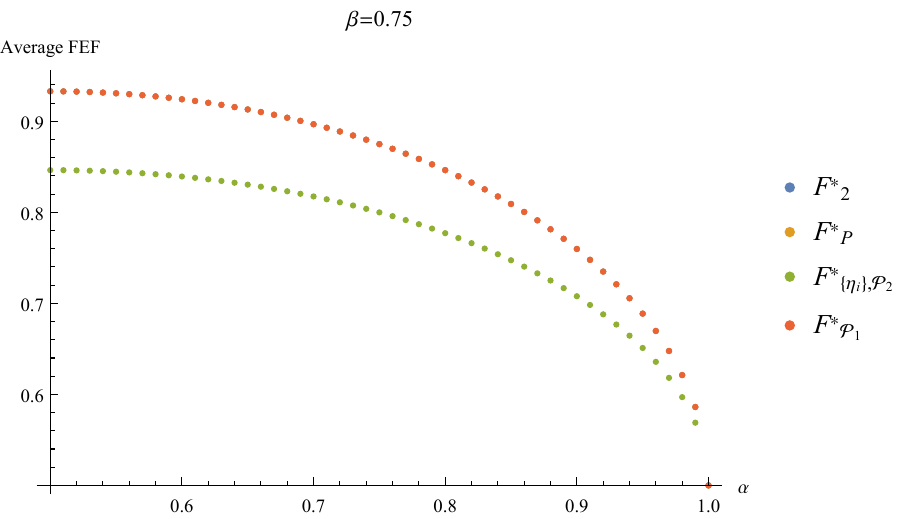}
     \caption{Here we have taken the state between Bob-Charlie as a fixed pure non-maximally entangled state with $\beta=0.75$. Here blue curve represents the upper bound $F^*_2=\dfrac{1}{2}\left( 1+C(\rho_{AB_1})~C(\sigma_{B_2C})\right)$, $F^*_P$ is the numerical upper bound from optimization problem represented by an orange curve, $F^*_{\mathscr{P}_1}$ is the fidleity achieved through Bell Measurement by Bob (\textbf{Protocol 1}) denoted by red curve and $F^*_{\{\eta_i\},{\mathscr{P}_2}}$ is the fidelity achieved through some Bell like measurement(\textbf{Protocol 2}) with the green curve. From the figure, it is clear that the orange curve, the red curve, and the blue curve coincide with each other.}\label{figpvp} 
\end{figure} 
\vskip 0.3cm 
$\bullet$ {\it Example 3:} {\it Consider both $\rho_{AB_1}$ and $\sigma_{B_2C}$ as convex mixtures of Bell states such that $\rho_{AB_1}= \sum_{i=0}^3 p_i ~\ket{\Phi_i}\bra{\Phi_i}$ and $\sigma_{B_2C}= \sum_{i=0}^3 q_i ~\ket{\Phi_i}\bra{\Phi_i}$, where $p_i\geq p_{i+1}$ and $q_i\geq q_{i+1}$ for $i=0,1,2,3$.} 
\begin{proof}
    This proof is not that straightforward. Here, it is hard to compute the exact value of $F^*(\rho_{AB_1}, \sigma_{B_2C})$. However, if we only make a comparison between {\bf Protocol 1} and {\bf Protocol 2}, then it is analytically possible to claim that {\bf Protocol 1} always outperforms {\bf Protocol 2} for such a choice of states. The detailed proof of this claim is given in {bf Appendix} \ref{belldiag}.
\end{proof}
Furthermore, the optimal \textit{fully entangled fraction} achievable by \textbf{Protocol 1} saturates the numerical upper bound. We have checked it for a plethora of Bell diagonal states of all possible ranks between Alice-Bob and Bob-Charlie and still haven't found an example otherwise. Thus, we can numerically conjecture the following
 \begin{numconjecture}
     For Bell diagonal states of the form $\rho_{AB_1}=\sum_{i=0}^3 p_i~ \ket{\Phi_i}\bra{\Phi_i}$ and $\sigma_{B_2C}= \sum_{i=0}^3~q_i ~\ket{\Phi_i}\bra{\Phi_i}$ such that $p_i\geq p_{i+1}$ and $q_i\geq q_{i+1}$ for $i=0,1,2,3$, the protocol initiated by Bob by performing complete Bell measurement on both the qubits in his lab, followed by one-way LOCC between Alice and Charlie ({\bf Protocol 1}), is the optimal protocol in our restricted setting and the expression of the optimal teleportation fidelity is
     \begin{align}
         F^*(\rho_{AB_1},\sigma_{B_2C})=\lbrace \dfrac{1}{2}, ~\sum_{i=0}^3p_iq_i\}
     \end{align}
 \end{numconjecture}
 For one of such choices of states, we plot the optimal \textit{fully entangled fraction} for Bell measurement (\textbf{Protocol 1}) by Bob vs non-Bell measurement (\textbf{Protocol 2}) by Bob with $\alpha_0=\beta_0=\frac{\sqrt{3}}{2}$ in Eq. \eqref{protocol2} in Fig. (\ref{figbdvbd}).
    \begin{figure}[!h]
    \centering
    \includegraphics[height=150px, width =246px]{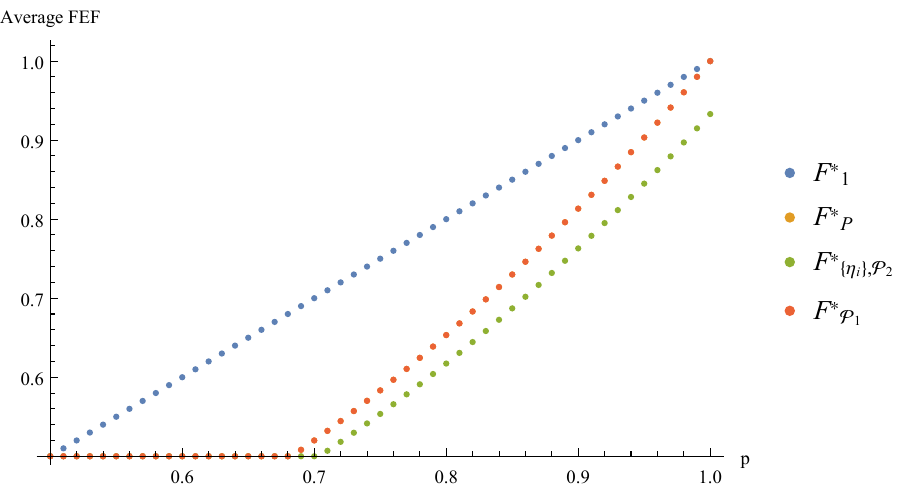}
     \caption{Here we have taken $p_i=q_i$ $\forall~i$ with $p_1=p$ and $p_2=p_3=p_4=\frac{1-p}{3}$. Here blue curve represents the upper bound $F^*_1=\min\{F^*_{\rho_{AB_1}}, F^*_{\rho_{B_2 C}}\}$, $F^*_P$ is the bound from optimization problem represented by an orange curve, $F^*_{\mathscr{P}_1}$ is the fidleity achieved through Bell Measurement by Bob (\textbf{Protocol 1}) denoted by the red curve and $F^*_{\{\eta_i\},{\mathscr{P}_2}}$ is the fidelity achieved through some Bell like measurement(\textbf{Protocol 2}) with the green curve. From the figure, it is clear that the orange curve and the red curve coincide with each other.} \label{figbdvbd} 
\end{figure} 

We observe that for this scenario involving Bell diagonal states, the numerical upper bound $F^*_P$ is strictly less than the bound $F^*_1$. Thus, the optimal \textit{fully entangled fraction} of the final state between Alice and  Charlie can never be equal to the optimal \textit{fully entangled fraction} of either of the initial Bell diagonal states shared.

\subsubsection{An instance where {\bf Protocol 2} also becomes the optimal protocol and performs same as {\bf Protocol 1}}
 $\bullet$ {\it Example 1:} { \it Consider a scenario, where Alice-Bob share a two-qubit state $\rho_{AB_1}= (\Lambda_{ADC} \otimes \mathbb{I}_2) \ket{\Psi}\bra{\Psi}$, where $\Lambda_{ADC}$ represents amplitude-damping qubit channel (ADC) acting on Alice's side and $\ket{\Psi}= \sqrt{\alpha_0} \ket{00}+ \sqrt{\alpha_1} \ket{11}$ such that $\alpha_0 \geq \alpha_1 >0$. Whereas, Bob and Charlie share a two-qubit pure entangled state $\ket{\Phi}= \sqrt{\beta_0} \ket{00} + \sqrt{\beta_1}\ket{11}$ such that $\beta_0 \geq \beta_1 >0$. }
\begin{proof}
    The proof follows from the observation in Eq.~\ref {sephorpure}) of the previous section. One can show for a class of LOCC protocols, where Bob initiates the protocol via performing projective measurement on Bell-like non-maximally entangled basis, the obtainable fully entangled fraction between Alice and Charlie saturates the LOCC achievable upper bound $\min \lbrace F^*(\rho_{AB_1}~,~F^*(\sigma_{B_2C}))\rbrace$. Therefore, the class of protocols that are optimal for this instance has the feature of {\bf Protocol 2} also. However, a detailed analysis is given in the {\bf Appendix} \ref{horodeckipureB}.
\end{proof}
Just like previous examples, for some given initial states shared between Alice-Bob and Bob-Charlie, we again plot the numerical upper bound from the optimization problem in Fig. (\ref{fighvp}). We see that in a certain region, both \textbf{Protocol 1} and \textbf{Protocol 2} coincide with each other as well as the numerical upper bound and the analytical bound $F^*_1$ also. We also see that there is a region where even \textbf{Protocol 1} strictly lies below the numerical upper bound, which itself is strictly less than the bound $F^*_1$. Thus, we get a hint that there are some instances where even Bell measurement by Bob is not optimal for the optimal teleportation fidelity distribution in a three-party scenario. 
    \begin{figure}[!h]
    \centering
    \includegraphics[height=150px, width =246px]{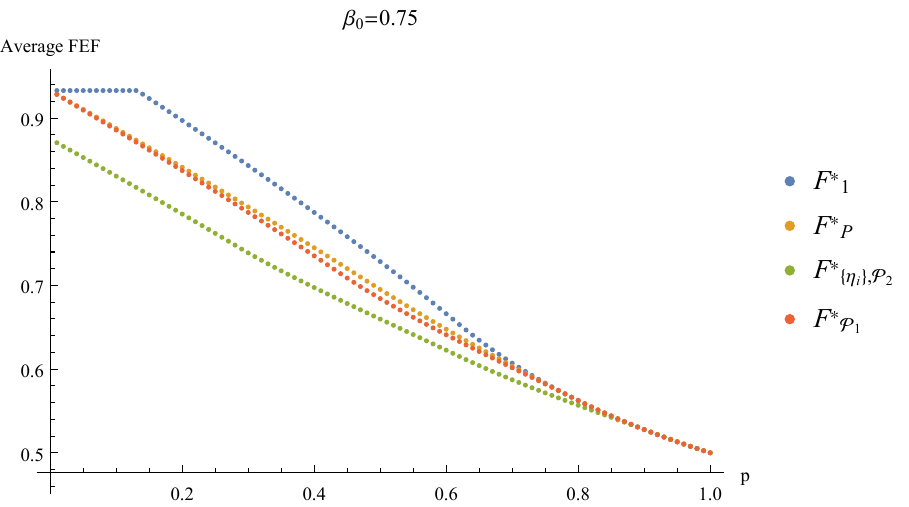}
     \caption{Here we have taken the state between Bob-Charlie as a fixed pure non-maximally entangled state with $\beta_0=0.75$ and ADC acts on a maximally entangled state shared between Alice-Bob, i.e., $\alpha_0=0.75$. Here blue curve represents the upper bound $F^*_1=\min\{F^*_{\rho_{AB_1}}, F^*_{\rho_{B_2 C}}\}$, $F^*_P$ is the numerical upper bound from optimization problem represented by an orange curve, $F^*_{\mathscr{P}_1}$ is the fidleity achieved through Bell Measurement by Bob (\textbf{Protocol 1}) denoted by red curve and $F^*_{\{\eta_i\},{\mathscr{P}_2}}$ is the fidelity achieved through some Bell like measurement(\textbf{Protocol 2}) with the green curve. From the figure, it is clear that the orange curve, the red curve, the blue curve, and the green curve all coincide with each other for a certain range of state parameter p of the state between Alice-Bob.}\label{fighvp} 
\end{figure}

\subsubsection{An instance where \textbf{Protocol 2} outperforms \textbf{Protocol 1}} \label{sub4}
\begin{proposition} \label{p5}
    Let Alice-Bob share a two-qubit Choi state $\rho_{AB_1}=(\Lambda_{ADC}\otimes \mathbb{I})\ket{\Phi_0}\bra{\Phi_0}$, where $\Lambda_{ADC}$ is an amplitude-damping channel with channel parameter $p\in (0,1)$.  Whereas, Bob-Charlie share a two-qubit Werner state  $\sigma_{B_2C}=\lambda\ketbra{\Phi_0}{\Phi_0}+(1-\lambda)\frac{\mathbb{I}_4}{4}$. Now, depending on the parameter $\lambda \in (\frac{1}{3},1)$, there exists a range $p_c(\lambda)<p<1$, for which \textbf{Protocol 2} is more efficient than \textbf{Protocol 1}, and also than all such LOCC, where Bob initiates with PVM in any complete set of maximally entangled basis. 
\end{proposition}
\begin{proof}
  Let us assume that $\rho_{AB_1}$ is a rank two state with a single state parameter $p\in (0,1)$ and can be expressed as
\begin{align}
        &\rho_{AB_1}=\left( 1-\frac{p}{2}\right )\ketbra{\psi_p}{\psi_p}+\frac{p}{2}\ketbra{01}{01}, \nonumber \\
         \text{where,}~~&\ket{\psi_p}=\frac{1}{\sqrt{2-p}}\left( \ket{00}+\sqrt{1-p} \ket{11} \right), \label{48}
    \end{align}
here note that $\rho_{AB_1}$ is a Choi state, where amplitude-damping channel $\Lambda_{ADC}$ acts one side of $\ket{\Phi_0}$.

Now let $\sigma_{B_2C}$ be a Werner state {\it i.e.,} $\sigma_{B_2C}=\lambda\ketbra{\Phi_0}{\Phi_0}+(1-\lambda)\frac{\mathbb{I}}{4}$ which is entangled if $\lambda>\frac{1}{3}$. Let us first choose $\sigma_{B_2C}$ to be weakly entangled. This means that we must choose $\lambda$ to be close to $\frac{1}{3}$. So for instance, let us choose $\lambda= \frac{2}{5}$ such that $\sigma_{B_2C}$ is weakly entangled. One can easily estimate that concurrence of $\sigma_{B_2C}$, \textit{i.e.,} $C(\sigma_{B_2C})=\frac{1}{10}$ while $\lambda=\frac{2}{5}$.\\

For the given choice of pre-shared states, finding the optimal LOCC protocol is hard unless one finds an optimal decomposition of $\eta_{AC}$. However, by considering the property of the pre-shared states and solving the SDP given in Eq.$~$(\ref{7}), one can estimate that
\begin{align}
     F^*(\rho_{AB_1})&=\frac{1}{2}\left( 1+ C(\rho_{AB_1})-\frac{p}{2}\right),~~\text{if}~~0<p<\frac{\sqrt{5}-1}{2}, \nonumber \\
    &= \frac{1}{2}\left( 1+\frac{1-p}{2p}\right),~~\text{if} ~~p>\frac{\sqrt{5}-1}{2}, \nonumber \\
     \text{and,} \quad &F^*(\sigma_{B_2C})=\frac{1+C(\sigma_{B_2C})}{2}=\frac{1+3\lambda}{4}.
\end{align}

\vskip 0.2cm 
\paragraph{Computing fully entangled fraction under \textbf{Protocol 1}:}
For the given initial state $\rho_{AB_1}\otimes \sigma_{B_2C}$, the maximum fully entangled fraction under \textbf{Protocol 1} and \textbf{Protocol 2} is defined as $F^*_{\mathcal{P}_1}(\rho_{AB_1},\sigma_{B_2C})$ and $F^*_{\mathcal{P}_2}(\rho_{AB_1},\sigma_{B_2C})$. As \textbf{Protocol 1} demands, let Bob performs PVM in an arbitrary MES basis, say $\lbrace \ket{\Psi_i^*}\rbrace$, where each $\ket{\Psi_i^*}$ can be expressed as  (see {\bf Appendix} \ref{appD} for details)
\begin{align}
    \ket{\Psi^*_i}=(\mathbb{I}\otimes W^*_i) \ket{\Phi_0}, \quad i=0,1,2,3,  
\end{align}
where $\langle \Psi_i | \Psi_j \rangle =\delta_{ij}$ and $W_i^*$ is a unitary from $SU(2)$ for every $i$. The orthogonality condition demands that $\tr(W^T_i W^*_j)=2\delta_{ij}$. 

Thus after Bob performs PVM in $\lbrace \ket{\Psi^*_i}\rbrace$ complete basis, Alice-Charlie end up with a post-measurement state 
\begin{align}
    \rho_{AC,\Psi^*_i}&= (\Lambda_{ADC}\otimes \Lambda_{dep})\ketbra{\Psi_i}{\Psi_i} \nonumber \\
    &= (\mathbb{I}\otimes W_i) (\Lambda_{ADC}\otimes \Lambda_{dep})\ketbra{\Phi_0}{\Phi_0} (\mathbb{I}\otimes W_i^{\dagger}), \label{51}
\end{align}
    which occurs with a uniform probability, \textit{i.e.,} $P_i=\tr[\rho_{AC,\Psi^*_i}]=\frac{1}{4}$. for all $i$. In Eq.$~$(\ref{51}),  $\Lambda_{dep}$ represents depolarzing channel with Kraus operators $\lbrace \sqrt{\frac{1+3\lambda}{4}} \mathbb{I},~\sqrt{\frac{1-\lambda}{4}}\sigma_1, ~\sqrt{\frac{1-\lambda}{4}}\sigma_2,\sqrt{\frac{1-\lambda}{4}}\sigma_3\rbrace$. Eq.$~$(\ref{51}) also implies that $\rho_{AC,\Psi^*_i}$ is connected with an unique state \[\eta(\Lambda_{ADC},\Lambda_{dep},\Phi_0)=(\Lambda_{ADC}\otimes \Lambda_{dep})\ketbra{\Phi_0}{\Phi_0}\] through a local unitary $W_i$. So after performing PVM in any MES, Charlie can transform the state $\rho_{AC, \Psi^*_i}$ to the state $\eta(\Lambda_{ADC}, \Lambda_{dep}, \Phi_0)$ using the reverse unitary operation $W_i^{\dagger}$. 

After preparing the $\rho_{AC,\Psi^*_i}$ for every $i$, the best thing Alice-Charlie can do is to perform optimal post-processing (see Eq.$~$(\ref{9})), and with this, the average fidelity between Alice-Charlie can be expressed as 
\begin{align}
  F^*_{ \mathcal{P}_1}(\rho_{AB_1}, \sigma_{B_2C})  = \sum_i P_i F^*(\rho_{AC,\Psi_i}).\label{47}
\end{align}
 It is shown in {\bf Appendix} \ref{appD} that if we fix $\lambda= \frac{2}{5}$, then within a range $\frac{1}{3}<p<1$, every two-qubit state $\rho_{AC,\Psi_i}$ is separable, hence, 
\begin{align}
    F^*(\rho_{AC,\Psi_i}) =\frac{1}{2} \quad and\quad C(\rho_{AC,\Psi_i})=0  \label{PVM-MES}
\end{align}
for any maximally entangled basis $\lbrace \ket{\Psi^*_i} \rbrace$ measurement. Hence, it is clear that within the mentioned value of $(\lambda,~p)$ we have 
\[
F^*_{ \mathcal{P}_1}(\rho_{AB_1}, \sigma_{B_2C}) =\dfrac{1}{2}.
\]
\newline 

\paragraph{Computing fully entangled fraction under \textbf{Protocol 2}: } Now let us consider that Bob performs PVM in some partially entangled states. Of course, the average entanglement cost of such measurement is strictly less than performing PVM in MES. Let us choose the measurement basis as 
\begin{align}
    & \ket{\eta_0}=\frac{\sqrt{3}}{2}\ket{00}+\frac{1}{2} \ket{11}, \quad \ket{\eta_1}=\frac{1}{2} \ket{00}-\frac{\sqrt{3}}{2} \ket{11}, \nonumber \\
    & \ket{\eta_2}=\frac{\sqrt{3}}{2}\ket{01}+\frac{1}{2} \ket{10}, \quad \ket{\eta_3}=\frac{1}{2}\ket{01}-\frac{\sqrt{3}}{2} \ket{10}.  \nonumber
\end{align}
After performing PVM on the above basis, the prepared state between Alice-Charlie can be written as 
\begin{align}
    \rho_{AC,\eta_i}=(\Lambda_{ADC}\otimes \Lambda_{dep})\ketbra{\eta_i}{\eta_i}, \quad \quad i=0,1,2,3, \nonumber
\end{align}
which occurs with equal probability $\frac{1}{4}$. Thus, after the state preparation, Alice-Charlie can now perform optimal post-processing on every $\rho_{AC,\eta_i}$. Hence, on average, the assisted maximum {\it fully entangled fraction} can be expressed as (see {\bf Appendix} \ref{appD} for details)
\begin{align}
    F^*_{\lbrace \eta_i\rbrace, \mathcal{P}_2}(\rho_{AB_1},\sigma_{B_2C})= \frac{1}{4} \sum_{i=0}^3 F^*(\rho_{AC,\eta_i})  \label{PVM-NMES} 
\end{align}
Now, if one compares the maximum fully entangled fraction obtained via \textbf{Protocol 1} and \textbf{Protocol 2}, one finds that the conditions
\begin{align}
    &F^*_{\mathcal{P}_1}(\rho_{AB_1},\sigma_{B_2C})< F^*_{\lbrace \eta_i\rbrace, \mathcal{P}_2}(\rho_{AB_1},\sigma_{B_2C}), \label{PVM-comparision}
\end{align}
are always satisfied while we choose $\lambda=\frac{2}{5}$ and $\frac{1}{3}<p<1$. So this proclaims that any LOCC protocol that starts with Bob's PVM operation in any MES is not necessarily optimal or even efficient. The comparison is clearly illustrated in Fig. \ref{fighvwLE}.\\

\end{proof}

 \begin{figure}[H]
    \centering
    \includegraphics[height=150px, width =246px]{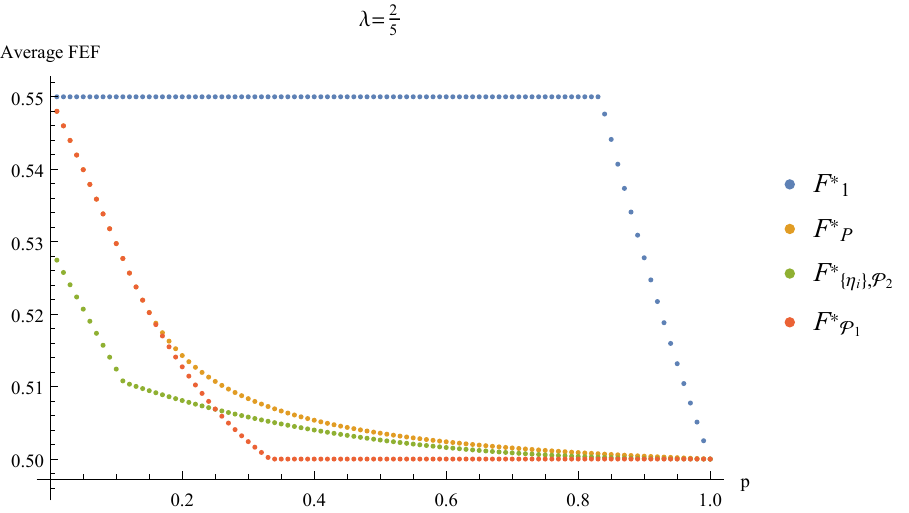}
     \caption{Here we have taken the state between Bob-Charlie as a fixed two-qubit Werner state with $\lambda=\frac{2}{5}$. Here blue curve represents the upper bound $F^*_1=\min\{F^*_{\rho_{AB_1}}, F^*_{\rho_{B_2 C}}\}$, $F^*_P$ is the numerical upper bound from optimization problem represented by an orange curve, $F^*_{\mathscr{P}_1}$ is the fidleity achieved through Bell Measurement by Bob (\textbf{Protocol 1}) denoted by red curve and $F^*_{\{\eta_i\},{\mathscr{P}_2}}$ is the fidelity achieved through some Bell like measurement(\textbf{Protocol 2}) with the green curve. From the figure, it is clear that initially the red and the orange curve coincide, but as the value of the state parameter p increases, the orange curve is strictly above the red and the green curve. For a certain parameter regime, the red curve also falls below the green curve.}\label{fighvwLE} 
\end{figure} 

For this plot, we have taken $\lambda=\frac{2}{5}$. This implies that the two-qubit Werner state shared between Bob-Charlie is weakly entangled. We can also consider a case where this state is not weakly entangled. One of such choices is by keeping $\lambda=\frac{2}{3}$. For such a scenario, we have an interesting observation. In this case, \textbf{Protocol 1} again outperforms the \textbf{Protocol 2} once again and matches with the numerical upper bound for practically all range of the state parameter p. This is illustrated in the Fig. (\ref{fighvwHE}).

\begin{figure}[h!]
    \centering
    \includegraphics[height=150px, width =246px]{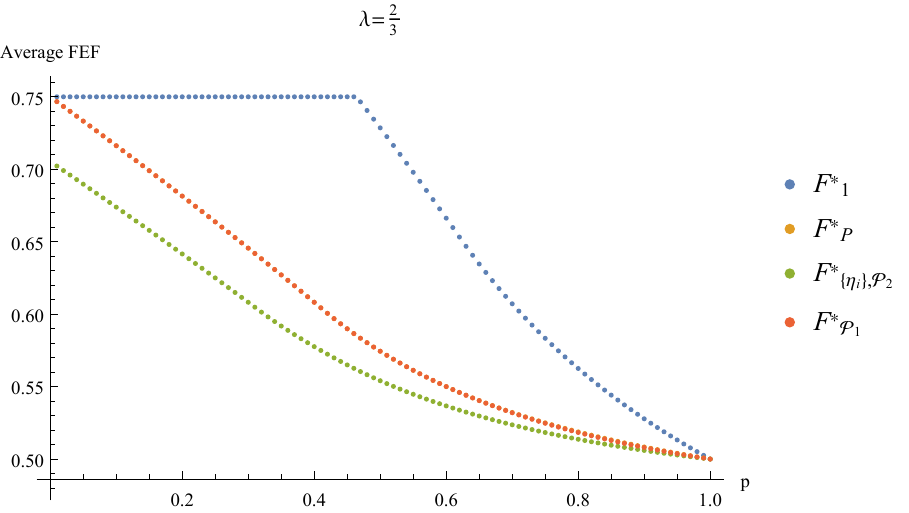}
     \caption{Here we have taken the state between Bob-Charlie as a fixed two-qubit Werner state with $\lambda=\frac{2}{3}$. Here blue curve represents the upper bound $F^*_1=\min\{F^*_{\rho_{AB_1}}, F^*_{\rho_{B_2 C}}\}$, $F^*_P$ is the numerical upper bound from optimization problem represented by an orange curve, $F^*_{\mathscr{P}_1}$ is the fidleity achieved through Bell Measurement by Bob (\textbf{Protocol 1}) denoted by red curve and $F^*_{\{\eta_i\},{\mathscr{P}_2}}$ is the fidelity achieved through some Bell like measurement(\textbf{Protocol 2}) with the green curve. From the figure, it is clear that the red and the orange curves coincide over the whole range of the state parameter p, and lie strictly above the green curve.}\label{fighvwHE} 
\end{figure} 

\subsubsection{An instance where {\bf Protocol 3} outperforms {\bf Protocol 1}}

\begin{proposition}
    Let Alice-Bob share a two-qubit Choi state $\rho_{AB_1}=(\Lambda_{ADC}\otimes \mathbb{I})\ket{\Phi_0}\bra{\Phi_0}$, where $\Lambda_{ADC}$ is an amplitude-damping channel with channel parameter $p\in (0,1)$.  Whereas, Bob-Charlie share a two-qubit Werner state  $\sigma_{B_2C}=\lambda\ketbra{\Phi_0}{\Phi_0}+(1-\lambda)\frac{\mathbb{I}_4}{4}$. Now, depending on $\lambda \in (\frac{1}{3},1)$, there exists a range of $p$ such that $0<p<p'(\lambda)<p_c(\lambda)$, where \textbf{Protocol 3} is efficient than \textbf{Protocol 1}, and also than all such LOCC, where Bob initiates with PVM in any complete set of maximally entangled basis.   \label{p5}
\end{proposition}

\begin{proof}
    Similar to our previous example, we assume that Alice-Bob share a rank two density matrix $\rho_{AB_1}=(\Lambda_{ADC} \otimes \mathbb{I})\ket{\Phi_0}\bra{\Phi_0}$, where $\Lambda_{ADC}$ is an amplitude damping channel. Now instead of applying \textbf{Protocol 1} and \textbf{2}, we apply \textbf{Protocol 3}. So for instance, Bob first applies a filtering operation,  
    \begin{align}
        A_{*}= \left(
\begin{array}{cc}
    1 & 0 \\
    0 & \dfrac{\sqrt{1-p}}{p}
\end{array}
        \right) \label{bfilter}
    \end{align}
    on the shared state $\rho_{AB_1}$. If the filter passes, then the transformed state becomes 
    \[
    \tilde{\rho}_{AB_1}= \dfrac{(\mathbb{I} \otimes A_*)~\rho_{AB_1} (\mathbb{I}\otimes A_*^{\dagger})}{\tr\left[ (\mathbb{I} \otimes A_*^{\dagger}A_*)~\rho_{AB_1} \right]},
    \]
where note that $A_*^{\dagger}A_*$ remains a valid POVM {\it iff} $p\in \left( \frac{\sqrt{5}-1}{2},1\right)$. If the filter passes, then Bob performs Bell basis measurement $\lbrace \ket{\Phi_i}\bra{\Phi_i}\rbrace $ on his part of the composite state $\tilde{\rho}_{AB_1} \otimes \sigma_{B_2C}$. Depending on the outcome Alice applies unitary correction, say $U_{A,i}$ and communicates to Charlie, who applies further unitary correction $U_{B,i}$. On the other hand, if the filter fails, then Bob does nothing and discards his system, and Alice-Charlie discard their system and replace it with a pure product state. So it is a both way LOCC between Alice and Bob, and one way from Alice to Charlie and Bob to Charlie.   

In the region $\frac{\sqrt{5}-1}{2}<p<1$, it can be shown that the average fully entangled fraction obtained between Alice and Charlie via \textbf{Protocol 3}, defined as $F^*_{\mathcal{P}_3}(\rho_{AB_1},\sigma_{B_2C})$, is strictly larger than the fully entangled fraction obtained through \textbf{Protocol 1}, \textit{i.e.,} $F^*_{\mathcal{P}_1}(\rho_{AB_1},\sigma_{B_2C})$, or  \textbf{Protocol 2}, \textit{i.e.,}even $F^*_{\lbrace \eta_i \rbrace, \mathcal{P}_2}(\rho_{AB_1},\sigma_{B_2C})$, that is
\begin{align}
&F^*_{\mathcal{P}_3}(\rho_{AB_1},\sigma_{B_2C}) > F^*_{\mathcal{P}_1}(\rho_{AB_1},\sigma_{B_2C}), \nonumber \\
&F^*_{\mathcal{P}_3}(\rho_{AB_1},\sigma_{B_2C})> F^*_{\lbrace \eta_i \rbrace, \mathcal{P}_2}(\rho_{AB_1},\sigma_{B_2C}),
\end{align}
which is clear from Fig.~\ref {fighvwLEmag}. This clearly indicates that the standard entanglement swapping protocol \cite{PhysRevA.78.032321}, RPBES protocol \cite{PhysRevLett.93.260501} are not efficient for entanglement distribution in a general sense.  

\begin{figure}[!h]
    \centering
    \includegraphics[height=150px, width =246px]{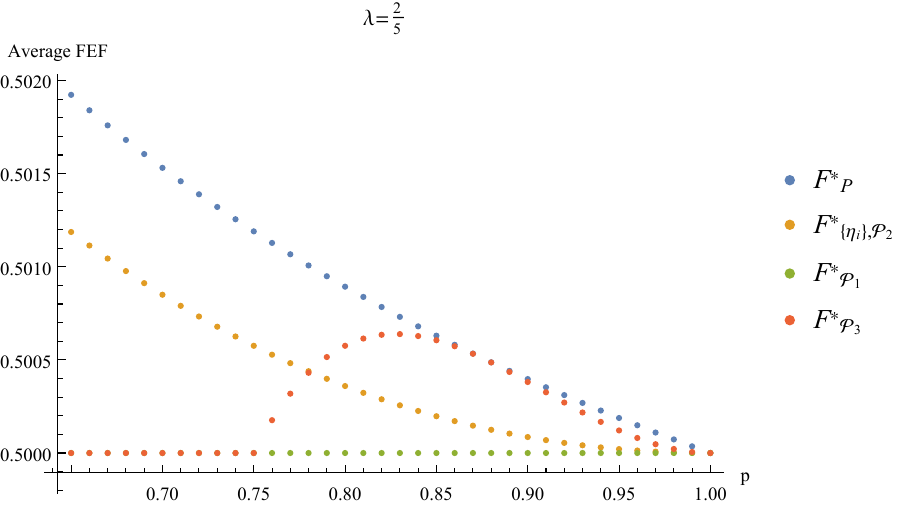}
     \caption{Here we have taken the state between Bob-Charlie as a fixed two-qubit Werner state with $\lambda=\frac{2}{5}$. Here blue curve represents the numerical upper bound from optimization problem, $F^*_{\mathscr{P}_1}$ is the fidleity achieved through Bell Measurement by Bob (\textbf{Protocol 1}) denoted by the green curve, $F^*_{\{\eta_i\},{\mathscr{P}_2}}$ is the fidelity achieved through some Bell like measurement(\textbf{Protocol 2}) with the orange curve, and the red curve here illustrates $F^*_{\mathscr{P}_3}$, which is the fidelity achieved through \textbf{Protocol 3}. From the figure, it is clear that within the range $\frac{\sqrt{5}-1}{2}<p<1$ of the state parameter p, there exists a region where the red curve lies above both the orange and the green curve and approaches the blue curve.}\label{fighvwLEmag} 
\end{figure} 

 Just like the previous example, instead of a weakly entangled two-qubit Werner state between Bob-Charlie ($\lambda=\frac{2}{5}$), we can take a more entangled two-qubit Werner state with $\lambda=\frac{2}{3}$. This gain leads us to another interesting observation illustrated in Fig. (\ref{fighvwHEmag}). \textbf{Protocol 1} again outperforms \textbf{Protocol 2} and \textbf{Protocol 3} and saturates the numerical upper bound practically for the whole allowed range of the state parameter p.

  \begin{figure}[!h]
    \centering
    \includegraphics[height=150px, width =246px]{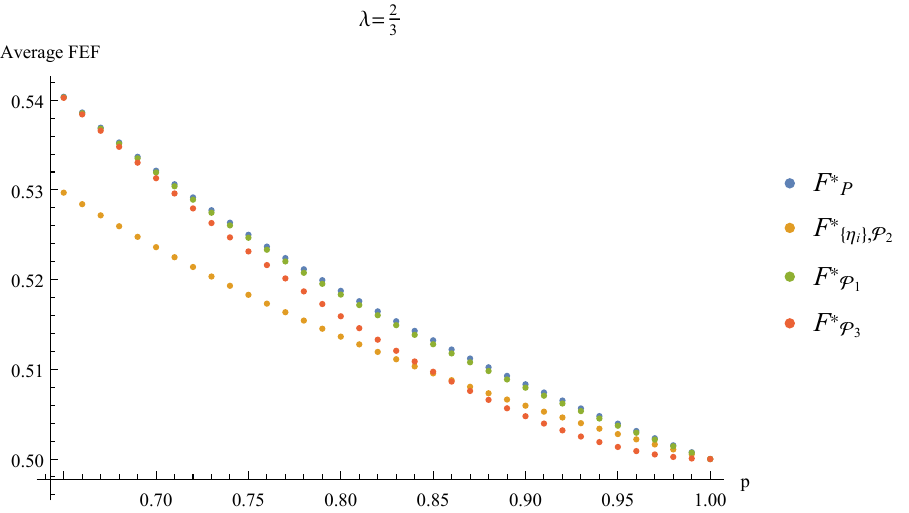}
     \caption{Here we have taken the state between Bob-Charlie as a fixed two-qubit Werner state with $\lambda=\frac{2}{3}$. Here blue curve represents the numerical upper bound from optimization problem, $F^*_{\mathscr{P}_1}$ is the fidleity achieved through Bell Measurement by Bob (\textbf{Protocol 1}) denoted by the green curve, $F^*_{\{\eta_i\},{\mathscr{P}_2}}$ is the fidelity achieved through some Bell like measurement(\textbf{Protocol 2}) with the orange curve, and the red curve here illustrates $F^*_{\mathscr{P}_3}$, which is the fidelity achieved through \textbf{Protocol 3}. From the figure, it is clear that within the range $\frac{\sqrt{5}-1}{2}<p<1$ of the state parameter p, the green curve almost coincides with the blue curve and strictly lies above both the red and the orange curve.}\label{fighvwHEmag} 
\end{figure}  

\end{proof}

\begin{observation}
    The above comparisons between different types of LOCC demonstrate that each of the protocols can be optimal under the subclass of LOCC {\it i.e.} $\Lambda_{rLOCC}$ depending on the preshared two-qubit states. The optimality of such protocol is established when the maximally obtainable fully entangled fraction under that protocol numerically agrees with the upper bound $F^*_P(\rho_{AB_1}, \sigma_{B_2C})$ which is the optimal feasible solution of the convex optimization $(iii)$ in Eq.$~$(\ref{SDP}). Also we found some instances of preshared two-qubit states, for example, when both are pure entangled states, or one is pure another is a convex mixture of a pure entangled state and a product state, for which  $F^*_P(\rho_{AB_1}, \sigma_{B_2C})$ becomes a LOCC monotone. 
\end{observation}

\section{Distribution of optimal teleportation channel beyond three parties}\label{IV}

\begin{observation}
    Consider a linear network of $N$ parties (nodes), where every segment shares a bipartite state $\rho_{i,i+1}$, - and all parties can perform LOCC. In such a case the maximum fully entangled fraction that $A_1$ and $A_N$ can achieve is upper bounded by the LOCC bound 
    \[
    F^*_{1N}\leq F^*=\min \left \lbrace F^*(\rho_{12}),~F^*(\rho_{23}),....,F^*(\rho_{N-1N}) \right \rbrace
    \]
\end{observation}
One more instance where we see the more nuanced and complex traits of teleportation fidelity optimizing LOCC protocols is the following.   Consider a linear network of $N$ nodes, where every segment shares a two-qubit state $\rho_{i,i+1}$ and all parties can perform LOCC. Only the first segment is faulty, \textit{i.e.,} $\rho_{12}$ is noisy (affected by local amplitude damping noise) and rest of {\color{blue}the} segments (free segments) share pure entangled states $\rho_{i,i+1}=\ket{\psi_{i,i+1}}\bra{\psi_{i,i+1}}$ such that $F^*(\rho_{12})< F^*(\rho_{i,i+1})$ holds for every $2\leq i\leq N$. The results of Appendix \ref{horodeckipureB} can be generalized to this scenario also. In \cite{PhysRevLett.134.160803} (Corollary), a similar type of scenario is considered. The noisy state has a similar structure to that of the state considered here in Appendix \ref{horodeckipureB}. If the same pure entangled state is shared in N segments, it is shown that they are not required to be maximally entangled to establish the optimal teleportation fidelity between the sender and the receiver. There is a lower bound on the amount of necessary concurrence. Using the value at which this bound is saturated, let's call it $\mathcal{C}$, a quantity is defined, namely, the resource saved $R_v$ given as
\begin{equation}
    R_c=N(1-\mathcal{C}).
\end{equation}
It quantifies how much less resource is consumed in an $N+1$ segment repeater scenario. As the structure of the states considered here in Appendix \ref{horodeckipureB} is similar to the one considered in Appendix C in the supplementary material of \cite{PhysRevLett.134.160803}, following the exact procedure for the states considered here the saved resource $R_c$ is upper bounded by 
\begin{equation}
   R^{I}_c\leq N-N \left(1-\Bigg(\frac{p^2- (1-p)^2}{p^2+ 
   (1-p)^2}\Bigg)^2\right)^{\frac{1}{2N}}.
\end{equation}
However, this number is achievable only for a particular strategy. First, we have to do RPBES measurement \cite{PhysRevLett.93.260501} in all the segments containing pure non-maximally entangled states. This will result in a three-party (two-segment) scenario where the first segment will have the noisy state and the second segment will have some non-maximally entangled state. Then, for this reduced structure, we perform the Bell measurement. Let us call this \textbf{Strategy I}.

We can also approach the task of establishing optimal fidelity in another way. We can perform Bell measurement at each segment. Let us label this as \textbf{Strategy II}. While the RPBES protocol deterministically prepares the state with optimal entanglement, and each segment Bell measurement will, on average, prepare the state with optimal concurrence. Due to this stark difference, the behavior in the variance of resource consumption will be different in both cases. 

Considering the above-mentioned scenario of $N+1$ segment quantum repeater scenario with Alice (A) being the sender and Bob (B) being the receiver, and there are $\{N_{i}\}^{N}_{i=1}$ nodes between them. The joint state shared by $N+1$ segments is written as:
\begin{align}
    \eta_{AN_1\cdots N_{N}B}=\eta=\Big(\Lambda_{ADC}& \otimes \mathbb{I})\ket{\phi^{+}}_{AN_1}\bra{\phi^{+}}\Big)\nonumber\\
    &\qquad \qquad\otimes \Big(\bigotimes_{i=2}^{N+1}\ket{\psi}\bra{\psi}\Big).
\end{align}
with $\ket{\psi}=\sqrt{\beta_0}\ket{00}+ \sqrt{\beta_1}\ket{11}$. Digressing away from RPBES protocol, if we perform Bell measurement in each segment, we will have $4^N$ possible states. Applying the procedure in Appendix \ref{horodeckipureB} repeatedly, the condition to obtain optimal fidelity comes out to be
\begin{align}
   \sqrt{\beta^{N}_0~\beta^{N}_1(1-p)}\leq p~\beta^{N}_1
\end{align}
As $\mathcal{C}=2\sqrt{\beta_0\beta_1}$, variance in the resource consumption can similarly be calculated as 

\begin{equation}
   R^{II}_c\leq N-N \Bigg(1-\Bigg(\frac{\Big(\frac{p^2}{1-p}\Big)^{\frac{1}{N}}-1}{\Big(\frac{p^2}{1-p}\Big)^{\frac{1}{N}}+1}\Bigg)^2\Bigg)^{\frac{1}{2}}.
\end{equation}
We plot the resource consumption for both strategies in Fig. \ref{Dist}.

\begin{figure}[!h]
    \centering
    \includegraphics[height=170px, width =246px]{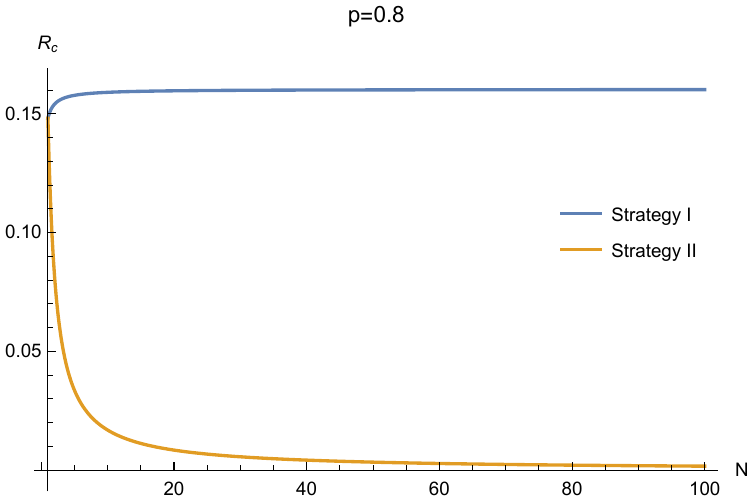}
   \caption{Plot of variance in resource consumption for \textbf{Strategy I}: RPBES measurements repeatedly followed by Bell measurement, and \textbf{Strategy II}: Just Bell measurement repeatedly for a fixed value of the noise parameter $p=0.8$ of the noisy state in the first segment with an increase in the number of nodes.}\label{Dist}
\end{figure}  

We see in the case of \textbf{Strategy I} that variance in resource consumption starts saturating to a fixed value, while for \textbf{Strategy II} it goes to zero. Although resource consumption is possible in both cases, its maintenance over a large number of nodes shows contrasting behaviors. Thus, we see that beyond the three-party scenario, the distribution of optimal fidelity is not that simple. A lot of other facets have to be taken into account as well.

\section{Discussion} \label{dis}

In this work, we consider a scenario involving three spatially separated parties---Alice, Bob, and Charlie---each located in different labs. Alice and Bob share a two-qudit state (d$\times$d state), as do Bob and Charlie. Thus, Alice and Charlie each hold one qudit, whereas Bob possesses two qudits. Each party can perform local quantum operations on their own qudits and then communicate the results classically. After Bob performs his operations and traces out his qudits, a two-qudit state is formed between Alice and Charlie, which corresponds to a tripartite Local Operations and Classical Communication (LOCC) protocol among the three parties. The objective of this work is to maximize the teleportation fidelity of the state between Alice and Charlie, or in other words, to establish the best teleportation channel between them using LOCC operations on the combined system of Alice, Bob, and Charlie.

In general, when the pre-shared states are noisy, things are quite nontrivial.  In this case, it is challenging to determine the optimal LOCC protocol for distributing the best teleportation channel between Alice and Charlie, as fully characterizing LOCC operations mathematically is difficult. Moreover, the question of who should initiate the protocol and how many rounds of classical communication are needed to achieve the desired result adds complexity to the problem. To circumvent these challenges, we instead consider separable operations (SEP), which have well-understood mathematical properties. While this doesn't give us an exact expression for the \textit{optimal teleportation fidelity}, since LOCC operations are a subset of SEP operations, it allows us to derive bounds on it by optimizing \textit{optimal teleportation fidelity} for given states between Alice-Bob and Bob-Charlie over the set of SEP operations. But although the optimization problem turns out to be convex, the variable over which the optimization has to be performed is a $d^8\times d^8$ matrix, which is difficult to tackle. But we observe that if we considered a restricted class of separable operations where Charlie's Krauss operators are unitary while Alice and Bob's Krauss operators are as general as possible, the optimization problem simplifies significantly to an optimization problem involving just a $d^4\times d^4$ matrix as a variable which has to be optimized. It can be solved numerically for a given set of states between Alice-Bob and Bob-Charlie. This provides us with the numerical value of \textit{optimal teleportation fidelity} that can be established between Alice and Charlie, and cannot be surpassed by any LOCC protocol which involves a both-way communication between Alice and Bob but only one-way communication to Charlie from both. Thus, we provide a testing bed for studying the gap between SEP operations and LOCC operations, although in a restricted but still interesting setting.  Besides this, we also gave different independent analytical bounds on the teleportation fidelity. We study the special case of qubits \textit{i.e.,} $d=2$. In this case, the independent bounds analytical on the \textit{optimal teleportation fidelity} have simpler forms. As hinted earlier a natural question arises: Are these bounds achievable through an LOCC protocol? The answer to this question is far from trivial. The protocol to achieve any of these bounds depends on the initial states shared between Alice-Bob and Bob-Charlie, making it difficult to define a general strategy. Therefore, we approach the problem on a case-by-case basis. Fortunately, we identified some examples where these bounds are saturated, or in some specific cases, for some specific class of states, even a different independent bound is achieved. For these cases, we also identified the corresponding LOCC protocols that achieve these bounds. It is again worth mentioning that all this is done in a restricted setting

The problem becomes relatively straightforward if the pre-shared states between Alice-Bob and Bob-Charlie are maximally entangled pure states, each possessing $1$ ebit of entanglement. In this case, the optimal protocol is the standard Entanglement Swapping (ESWAP) protocol. Bob teleports one of his qubits using the maximally entangled Bell state shared with either Alice or Charlie, and Alice and Charlie perform local unitary corrections to establish a noiseless teleportation channel. Even when the states shared between Alice-Bob and Bob-Charlie are not maximally entangled, but still pure, the optimal protocol for establishing \textit{optimal teleportation fidelity} between Alice and Charlie is straightforward and involves the RPBES protocol \cite{PhysRevLett.93.260501}.

In another scenario where Alice and Bob share a noisy entangled state, if Bob-Charlie share less than one ebit of entanglement, one might expect that the optimal teleportation fidelity between Alice and Charlie would be lower compared to a case where one ebit of entanglement is shared between Alice and Charlie. However, we show that for a certain class of noisy states, these two strategies can yield the same teleportation fidelity within certain parameter ranges. Moreover, even if Bob performs a measurement in a non-maximally entangled basis (rather than using the maximally entangled Bell basis in the ESWAP protocol), it is still possible to achieve the same optimal teleportation fidelity for Alice and Charlie. When both Alice-Bob and Bob-Charlie share mixed states, things get more complicated. We found that if both parties share Bell-diagonal states, the optimal teleportation has a tighter analytical bound than those previously discussed independent analytical bounds. In this case, a one-way LOCC protocol, where Bob performs a Bell measurement and Alice corrects with a unitary operation, is sufficient to achieve optimal fidelity. Interestingly, we also discovered in a scenario where Alice-Bob share a noisy state belonging to a class and Bob-Charlie share an isotropic Werner state, one can always find a measurement in a non-maximally entangled basis that achieves higher fidelity than measurements in any maximally entangled basis, including the Bell basis, given the isotropic Werner state is weakly entangled. Thus, even weakly entangling operations can have an operational advantage in certain scenarios, although these strategies are suboptimal. But as the entanglement of the isotropic Werner state increases, Bell-measurement not only surpasses the non-maximally entangled measurements in terms of the teleportation fidelity that can be established between Alice and Charlie, but also becomes optimal. These findings highlight that, in the case of noisy states, the optimal strategies are not always straightforward to characterize.

We have also studied the scenario beyond the three-party case and observed that with an increasing number of parties, the problem becomes more complex. For an N-party network, we compare two different strategies for establishing \textit{optimal teleportation fidelity} and see how the distribution of entanglement after measurement by each party subtly affects its distribution. However, our study has some limitations. As mentioned earlier, even for the three-party case, we could not answer with full generality that for arbitrary pre-shared states, what should be the exact structure of the optimal protocol? We have tried to propose a convex optimization problem for the most general SEP protocol, but could solve it only in a restricted scenario.  One direction to venture would be to take advantage of some kind of symmetries of the analytic expression of teleportation fidelity, if present. It might be difficult, but it is still an interesting and worthwhile direction to explore, which we leave as the subject for future work. There are also other open directions for our study. As the bounds on optimal teleportation fidelity are independent of each other it would be interesting to classify the class of states for which a particular bound is achieved. As it is difficult to handle the LOCC protocols mathematically in full generality, one has to consider a specific protocol, but still, this kind of classification will be quite interesting. This may also lead to a unified expression for optimal teleportation fidelity which reduces to different bounds for different classes of states. But to our knowledge, this is a very involved and difficult task to accomplish. One can also consider the pre-shared states of $\mathcal{AB}$ and $\mathcal{BC}$ to be of infinite dimensions, a two-mode Gaussian state, for example. This is because Gaussian states are easy to prepare experimentally. Thus, extending this formalism to continuous variable cases will be an interesting work. Furthermore, in any network-based problem, the distribution of the teleportation channel is only one of the properties that you want to be distributed. There can be other different figures of merit one can think of in such distribution problems, where each of the figures of merit comes from a physical task. Examining such distribution tasks may lead to more interesting properties.

\bibliographystyle{apsrev4-1}
	\bibliography{ref}

\newpage
\appendix 
\begin{widetext}
\section{Non-Convexity of the restricted SEP optimization problem}\label{nonconvexity}
As the SEP optimization problem is nonconvex, it is relatively easier to deal with it if we can split it into two-parts to compute $F^*_{rSEP}(\rho_{AB_1},\sigma_{B_2C})$ as
\begin{align}
     & \max_{\lbrace X_i, ~N_{i}\rbrace}~  \dfrac{1}{d}\nonumber\\
     &- \sum_{i=1}^K \tr\left[ (Tr_C[X^{AC}_i(m_i^A)]\otimes N_i^B\otimes \mathbb{I}_C)~\rho_{AB_1} \otimes \sigma_{B_2C}\right] \nonumber \\
     & + \sum_{i=1}^K \tr\left[ (X^{AC}_i(m_i^A) \otimes N_i^B)~\rho_{AB_1} \otimes \sigma_{B_2C}\right] \nonumber \\
     & =\max_{\lbrace X^{AC}_i(m_i^A) \rbrace} ~f(\rho_{AB_1}, \sigma_{B_2C}, \lbrace X^{AC}_i(m_i^A) \rbrace),
\end{align}
where $\lbrace X^{AC}_i(m_i^A) \rbrace $ is a finite set of rank one positive operators as defined in Eq.$~$(\ref{primalsep}). So we now propose a two-step optimization as 
\begin{align}
     &(ii).a \quad  f(\rho_{AB_1},\sigma_{B_2C},\lbrace X^{AC}_i(m_i^A)\rbrace)=\max_{\lbrace N_i\rbrace} ~  \dfrac{1}{d}\nonumber\\
     &- \sum_{i=1}^K \tr\left[ (Tr_C[X^{AC}_i(m_i^A)]\otimes N_i^B\otimes \mathbb{I}_C)~\rho_{AB_1} \otimes \sigma_{B_2C}\right] \nonumber \\
     &+\sum_{i=1}^K \tr \left[ (X^{AC}_i(m_i^A) \otimes N^B_i)~\rho_{AB_1} \otimes \sigma_{B_2C}\right] \nonumber  \\
    s.t.  \quad & d~\sum_{i=1}^K Tr_C[X^{AC}_i(m_i^A)]\otimes N^B_i\leq  \mathbb{I}_{d^3}\nonumber \\
      & N_i^B\geq \mathbf{0} \label{prim1}     
\end{align}
More precisely, what problem \hyperref[prim1]{$(ii).a$} actually implies is that we fix the set of complex operators $\lbrace m_i \rbrace $ and solve a maximization problem over the set of all positive operators $\mathcal{P}(\mathcal{H}(d^2))$, which forms a convex set. Thus, problem \hyperref[prim1]{$(ii)$} is a convex optimization problem for a given choice of operators $\lbrace m_i \rbrace $, and this is because the objective function is linear in terms of $N_i$ and every $N_i$ is a member of $\mathcal{S}_{pos}$. 

After solving problem \hyperref[prim1]{$(ii).a$}, we have to solve a nonconvex optimization,
\begin{align}     
      (ii).b\quad&F^*_{rSEP}(\rho_{AB_1},\sigma_{B_2C})\nonumber\\ &\qquad\qquad\qquad=\max_{\lbrace X^{AC}_i(m_i^A)\rbrace} f(\rho_{AB_1},\sigma_{B_2C},\lbrace X^{AC}_i(m_i^A)\rbrace) \nonumber \\
      &s.t. \quad  d~\sum_{i=1}^K Tr_C[X^{AC}_i(m_i^A)]\otimes N_i^B\leq  \mathbb{I}_{d^3} \nonumber \\
      & X^{AC}_i(m_i^A) \geq \mathbf{0}, \quad rank(X^{AC}_i(m_i^A))=1,\label{primal2}
    \end{align}

where the optimization problem \hyperref[primal2]{$(ii).b$} is not a convex optimization problem as the variable set is not convex. Hence, the overall optimization problem is again nonconvex but it is at least tractable by parts, i.e., if we fix the set $\lbrace X^{AC}_i(m_i^A) \rbrace$, we can definitely figure out the optimal $\lbrace N_i^B\rbrace$ via solving problem \hyperref[prim1]{$(ii).a$}. However, there is a way via which we can convexify the nonconvex problem \hyperref[primal2]{$(ii).b$} by considering the convex hull of the variable set $\lbrace X^{AC}_i(m_i^A) \rbrace$ as done in the main text. 

After considering this convex hull, we can rewrite the nonconvex optimization under SEP in Eq.$~$(\ref{primal2}) as its convex version: 
\begin{align}
    (ii).c \quad & F^*_{rSEP}(\rho_{AB_1},\sigma_{B_2C}) = \max_{\left\lbrace L^{AC}_i\right \rbrace}~ f(\rho_{AB_1}, \sigma_{B_2C}, \lbrace L_i^{AC}\rbrace) \nonumber \\
    s.t. \quad & \sum_{i=1}^K d~Tr_C[L^{AC}_i] \otimes N_i^* \otimes \mathbb{I}_C \leq  \mathbb{I}_{AB_1B_2C} \nonumber \\
    &L_i^{AC}\in Conv(\mathcal{E}), ~~~\mathbf{0} \leq L_i^{AC}\otimes N_i^* \leq \mathbb{I}_{AB_1B_2C} \label{SDPp} 
\end{align}
where the set $\lbrace N_i^* \rbrace$ is a set of operators that are optimal for a given choice of the initial set of variables $ \lbrace L_i^{AC}\rbrace$. As every $L_i^{AC}\in Conv(\mathcal{E})$, hence, every $L_i^{AC}$ can be expressed as a convex combination of $\lbrace X^{AC}_i(m_i^A) \rbrace _{i=1}^K$, hence, the maximum value of the objective function $f(\rho_{AB_1}, \sigma_{B_2C}, \lbrace L_i^{AC}\rbrace)$ is achieved at the extreme points, that is $\lbrace X_{i}^{AC}(m_i^A)\rbrace$.

 The optimization problem \hyperref[SDPp]{$(iv)$} is a convex optimization problem in terms of the individual set of variables, $\lbrace L_i^{AC}\rbrace$ or $\lbrace N_i^B\rbrace$ while the other variable set remains constant. However, it is not globally a convex optimization problem as the objective function is nonlinear in terms the set of variables. This optimization problem can be solved via an iterative method, where one first fixes a set of variables as an initial guess and solves the optimization in terms of the other variable set and again continue the same process until one reaches the true global maxima.  However, it is not necessarily guaranteed that one can always reach the global maxima for arbitrary choice of $\rho_{AB_1}$ and $\sigma_{B_2C}$. In fact there can be a chance that this repetitive approach could not find the global maxima.

\section{The proof that Protocol 2 becomes the optimal protocol for specific choice of two-qubit states } \label{horodeckipureB}
$\bullet $ {\bf Protocol 1}: Consider a class of $\Lambda_{B \leftrightarrow A \rightarrow C}$ protocols, where Bob first performs a projective measurement in any maximally entangled basis and broadcasts the outcome to Alice and Charlie and then Bob is discarded. After hearing from Bob, Alice and Charlie performs a one-way LOCC from Alice to Charlie.

\vskip 0.3cm 

$\bullet$ {\bf Protocol 2}: Consider a class of $\Lambda_{B \leftrightarrow A \rightarrow C}$ protocols, where Bob first performs a projective measurement in a complete basis $\lbrace \ket{\eta_i}\rbrace_{i=0}^3$ of only non-maximally entangled states which are Bell like and broadcasts the outcome to Alice and Charlie. After that Bob's system is discarded and a one-way LOCC is performed from Alice to Charlie. Here Bob's measurement basis is restricted as
\begin{align}
    &\ket{\eta_0} = \sqrt{\delta_0} \ket{00}+ \sqrt{\delta_1} \ket{11}, \quad \ket{\eta_1} = \sqrt{\delta_1} \ket{00}- \sqrt{\delta_0} \ket{11}, \nonumber \\
    & \ket{\eta_2} = \sqrt{\delta'_0} \ket{01}+ \sqrt{\delta'_1} \ket{10}, \quad \ket{\eta_1} = \sqrt{\delta'_1} \ket{01}- \sqrt{\delta'_0} \ket{10}, \nonumber
\end{align}
where $\delta_0 > \delta_1 >0$ and $\delta'_0 > \delta'_1 >0$ and $\delta_0 + \delta_1 = \delta'_0 + \delta'_1 =1$. 

\vskip 0.5cm 
$\bullet$ {\bf Scenario}: { \it Consider a scenario, where Alice-Bob share a two-qubit state $\rho_{AB_1}= (\Lambda_{ADC} \otimes \mathbb{I}_2) \ket{\Psi}\bra{\Psi}$, where $\Lambda_{ADC}$ represents amplitude-damping qubit channel (ADC) acting on Alice's side and $\ket{\Psi}= \sqrt{\alpha_0} \ket{00}+ \sqrt{\alpha_1} \ket{11}$ such that $\alpha_0 \geq \alpha_1 >0$. Whereas, Bob and Charlie share a two-qubit pure entangled state $\ket{\Phi}= \sqrt{\beta_0} \ket{00} + \sqrt{\beta_1}\ket{11}$ such that $\beta_0 \geq \beta_1 >0$. }

\vskip 0.4cm
\begin{proof}
    Suppose we have a pair of two-qubit entangled states $(\rho_{AB_1},\sigma_{B_2C})$ shared between Alice-Bob and Bob-Charlie, such that $\rho_{AB_1}=(\Lambda_{ADC} \otimes \mathbb{I})\ket{\psi}_{AB_1}\bra{\psi}$ and $\sigma_{B_2C}=\ket{\phi}_{B_2C}\bra{\phi}$. Here $\ket{\psi}_{AB_1}$ and $\ket{\phi}_{B_2C}$ have their Schmidt decomposition as,  
\begin{align}
    \ket{\psi}_{AB_1}&= \sqrt{\alpha_0}\ket{00}+ \sqrt{\alpha_1} \ket{11}, \quad \quad \quad \alpha_0\geq \alpha_1 >0\nonumber \\
    &= (A_{\alpha} \otimes \mathbb{I})\ket{\Phi_0}, \quad  \quad \text{where,} \quad A_{\alpha} = diag\lbrace \sqrt{2~\alpha_0},~\sqrt{2~\alpha_1}~\rbrace \nonumber \\
     \ket{\phi}_{B_2C}&= \sqrt{\beta_0}\ket{00}+ \sqrt{\beta_1}\ket{11}, \quad \quad \quad \beta_0 \geq \beta_1 >0, \nonumber \\
     &= (\mathbb{I} \otimes B_{\beta})\ket{\Phi_0}, \quad \quad \text{where,} \quad B_{\beta} = diag\lbrace \sqrt{2~\beta_0},~\sqrt{2~\beta_1}\rbrace, \label{b1}
\end{align}

where $A_{\alpha}^{\dagger}A_{\alpha}$ and $B_{\beta}^{\dagger}B_{\beta}$ in Eq.$~$(\ref{b1}) are trace non-increasing CP maps, \textit{i.e.,} $A_{\alpha}^{\dagger}A_{\alpha}\leq \mathbb{I}$ and $B_{\beta}^{\dagger}B_{\beta} \leq \mathbb{I}$, acting on the Bell state $\ket{\Phi_0}$. So the composite four-qubit state can be expressed as  
\begin{align}
    \eta_{AB_1B_2C}&= (\Lambda_{ADC} \otimes \mathbb{I})\ket{\psi}_{AB_1}\bra{\psi} \otimes \ket{\phi}_{B_2C}\bra{\phi} \nonumber \\
    & = \sum_{i=0}^{1} (K_{i,p}\otimes \mathbb{I})\ket{\psi}_{AB_1}\bra{\psi} (K_{i,p}^{\dagger} \otimes \mathbb{I}) \otimes \ket{\phi}_{B_2C}\bra{\phi} \nonumber \\
    &= \sum_{i=0}^{1} (K_{i,p}~A_{\alpha}\otimes \mathbb{I})\ket{\Phi_0}_{AB_1}\bra{\Phi_0}(A_{\alpha}^{\dagger}~K_{i,p}^{\dagger}\otimes \mathbb{I})\otimes (\mathbb{I} \otimes B_{\beta})\ket{\Phi_0}_{B_2C}\bra{\Phi_0}(\mathbb{I} \otimes B_{\beta}^{\dagger }) \nonumber \\
    &= \dfrac{1}{4}\sum_{k,l=0}^{3}\ket{\psi^{*}_k}_{B_1B_2}\bra{\psi^{*}_l}\otimes \sum_{i=0}^{1}(K_{i,p}~A_{\alpha} \otimes B_{\beta}) \ket{\psi_k}_{AC}\bra{\psi_l}(A_{\alpha}^{\dagger}~K_{i,p}^{\dagger}\otimes B_{\beta}^{\dagger}),
\end{align}
where in the above derivation $\lbrace K_{0,p},~K_{1,p}\rbrace$ denotes the Kraus operators of amplitude-damping channel (ADC) with channel parameter $p$ such that 
\begin{align}
    K_{0,p}= \left( \begin{array}{cc}
        1 & 0 \\
        0 & \sqrt{1-p}
    \end{array}
    \right), \quad \quad K_{1,p}= \left( \begin{array}{cc}
        0 & \sqrt{p} \\
        0 & 0
    \end{array}
    \right), \quad \quad 0<p<1
\end{align}
and also in the above derivation, $\lbrace \ket{\psi_k} \rbrace$ forms any orthonormal set of basis and $*$ is the complex conjugate of the vectors. Thus if Bob performs measurement in the basis $\lbrace \ket{\psi^{*}_j}_{B_1B_2}\rbrace$, then the corresponding prepared state between Alice-Charlie becomes 
\begin{align}
    \rho_{AC,\psi_j}&=\dfrac{1}{p_j}~Tr_{B}\left((\ket{\psi^{*}_j}_{B_1B_2}\bra{\psi^{*}_j}\otimes \mathbb{I})\eta_{AB_1B_2C} \right) \nonumber \\
    & =\dfrac{1}{p_j}~ \sum_{i=0}^{1}(K_{i,p}~A_{\alpha} \otimes B_{\beta}) \ket{\psi_j}_{AC}\bra{\psi_j}(A_{\alpha}^{\dagger~}K_{i,p}^{\dagger}\otimes B_{\beta}^{\dagger}) \nonumber \\
    & =\dfrac{1}{p_j}~ (\Lambda_{ADC} \otimes \mathbb{I})(A_{\alpha} \otimes B_{\beta}) \ket{\psi_j}_{AC}\bra{\psi_j} (A_{\alpha}^{\dagger} \otimes B_{\beta}^{\dagger}) \nonumber \\
    & = (\Lambda_{ADC} \otimes \mathbb{I}) \ket{\chi_j}_{AC} \bra{\chi_j}, \quad \quad j=0,1,2,3, 
\end{align}

where  $\ket{\chi_{j}}_{AC}=\dfrac{1}{\sqrt{p_j}} (A_{\alpha} \otimes B_{\beta})\ket{\psi_j}_{AC}$, and $p_j=\tr\left( (A_{\alpha}^{\dagger}A_{\alpha}\otimes B_{\beta}^{\dagger}B_{\beta})\ket{\psi_j}_{AC}\bra{\psi_j}\right)$ is the probability of obtaining the state $\rho_{AC,\psi_j}$.

\paragraph{{\bf Bob performs PVM}:}
Now let us assume that Bob first starts the protocol by performing a PVM in the following basis: 
\begin{align}
    & \ket{\eta_0}= \sqrt{\delta_0} \ket{00}+ \sqrt{\delta_1} \ket{11}, \quad \ket{\eta_1}= \sqrt{\delta_1} \ket{00}-\sqrt{\delta_0} \ket{11}, \nonumber \\
    & \ket{\eta_2}= \sqrt{\delta_0} \ket{01}+ \sqrt{\delta_1} \ket{10}, \quad \ket{\eta_3}= \sqrt{\delta_1} \ket{01}-\sqrt{\delta_0} \ket{10}, \label{a4}
\end{align}
where $\delta_0 \geq \delta_1 >0$. Hence, one can express each of the state $\ket{\chi_{i}}_{AC}=\dfrac{1}{\sqrt{p_i}} (A_{\alpha} \otimes B_{\beta})\ket{\eta_i}_{AC}$ as 
\begin{align}
    & \ket{\chi_0}_{AC}= \dfrac{1}{\sqrt{p_0}}\left(\sqrt{\delta_0 \alpha_0 \beta_0}  \ket{00} + \sqrt{\delta_1 \alpha_1 \beta_1 } \ket{11} \right) = \sqrt{A_0} \ket{00} + \sqrt{A_1}\ket{11} \nonumber \\
    & \ket{\chi_1}_{AC}= \dfrac{1}{\sqrt{p_1}}\left(\sqrt{\delta_1 \alpha_0 \beta_0 } \ket{00} - \sqrt{\delta_0 \alpha_1 \beta_1}  \ket{11} \right) = \sqrt{B_0} \ket{00} -\sqrt{B_1}\ket{11} \nonumber \\
    & \ket{\chi_2}_{AC}= \dfrac{1}{\sqrt{p_2}}\left(\sqrt{\delta_0 \alpha_0 \beta_1}  \ket{01} + \sqrt{\delta_1 \alpha_1 \beta_0}  \ket{10} \right)= \sqrt{C_0} \ket{01} +\sqrt{C_1} \ket{10} \nonumber \\
    & \ket{\chi_3}_{AC}= \dfrac{1}{\sqrt{p_3}}\left(\sqrt{\delta_1 \alpha_0 \beta_1}  \ket{01} - \sqrt{\delta_0 \alpha_1 \beta_0 } \ket{10} \right) = \sqrt{D_0} \ket{01} -\sqrt{D_1} \ket{10}, \label{a5}
\end{align}
where $ p_i $ is the normalization factor for every $i$ and $\lbrace A_i,~B_i,~C_i,~D_i\rbrace$ are the Schmidt coefficients of $\lbrace \ket{\chi_i}_{AC}\rbrace$. Now from the knowledge of $\lbrace \ket{\chi_i}_{AC}\rbrace$, one can figure out the optimal {\it fully entangled fraction} of every $\rho_{AC,\eta_i}$ as 
\begin{align}
    F^{*}(\rho_{AC,\eta_0})&=\dfrac{1}{2}\left( 1+2~\sqrt{A_0~A_1~(1-p)} -p~A_1\right), \quad \quad \text{if,} \quad \sqrt{A_0A_1(1-p)}\geq pA_1,\nonumber \\
    F^{*}(\rho_{AC,\eta_1})&=\dfrac{1}{2}\left( 1+2~\sqrt{B_0~B_1~(1-p)} -p~B_1\right), \quad \quad \text{if,} \quad \sqrt{B_0B_1(1-p)}\geq pB_1, \nonumber \\
    &\text{and,} \\
    F^{*}(\rho_{AC,\eta_0}) &=\dfrac{1}{2}\left( 1+ A_0~\dfrac{1-p}{p}\right), ~~~~~~~~~~~~~~~~\quad \quad \quad \quad \text{if,} \quad \sqrt{A_0A_1(1-p)}< pA_1, \nonumber \\
    F^{*}(\rho_{AC,\eta_1})&=\dfrac{1}{2}\left( 1+ B_0~\dfrac{1-p}{p}\right), ~~~~~~~~~~~~~~~~\quad \quad \quad \quad \text{if,} \quad \sqrt{B_0B_1(1-p)}< pB_1,
    \end{align}
    whereas, 
    \begin{align}
     F^{*}(\rho_{AC,\eta_2})&=\dfrac{1}{2}\left( 1+2~\sqrt{C_0~C_1~(1-p)} -p~C_1\right), \quad \quad \text{if,} \quad \sqrt{C_0C_1(1-p)}\geq pC_1, \nonumber \\
    F^{*}(\rho_{AC,\eta_3})&=\dfrac{1}{2}\left( 1+2~\sqrt{D_0~D_1~(1-p)} -p~D_1\right), \quad \quad \text{if,} \quad \sqrt{D_0D_1(1-p)}\geq pD_1, \nonumber \\
    &\text{and,} \\
     F^{*}(\rho_{AC,\eta_2}) &=\dfrac{1}{2}\left( 1+ C_0~\dfrac{1-p}{p}\right), ~~~~~~~~~~~~~~~~\quad \quad \quad \quad \text{if,} \quad \sqrt{C_0C_1(1-p)}< pC_1, \nonumber \\
     F^{*}(\rho_{AC,\eta_3})&=\dfrac{1}{2}\left( 1+ D_0~\dfrac{1-p}{p}\right), ~~~~~~~~~~~~~~~~\quad \quad \quad \quad \text{if,} \quad \sqrt{D_0D_1(1-p)}< pD_1, 
\end{align}
Now let us consider a particular domain where all four inequality conditions, 
\begin{align}
    & \quad \sqrt{A_0~A_1(1-p)}< p~A_1  \nonumber \\
    & \quad \sqrt{B_0~B_1(1-p)}< p~B_1 \nonumber \\
    & \quad \sqrt{C_0~C_1(1-p)}< p~C_1  \nonumber \\
    & \quad \sqrt{D_0~D_1(1-p)}< p~D_1  \nonumber 
\end{align}
are simultaneously satisfied. From these inequality conditions, one can argue that the following inequality conditions should be satisfied, 
\begin{align}
    (1-p)~\delta_0~\delta_1~\alpha_0~\alpha_1~\beta_0~\beta_1 < p^2~\delta^2_1~\alpha^2_1~\beta^2_1\leq p^2~\delta^2_0~\alpha^2_1~\beta^2_1 \leq p^2~\delta^2_0~\alpha^2_1~\beta^2_0,  \label{b10}
\end{align}
where note that ranges of $\alpha_0,~\beta_0,~\delta_0$ is given by $\dfrac{1}{2}\leq \alpha_0,~\beta_0,~\delta_0 <1$. Now if one trivially choose $\alpha_0=\beta_0 =\delta_0 =\dfrac{1}{2}$, then for this case, the inequality in Eq.$~$(\ref{b10}) will satisfy if the channel parameter $p$ is in the range $\dfrac{\sqrt{5}-1}{2}<p<1$. Note that the more general restriction on the channel parameter $p$ is  
\begin{align}
    p> \dfrac{1}{2}\left[ \sqrt{\dfrac{\alpha_0^2~\beta_0^2~\delta_0^2}{\alpha_1^2~\beta_1^2~\delta_1^2}+4\left(\dfrac{\alpha_0~\beta_0~\delta_0}{\alpha_1~\beta_1~\delta_1}\right)}-\dfrac{\alpha_0~\beta_0~\delta_0}{\alpha_1~\beta_1~\delta_1}\right] \geq \dfrac{\sqrt{5}-1}{2},
\end{align}
because of the fact that $\alpha_0~\beta_0~\delta_0 \geq \alpha_1~\beta_1~\delta_1$ holds in general. So as per the range of the parameters of $\alpha_0,~\beta_0,~\delta_0$, let us first assume that $\alpha_0 \geq \alpha_1$, whereas $\beta_0=\delta_0 =\dfrac{1}{2}$. With this, one can easily find the domain of allowed $\alpha_0$ as 
\begin{align}
    \dfrac{1}{2}\leq \alpha_0 < \dfrac{p^2}{1-p+p^2}, \quad \quad \quad \quad \text{while,} \quad \dfrac{\sqrt{5}-1}{2}<p<1. \label{b12}
\end{align}
Next, let us assume that another independent parameter $\beta_0$ satisfy $\beta_0 \geq \beta_1$ while also assuming $\alpha_0 \geq \alpha_1 $ from the restricted domain of Eq.$~$(\ref{b12}). Whereas, we assume, $\delta_0 =\delta_1 =\dfrac{1}{2}$. For this case, one can easily find out the restriction on $\beta_0$ as 
\begin{align}
    \dfrac{1}{2}\leq \beta_0 < \dfrac{p^2 ~(1-\alpha_0)}{p^2 + \alpha_0~(1-p+p^2)},  \quad \quad \text{while,} \quad \dfrac{1}{2}\leq \alpha_0 < \dfrac{p^2}{1-p+p^2}, \quad \text{and,} \quad \dfrac{\sqrt{5}-1}{2}<p<1. \label{b13}
\end{align}
Similarly, assuming the general case, where $\delta_0 \geq \delta_1$, one finds out the most general condition on the parameters as 
\begin{align}
    \dfrac{1}{2}\leq \delta_0 < \dfrac{p^2(1-\alpha_{0})(1-\beta_{0})}{\alpha_0 \beta_0 (1-p)+p^2(1-\alpha_{0})(1-\beta_{0})}, \label{b14}
\end{align}
while $\alpha_0$ and $\beta_0$ satisfy Eq.$~$(\ref{b12}) and Eq.$~$(\ref{b13}).  Comparing it with the result in \cite{PhysRevLett.134.160803}, we see that the behavior and features of optimal teleportation fidelity are quite similar. The only difference is in the range where you achieve the bound. This is because the ADC channel acting on a non-maximally entangled Bell-like state is similar to the structure of noisy states considered in \cite{PhysRevLett.134.160803}. Thus ADC channel indicates a way of physically modeling such states in a practical network scenario. So while all these conditions hold, one can easily find out {\color{red}\st{that}} the average {\it fully entangled fraction}, {\it i.e.,} 
\begin{align}
    F^*_{AC,\lbrace \eta_i\rbrace}&= \sum_{i=0}^3 p_i ~F^*(\rho_{AC,\eta_i}) \nonumber \\
&= \dfrac{1}{2}+\dfrac{(1-p)}{2p}(p_0~A_0+p_1~A_1+p_2~A_2+p_3~A_3) \nonumber \\
&= \dfrac{1}{2}\left[ 1+ \alpha_0~\dfrac{(1-p)}{p}\right] = \min \lbrace F^*(\rho_{AB_1}), F^*(\sigma_{B_2C})\rbrace =F^*_{AC}(\rho,~\sigma), \label{b15}
\end{align}
which ends the proof. A major point to be noted here is that the average optimal {\it fully entangled fraction} in Eq.$~$(\ref{b15}) is hitting the LOCC upper bound $\min \lbrace F^*(\rho_{AB_1}), F^*(\sigma_{B_2C})\rbrace$ which means Eq.$~$(\ref{b13}) and Eq.$~$(\ref{b14}) must obey the strict inequality condition $F^*(\rho_{AB_1})< F^*(\sigma_{B_2C})$. This means that the following inequality must hold, {\it i.e.,} 
\begin{align}
    \dfrac{1}{2}\left[ 1+ \alpha_0~\dfrac{(1-p)}{p}\right] < F^*(\sigma_{B_2C})=\dfrac{1+2\sqrt{\beta_0~(1-\beta_0)}}{2}, \label{b16}
\end{align}
if Eq.$~$(\ref{b13}) and Eq.$~$(\ref{b14}) are satisfied. From Eq.$~$(\ref{b16}) one can find a domain of $\beta_0$ as \begin{align}
    \dfrac{1}{2}\leq \beta_0 < \dfrac{p^2 ~(1-\alpha_0)}{p^2 + \alpha_0~(1-p+p^2)} <\dfrac{1}{2}\left[ 1+ \sqrt{1-\dfrac{\alpha_0^2~(1-p)^2}{p^2}}\right], \nonumber
\end{align}
which is consistent with Eq.$~$(\ref{b13}). Hence, this implies that both the ineqaulity conditions Eq.$~$(\ref{b13}) and Eq.$~$(\ref{b14}) are physically consistent. 
\end{proof}

\section{Proof that Protocol 1 outperforms Protocol 2 while the initial shared two-qubit states are mixture of Bell states }\label{belldiag}
Let us again recapitulate the descriptions of {\bf Protocol 1} and {\bf Protocol 2} as follows:

\paragraph*{\textbf{Protocol 1}: Consider a class of $\Lambda_{B \leftrightarrow A \rightarrow C}$ protocols, where Bob first performs a projective measurement in any maximally entangled basis and broadcasts the outcome to Alice and Charlie, and then Bob's system is discarded. After hearing from Bob, Alice and Charlie perform a one-way LOCC from Alice to Charlie. \\}
\vskip 0.2cm

\vskip 0.2cm
\paragraph*{\textbf{Protocol 2}: Consider a class of $\Lambda_{B \leftrightarrow A \rightarrow C}$ protocols, where Bob first performs a projective measurement in a complete basis $\lbrace \ket{\eta_i}\rbrace_{i=0}^3$ of only non-maximally entangled states which are Bell like and broadcasts the outcome to Alice and Charlie. After that Bob's system is discarded and a one-way LOCC is performed from Alice to Charlie. Here Bob's measurement basis is restricted as} 
\begin{align}
    &\ket{\eta_0} = \alpha_0 \ket{00}+ \alpha_1 \ket{11}, \quad \ket{\eta_1} = \alpha_1 \ket{00}- \alpha_0 \ket{11}, \nonumber \\
    & \ket{\eta_2} = \beta_0 \ket{01}+ \beta_1 \ket{10}, \quad \ket{\eta_1} = \beta_1 \ket{01}- \beta_0 \ket{10}, \nonumber
\end{align}
where $\alpha_0^2 > \alpha_1^2 >0$ and $\beta_0^2 > \beta_1^2 >0$ and $\alpha_0^2 + \alpha_1^2 = \beta_0^2 + \beta_1^2 =1$. 
\vskip 0.3cm
\textbf{Proposition 5.}  
    {\it Let Alice-Bob share a two-qubit entangled Bell diagonal state with eigenvalues $\lbrace p_i \rbrace$ such that $p_i \geq p_{i+1}$ and Bob-Charlie share another two-qubit entangled Bell diagonal state with eigenvalues $\lbrace q_i \rbrace$ such that $q_i\geq q_{i+1}$ for $i=0,1,2,3$. Within a restricted class of LOCC protocols $\lbrace \mathcal{P} \rbrace$, where Bob initiates the protocol $\mathcal{P}$ and depending on Bob's measurement, Alice - Charlie perform a one-way LOCC, the maximum achievable fully entangled fraction cannot surpass the following bound} 
    \begin{align}
        F^*_{\lbrace \mathcal{P} \rbrace }&\leq \left \lbrace \dfrac{1}{2},~\lambda_{max}\left[ (\Lambda_p \otimes \Lambda_q) \ket{\Phi_0}\bra{\Phi_0}\right] \right \rbrace \nonumber \\
        & \leq min \lbrace F_1^*~,~F_2^* \rbrace \nonumber
    \end{align}
    where $\Lambda_p$ and $\Lambda_q$ are Pauli qubit channels with Kraus operators $K_i= \sqrt{p_i}~\sigma_i$ and $L_j= \sqrt{q_j}~\sigma_j$ and the set $\mathcal{S}_{\mathcal{P}}$ includes all Bob assisted LOCC (Bob starts the protocol).

\begin{proof}
Let us assume that $\mathcal{AB}$ and $\mathcal{BC}$ both share  Bell-diagonal states, \textit{i.e.,} $\rho_{AB_1}=\sum_ip_i ~\ket{\Phi_i}_{AB_{1}}\bra{\Phi_i}$, and $\sigma_{B_2C}=\sum_j q_j ~\ket{\Phi_j}_{B_{2}C}\bra{\Phi_j}$.  Hence, the initial composite state $\eta_{AB_1B_2C}$ can be expressed as
    \begin{align}
        \eta_{AB_1B_2C}& = \sum_{i,j=0}^3 p_iq_j \ket{\Phi_i ~\Phi_j}_{AB_1B_2C} \bra{\Phi_i ~\Phi_j} \nonumber \\
       & = \sum_{i,j=0}^{3} p_i~q_j~(\sigma_i \otimes \mathbb{I} \otimes \mathbb{I} \otimes \sigma_j)\ket{\Phi_0 ~\Phi_0}_{AB_1B_2 C} \bra{\Phi_0 ~\Phi_0} (\sigma_i \otimes \mathbb{I} \otimes \mathbb{I} \otimes \sigma_j) \nonumber \\
       &= \sum_{i,j=0}^3 A_{ij}~ \ket{\Phi_0 ~\Phi_0}_{AB_1B_2 C} \bra{\Phi_0 ~\Phi_0} ~A_{ij}^{\dagger},
       \label{35}
    \end{align}

where $\sigma_i$ is the $i^{th}$ Pauli matrix and $A_{ij}= \sqrt{p_i~q_j}~(\sigma_i \otimes \mathbb{I} \otimes \mathbb{I} \otimes \sigma_j)$. Now we can rewrite the state $\ket{\Phi_0 }_{AB_1}\ket{\Phi_0}_{B_2C}$ as 
\begin{align}
    \ket{\Phi_0 ~ \Phi_0}_{AB_1B_2C} = \dfrac{1}{2}\sum_{k=0}^3 \ket{\psi^*_k}_{B_1B_2}\otimes \ket{\psi_k}_{AC}, \label{b3} 
\end{align}
where $\lbrace \ket{\psi_k}\rbrace _{k=0}^3$ forms a complete set of orthonormal basis and $\ket{\psi^{*}_k}$ is the complex conjugate of $\ket{\psi_k}$. So using the above identity one can again rewrite Eq.$~$(\ref{35}) as 
\begin{align}
    \eta_{AB_1B_2C}&= \sum_{i,j=0}^3 A_{ij}~\ket{\Phi_0 ~\Phi_0}_{AB_1B_2 C} \bra{\Phi_0 ~\Phi_0} ~A_{ij}^{\dagger} \nonumber \\
    &=\dfrac{1}{4} \sum_{k,l=0}^3 \ket{\psi^*_k}_{B_1B_2} \bra{\psi^*_l} \otimes \sum_{i,j=0}^3 \tilde{B}_{ij} \ket{\psi_k}_{AC}\bra{\psi_l} \tilde{B}_{ij}^{\dagger}, \label{b4} 
\end{align}
where $\tilde{B}_{ij}=\sqrt{p_iq_j} ~\sigma_i \otimes \sigma_j$ denotes an independent action of two Pauli channels on Alice and Charlie respectively. So if we consider two Pauli channels as $\Lambda_p$ and $\Lambda_q$ with Kraus operators $\lbrace \sqrt{p_i}~\sigma_i\rbrace$ and $\lbrace \sqrt{q_j}~\sigma_j\rbrace$, then Eq.$~$(\ref{b4}) can also be written as 
\begin{align}
    \eta_{AB_1B_2C}= \dfrac{1}{4} \sum_{k,l=0}^3 \ket{\psi^*_k}_{B_1B_2} \bra{\psi^*_l} \otimes (\Lambda_p \otimes \Lambda_q) \ket{\psi_k}_{AC}\bra{\psi_l}. \label{b5}
\end{align}
Now we will show what is the optimal three-party LOCC protocol between Alice, Bob and Charlie for maximizing {\it fully entangled fraction} between Alice-Charlie. Since we have already shown that it is sufficient to choose Charlie's operation to be trivial (identity). Therefore, either Alice starts the protocol or Bob. We first focus on Bob starting the protocol.  
\vskip 0.2cm 
\paragraph{{\bf Bob starts the protocol with PVM}:}
Performing PVM on an arbitrary basis $\lbrace \ket{\psi^*_k}_{B_1B_2}\rbrace$ shall prepare the following state between Alice-Charlie 
\begin{align}
    \rho_{AC}(\psi_k)=(\Lambda_p \otimes \Lambda_q) \ket{\psi_k} \bra{\psi_k},
\end{align}
with a probability $P_k=\tr\left[ (\ket{\psi_k^*}_{B_1B_2}\bra{\psi_k^*}\otimes \mathbb{I})\eta_{AB_1B_2C}\right]=\frac{1}{4}$. Now, since Bob is starting the protocol with PVM, we first need to optimize over all PVM basis that Bob must choose. 

Let us now define $ F_1(\rho_{AC,PVM})$ as the maximum average fully entangled fraction of the set of states $\lbrace \rho_{AC}(\psi_K)\rbrace$, where the maximization is taken over any PVM measurement performed by Bob. We define $F_2(\rho_{AC,PVM})$ as the maximum {\it fully entangled fraction} that is further optimized over all local post-processing between Alice and Charlie, given Bob's optimal choice of initial measurement. Hence, we can write 
\begin{align}
    F_{1}(\rho_{AC,PVM}) & = \sum_{k=0}^3 P_k ~F_1(\rho_{AC}(\psi_k)) \nonumber \\
    &= \sum_{k=0}^3 P_k ~\max_{\lbrace \ket{\psi^*_k}\rbrace} \max_{\ket{\Phi}\in MES} \bra{\Phi} \rho_{AC}(\psi_k)\ket{\Phi} \nonumber \\
    &\leq \sum_{k=0}^3 P_k~\max_{\Lambda_{AC} \in LOCC} \max_{\lbrace \ket{\psi^*_k}\rbrace} \max_{\ket{\Phi}\in MES} \bra{\Phi} \Lambda_{AC}(\rho_{AC}(\psi_k))\ket{\Phi} \nonumber \\
    & =F_{2}(\rho_{AC,PVM}).
\end{align}

Now we can find the expression of $F_1(\rho_{AC}(\psi_k))$ as 

\begin{align}
    F_1(\rho_{AC}(\psi_k))&= \max_{\ket{\Phi}\in MES} \max_{\ket{\psi_k}} \bra{\Phi} (\Lambda_p \otimes \Lambda_q)\rho_{\psi_k}\ket{\Phi}, \quad \quad \text{where,} \quad \rho_{\psi_k}=\ket{\psi_k}\bra{\psi_k} \nonumber \\
    &= \max_{\ket{\Phi}} \max_{\ket{\psi_k}} \sum_{i,j=0}^3 p_i q_j ~\bra{\Phi}(\sigma_i \otimes \sigma_j) \rho_{\psi_k} (\sigma_i \otimes \sigma_j)\ket{\Phi} \nonumber \\
    &= \max_{\ket{\Phi}} \max_{\ket{\psi_k}} \sum_{i,j=0}^3 p_i q_j ~\bra{\psi_k}(\sigma_i \otimes \sigma_j) \rho_{\Phi} (\sigma_i \otimes \sigma_j)\ket{\psi_k} \nonumber \\
    &= \max_{W} \max_{\ket{\psi_k}} \sum_{i,j=0}^3 p_i q_j ~\bra{\psi_k}(\sigma_i \otimes \sigma_j)(W \otimes \mathbb{I}) \ket{\Phi_0} \bra{\Phi_0}(W^{\dagger}\otimes \mathbb{I}) (\sigma_i \otimes \sigma_j)\ket{\psi_k} \nonumber \\
    &= \max_{W} \max_{\ket{\psi_k}} \sum_{i,j=0}^3 p_i q_j ~\bra{\psi_k}(\sigma_i~W~\sigma_j^T \otimes \mathbb{I}) \rho_{\Phi_0}(\sigma_j^{*}W^{\dagger}~\sigma_i\otimes \mathbb{I}) \ket{\psi_k} \nonumber \\
    &= \max_{W} \max_{\ket{\psi_k}}  \bra{\psi_k} (\Lambda_p \circ W \circ \Lambda^T_q\otimes \mathbb{I}) \rho_{\Phi_0} \ket{\psi_k} \nonumber \\
    &= \max_{W} \lambda_{max}\left[ (\Lambda_{p,W,q} \otimes \mathbb{I})\rho_{\Phi_0}\right], \label{b7}
\end{align}
where $W$ is a unitary from $SU(2)$ and note that the concatenated channel $\Lambda_{p,W,q}=\Lambda_p \circ W\circ \Lambda_q $ is a unital channel for any choice of $W$. Hence, from the observation of Ref. \cite{PhysRevA.90.052304}, one can conclude that the largest eigenvector of the state $(\Lambda_{p,W,q} \otimes \mathbb{I})\ket{\Phi_0} \bra{\Phi_0}$ which is $\ket{\psi}$ must be a MES. Note that since $\Lambda_{p,W,q}$ is unital, the state $(\Lambda_{p,W,q} \otimes \mathbb{I})\ket{\Phi_0} \bra{\Phi_0}$ can always be converted to a Bell-diagonal state using some product unitary operation $U \otimes V$ such that the largest eigenvector $\ket{\psi_k}=\ket{\Phi_0}$ for all $k=0,1,2,3$. \\ 

This implies, without any loss of generality, the best PVM Bob can perform is doing measurement in the standard Bell basis, {\it i.e.,} $\ket{\psi_k}=(\mathbb{I}\otimes \sigma_k)\ket{\Phi_0}$, for $k=0,1,2,3$. After suitable unitary corrections by Alice and Charlie, the post-measurement state of Alice-Charlie can be uniquely expressed in the following Bell mixture, 
\begin{align}
    \rho_{AC}(\Lambda_p,\Lambda_q)= \sum_{i=0}^3 \lambda_i \ket{\Phi_i}\bra{\Phi_i}, \quad \quad \lambda_0 \geq \lambda_1 \geq \lambda_2 \geq \lambda_3, \label{b8}
\end{align}
where $\lambda_0 =\max_{W}~\lambda_{max}\left[ (\Lambda_{p,W,q}\otimes \mathbb{I})\rho_{\Phi_0}\right]=\lambda_{max}\left[(\Lambda_{p}\circ \Lambda_q^T\otimes \mathbb{I})\rho_{\Phi_0} \right]$ is the largest eigenvalue of $\rho_{AC}(\Lambda_p,\Lambda_q)$. Note that $\rho_{AC}(\Lambda_p, \Lambda_q)$ is entangled \textit{iff} 
\begin{align}
    \lambda_0 = \lambda_{max}\left[(\Lambda_{p}\circ \Lambda_q^T\otimes \mathbb{I})\rho_{\Phi_0} \right] > \dfrac{1}{2}
\end{align}

Hence, the {\it fully entangled fraction} $F_{1}(\rho_{AC,PVM})$ can be written as
\begin{align}
    F_{1}(\rho_{AC,PVM})= &\lambda_{max}\left[(\Lambda_{p}\circ \Lambda_q^T\otimes \mathbb{I})\rho_{\Phi_0} \right]\nonumber\\
    =&\max\{(p_3 q_0+p_2 q_1+p_1 q_2+p_0 q_3),(p_2 q_0+p_3 q_1+p_0 q_2+p_1
   q_3),\nonumber\\
   &\qquad\qquad\qquad(p_1 q_0+p_0 q_1+p_3 q_2+p_2 q_3),(p_0 q_0+p_1 q_1+p_2 q_2+p_3
   q_3)\},
\end{align}
as $p_i\geq p_{i+1}$ and $q_i\geq q_{i+1}$ for $i=0,1,2,3$ using rearrangement inequality \cite{day1972rearrangement} we can conclude that
\begin{align}
    F_{1}(\rho_{AC,PVM})=\sum_{i=0}^3p_iq_i.
\end{align}
Now after Bob's action if we include post-processing between Alice and Charlie, we can obtain a value
\begin{align}
    F_{2}(\rho_{AC,PVM})= p_{succ}~F_1(\rho_{AC,PVM}) + \dfrac{1- p_{succ}}{2},  
\end{align}
where $p_{succ}$ represents a success probability that Bob performs PVM in MES and we obtain $F_1(\rho_{AC,PVM})> \frac{1}{2}$, otherwise if $F_1(\rho_{AC,PVM})\leq \frac{1}{2}$, then Alice and Charlie do not share any entangled state and hence, with zero success probability i.e., $p_{succ}=0$, Alice and Charlie replace their existing state by a product state giving a value at least small as $\frac{1}{2}$. 

This proves that {\bf Protocol 1} is more efficient than {\bf Protocol 2} while $\rho_{AB_1}$ and $\sigma_{B_2C}$ are mixture of Bell states.  

\vskip 0.5cm 
\paragraph{{\bf Bob starts with POVM}:} Suppose Bob starts the protocol by performing an $N$ outcome POVM operation with elements $\lbrace \pi_{i, B_1B_2}\rbrace$ such that $\sum_{i=1}^N \pi_{i,B_1B_2} =\mathbb{I}$. Note that each of the elements can be expressed as 
\begin{align}
    \pi_{i,B_1B_2}=\sum_{k=0}^3 q_{i,k} \ket{\psi^*_{i,k}}\bra{\psi^*_{i,k}}, \quad \quad \quad q_{i,k}\geq 0 \quad \forall i,k, \label{b11}
\end{align}
where $\sum_k q_{i,k}\leq 1$ and $\lbrace \ket{\psi_{i,k}^*}\rbrace$ forms a complete set of orthonormal basis. Note that the representation in Eq.$~$(\ref{b11}) is not unique.

So if Bob performs $\pi_{i, B_1B_2}$ then without any loss of generality, the post-measurement state between Alice-Charlie can be expressed directly from Eq.$~$(\ref{b4}) as 
\begin{align}
    \rho_{AC}(\pi_i)&= \dfrac{Tr_{B_1B_2}\left[ (\sqrt{\pi_{i,B_1B_2}}\otimes \mathbb{I})~\eta_{AB_1B_2C}~ (\sqrt{\pi_{i,B_1B_2}}\otimes \mathbb{I})^{\dagger}\right]}{\tr\left[ (\pi_{i,B_1B_2}\otimes \mathbb{I})\eta_{AB_1B_2C}\right]} \nonumber \\
    &= \dfrac{1}{Q_i}\sum_{k=0}^3 q_{i,k}~ (\Lambda_p \otimes\Lambda_q)\ket{\psi_{i,k}}\bra{\psi_{i,k}}, 
\end{align}
where $\lbrace \ket{\psi_{i,k}}\rbrace$ forms an orthonormal complete basis set for every $i$ and $Q_i= \tr\left[ (\pi_{i,B_1B_2}\otimes \mathbb{I})\eta_{AB_1B_2C}\right] =\sum_{k=0}^3 q_{i,k} $ is the probability of the measurement outcome. 

The {\it fully entangled fraction} of the two-qubit state $\rho_{AC}(\pi_i)$ has the following upper bound, 
\begin{align}
    F(\rho_{AC},\pi_i) &= \max_{\ket{\Phi}\in MES} \max_{\pi_i}\bra{\Phi} \rho_{AC}(\pi_i) \ket{\Phi} \nonumber \\
    &\leq  \max_{\lbrace q_{i,k}\rbrace}\dfrac{1}{Q_i}\sum_{k=0}^3 q_{i,k}~\max_{\ket{\psi_{i,k}}}~\max_{\Phi} \bra{ \Phi}(\Lambda_p \otimes\Lambda_q)\ket{\psi_{i,k}}\bra{\psi_{i,k}}\ket{\Phi} \nonumber \\ 
    &\leq \lambda_{max}\left[ (\Lambda_p \circ \Lambda_q^T\otimes \mathbb{I})\rho_{\Phi_0}\right]. \nonumber\\
    &\leq\sum_{i=0}^3p_iq_i.
\end{align}

This is because each of the extreme points $\ket{\psi_{i,k}}$ of the convex mixture $\lbrace q_{i,k}, \ket{\psi_{i,k}}\rbrace$ becomes optimal {\it iff} it is a MES. Hence, optimal POVM elements must be a PVM measurement in a complete set of MES. Now after this if we allow for further local post-processing by Alice and Charlie,  we can say that if $F(\rho_{AC}, \pi_i) \leq \frac{1}{2}$, then Alice and Charlie can replace their existing state with a product state. 

Hence, we prove that {\bf Protocol 1} is more efficient than {\bf Protocol 2} and even any protocol $\mathcal{P}$, where Bob initiates it.

\end{proof}

\section{Detail proof of {\bf Proposition} \ref{p5}} \label{appD}

\begin{proof}

Let us assume $\rho_{AB_1}=(\Lambda_{ADC}\otimes \mathbb{I})\ketbra{\Phi_0}{\Phi_0}$ and $\sigma_{B_2C}=\lambda\ketbra{\Phi_0}{\Phi_0}+(1-\lambda)\frac{\mathbb{I}}{4}$. We know that $\sigma_{B_2C}$ is entangled if $\lambda>\frac{1}{3}$. One can express $\sigma_{B_2C}$ as mixtures of Bell basis, 
\begin{align}
    \sigma_{B_2C} &= \lambda~ \ketbra{\Phi_0}{\Phi_0} +(1-\lambda)~\dfrac{\mathbb{I}}{4} \nonumber \\
    & = \left( \dfrac{1+3\lambda}{4}\right) \ketbra{\Phi_0}{\Phi_0} + \dfrac{1-\lambda}{4} \sum_{i=1}^3 \ketbra{\Phi_i}{\Phi_i} \nonumber \\
    &= F ~\ketbra{\Phi_0}{\Phi_0}+\dfrac{1-F}{3} \sum_{j=1}^3 (\mathbb{I}\otimes \sigma_j)\ketbra{\Phi_0}{\Phi_0}(\mathbb{I}\otimes \sigma_j) \nonumber \\
    &= (\mathbb{I}\otimes \Lambda_{dep})\ketbra{\Phi_0}{\Phi_0},
\end{align}
where $\Lambda_{dep}$ is the depolarzing channel with Kraus operators $\lbrace \sqrt{F}~\mathbb{I},~\sqrt{\frac{1-F}{3}}~\sigma_1,~\sqrt{\frac{1-F}{3}}~\sigma_2,~\sqrt{\frac{1-F}{3}}~\sigma_3 \rbrace$ respectively.

The composite four-qubit state $\eta_{AB_1B_2C}$ can be written as  (see Eq.$~$(\ref{b4})
\begin{align}
    \eta_{AB_1B_2C}&=\rho_{AB_1}\otimes \sigma_{B_2C} \nonumber \\
    & = \dfrac{1}{4} \sum_{k,l=0}^3 \ket{\psi^*_k}_{B_1B_2}\bra{\psi^*_l} \otimes (\Lambda_{ADC}\otimes \Lambda_{dep}) \ket{\psi_k}_{AC} \bra{\psi_l} , \label{c2}
\end{align}
where both the set $\lbrace \ket{\psi_k}\rbrace$ and $\lbrace \ket{\psi^*_k}\rbrace$ form complete set of orthonormal basis. 
$\\$
\paragraph{{\bf Bob performs PVM in any MES}:}
Now suppose Bob performs PVM in any orthogonal MES basis with four basis elements, say $\lbrace \ket{\Psi^*_i}\rbrace$ for $i=0,1,2,3$, where each state $\ket{\Psi^*_i}$ can be written as 
\begin{align}
    \ket{\Psi^*_i} &=(U^*\otimes V^*) ~\ket{\Phi_i} , \quad \quad (U^*)^{\dagger}U^*=(V^*)^{\dagger}V^*=\mathbb{I},~~ i=0,1,2,3 \nonumber \\
    & = (U^* \otimes \mathbb{I})(\mathbb{I} \otimes V^*\sigma_i) \ket{\Phi_0} \nonumber \\
    &= (\mathbb{I} \otimes V^*\sigma_i)(U^* \otimes \mathbb{I}) \ket{\Phi_0} \nonumber =(\mathbb{I}\otimes V^* \sigma_i (U^*)^T)\ket{\Phi_0}= (\mathbb{I} \otimes W^*_i) \ket{\Phi_0} ,
\end{align}
where $W^*_i = V^*\sigma_i U^{\dagger}$ is again a unitary from $SU(2)$. Now since the vectors $\lbrace \ket{\Phi^*_i}\rbrace$ are orthogonal, hence, 
\begin{align}
    \langle \Psi^*_i | \Psi^*_j \rangle = \bra{\Phi_0} (\mathbb{I} \otimes W_i^{T}W_j^*) \ket{\Phi_0} = \dfrac{1}{2}\tr(W_i^{T}W_j^*)= \delta_{ij}
\end{align}
must be satisfied. Thus Bob's measurement in $ \ket{\Psi^*_i}$ prepares a two-qubit state between Alice-Charlie as 
\begin{align}
    \rho_{AC,\Psi_i}& =\dfrac{1}{P_i}Tr_{B_1B_2}\left[ (\ket{\Psi^*_i}_{B_1B_2}\bra{\Psi^*_i}\otimes \mathbb{I})\eta_{AB_1B_2C}\right], \quad \quad i=0,1,2,3\nonumber \\
    & = (\Lambda_{ADC} \otimes \Lambda_{dep})\ketbra{\Psi_i}{\Psi_i} \nonumber \\ 
    &= (\Lambda_{ADC}\otimes \mathbb{I})\left[ F ~\ketbra{\Psi_i}{\Psi_i}+\frac{1-F}{3}\sum_{j=1}^3(\mathbb{I}\otimes\sigma_j) \ketbra{\Psi_i}{\Psi_i}(\mathbb{I}\otimes\sigma_j)\right], 
\end{align}
where note that $\lbrace \ket{\Psi_i}\rbrace$ forms an orthonormal basis set since their complex conjugates are assumed to be orthonormal. Now notice that for every $i=0,1,2,3$ the following four vectors 
\begin{align}
    \ket{\Psi_i}= \ket{\tilde{\Psi}_{0,i}}= (\mathbb{I}\otimes \sigma_0~W_i) \ket{\Phi_0}, \quad \ket{\tilde{\Psi}_{j,i}}=(\mathbb{I}\otimes \sigma_j)~\ket{\Psi_i}= (\mathbb{I}\otimes \sigma_j ~W_i) \ket{\Phi_0} \quad \quad j=1,2,3\nonumber  
\end{align}
are orthogonal because the following condition 
\begin{align}
    \langle \tilde{\Psi}_{k,i}| \tilde{\Psi}_{l,i}\rangle = \bra{\Phi_0} (\mathbb{I}\otimes W_i^{\dagger}~\sigma_k ~\sigma_l ~W_i)\ket{\Phi_0} =\dfrac{1}{2}\tr\left( W_i^{\dagger}~ \sigma_k ~\sigma_l ~W_i\right) =\delta_{kl} \quad \quad k,l=0,1,2,3\nonumber
\end{align}
is always satisfied. So without any loss of generality, $\rho_{AC,\Psi_i}$ can be modified as 
\begin{align}
    \rho_{AC,\Psi_i}= (\Lambda_{ADC}\otimes \mathbb{I})\left[ F ~\ketbra{\tilde{\Psi}_{0,i}}{\tilde{\Psi}_{0,i}}+\frac{1-F}{3}\sum_{j=1}^3 \ketbra{\tilde{\Psi}_{j,i}}{\tilde{\Psi}_{j,i}}\right], \quad \quad  \langle \tilde{\Psi}_{k,i}| \tilde{\Psi}_{l,i}\rangle = \delta_{kl} \quad and ~~  k,l=0,1,2,3\nonumber
\end{align}
Now again notice that for every $i=0,1,2,3$ the following four vectors,
\begin{align}
    \ket{\Phi_0}= \ket{\tilde{\Phi}_{0,i}}= (\mathbb{I}\otimes W_i^{\dagger}) \ket{\tilde{\Psi}_{0,i}} , \quad \quad \ket{\tilde{\Phi}_{j,i}}=  (\mathbb{I}\otimes W_i^{\dagger}) \ket{\tilde{\Psi}_{j,i}} \quad \quad j=1,2,3\nonumber
\end{align}
are again orthogonal to each other because of the simple fact that  
\begin{align}
    \langle \tilde{\Phi}_{k,i}| \tilde{\Phi}_{l,i} \rangle = \langle \tilde{\Psi}_{k,i}| \tilde{\Psi}_{l,i}\rangle =\delta_{kl}. 
\end{align}
Hence, without any loss of generality, the state $\rho_{AC,\Psi_i}$ can be written as 
\begin{align}
    \rho_{AC,\Psi_i}& = (\Lambda_{ADC}\otimes \mathbb{I}) (\mathbb{I}\otimes W_i) \left[ F ~\ketbra{\Phi_0}{\Phi_0}+\dfrac{1-F}{3} \sum_{j=1}^3\ketbra{\tilde{\Phi}_{j,i}}{\tilde{\Phi}_{j,i}}\right] (\mathbb{I}\otimes W_i^{\dagger})\nonumber \\
    & = (\Lambda_{ADC}\otimes \mathbb{I}) (\mathbb{I} \otimes W_i)\left[ \dfrac{4F-1}{3} \ketbra{\Phi_0}{\Phi_0} +\dfrac{1-F}{3} \mathbb{I} \right] (\mathbb{I} \otimes W_i^{\dagger}) \nonumber \\
    & = (\Lambda_{ADC} \otimes \mathbb{I}) (\mathbb{I} \otimes W_i) \left[ \lambda \ketbra{\Phi_0}{\Phi_0} + (1-\lambda) \dfrac{\mathbb{I}}{4}\right] (\mathbb{I} \otimes W_i^{\dagger}).
\end{align}
So after unitary correction by Charlie via unitary $W_i^{\dagger}$, the modified state becomes independent of $i$, that is  
\begin{align}
    (\mathbb{I}\otimes W_i^{\dagger})~\rho_{AC,\Psi_i} ~(\mathbb{I} \otimes W_i) = (\Lambda_{ADC} \otimes \mathbb{I}) \left[ \lambda \ketbra{\Phi_0}{\Phi_0} + (1-\lambda) \dfrac{\mathbb{I}}{4}\right]. \label{c7}
\end{align}
So Eq.$~$(\ref{c7}) implies that even if Bob performs PVM in any complete set of MES, Charlie can suitably find a unitary for which Alice-Charlie always end up with a unique state 
\begin{align}
    \rho_{AC,\Phi_0}&= (\Lambda_{ADC} \otimes \mathbb{I}) \left[ \lambda \ketbra{\Phi_0}{\Phi_0} + (1-\lambda) \dfrac{\mathbb{I}}{4}\right] \nonumber \\
    &= (\Lambda_{ADC}\otimes \Lambda_{dep})\ketbra{\Phi_0}{\Phi_0}.  
\end{align}
Therefore, one can argue that the average {\it fully entangled fraction} after optimal post-processing by Alice-Charlie can be expressed as 
\begin{align}
    F^*_{AC,\lbrace \ket{\Psi_i}\rbrace}=\dfrac{1}{4} \sum_{i=0}^3 F^*(\rho_{AC,\Psi_i}) = F^*(\rho_{AC,\Phi_0}).
\end{align}

Now let us fix the parameter $\lambda=\frac{2}{5}$ which implies that the pre-shared state $\sigma_{B_2C}$ is weakly entangled. So after fixing $\lambda$, one can evaluate $F^*(\rho_{AC,\Phi_0})$. For that one needs to solve the SDP (mentioned in Eq.$~$(\ref{9})) and by solving it, one can find the optimal feasible solutions 
\begin{align}
    F^*(\rho_{AC,\Phi_0}) & = \left \lbrace
    \begin{array}{cc}
       \dfrac{1}{20}\left( 7+4\sqrt{1-p}-2p\right)>\dfrac{1}{2},  & \quad \text{if,} \quad 0\leq p\leq \dfrac{1}{49}(4\sqrt{74}-29)  \\
       & \\
        \dfrac{67+112~p+21~p^2}{40(3+7~p)}>\dfrac{1}{2}, & \quad \text{if,} \quad \dfrac{1}{49}(4\sqrt{74}-29)<p<1/3 \\
        & \\
        \dfrac{1}{2}, & \quad \quad \text{if,} \quad \dfrac{1}{3}\leq p <1
    \end{array}\right . \label{c10}
\end{align}
Note that the solutions in Eq.$~$(\ref{c10}) are physically consistent because if one estimates the concurrence value of $\rho_{AC,\Phi_0}$, then one can find that 
\begin{align}
    C(\rho_{AC,\Phi_0})& = \max \left \lbrace 0, ~~\dfrac{1}{10}\sqrt{1-p}\left( \sqrt{9+21 ~p}- 4\right)\right \rbrace, \quad \quad \text{for,} \quad 0\leq p<1 \nonumber \\
    &= \dfrac{1}{10}\sqrt{1-p}\left( \sqrt{9+21 ~p}- 4\right) >0, \quad \quad \quad \quad ~~~~  \text{for,} \quad 0\leq p<1/3 \nonumber \\
    &=0, \quad \quad \quad \quad \quad \quad ~~~~~~~~~~~~~~~~~~~~~~~~~~~~~~~~~~~~~~~~~~~~ \text{for,} \quad \dfrac{1}{3}\leq p<1,
\end{align}
Since any $\rho_{AC,\Psi_i}$ is connected with $\rho_{AC,\Phi_0}$ via local unitary, the concurrence of $\rho_{AC,\Phi_0}$ is same as the concurrence of $ \rho_{AC,\Psi_i}~\forall i$. So in the entire range of channel parameter $p$, {\it i.e.,}  $0\leq p<1$, one can find that $F^*(\rho_{AC,\Phi_0})$ is always upper bounded by
\begin{align}
    & F^*(\rho_{AC,\Phi_0})\leq \dfrac{1+C(\rho_{AC,\Phi_0})}{2} \nonumber \\
    \text{Thus,} \quad  \quad & \dfrac{1}{4}\sum_{i=0}^3 F^*(\rho_{AC,\Psi_i})\leq \dfrac{1+C(\rho_{AC,\Phi_0})}{2}= \dfrac{1}{4} \sum_{i=0}^3 \dfrac{1+C(\rho_{AC,\Psi_i})}{2}
\end{align}
which is consistent in the entire range of $p\in [0,1)$. 
$\\$
\paragraph{{\bf Bob performs PVM in NMES}:}
Now let us consider that Bob performs PVM in some partially entangled basis, say, 
\begin{align}
    & \ket{\eta_0}=\frac{\sqrt{3}}{2}\ket{00}+\frac{1}{2}\ket{11}, \quad \quad \ket{\eta_1}=\frac{1}{2}\ket{00}-\frac{\sqrt{3}}{2}\ket{11},  \nonumber \\
    & \ket{\eta_2}=\frac{\sqrt{3}}{2}\ket{01}+\frac{1}{2}\ket{10}, \quad \quad \ket{\eta_3}=\frac{1}{2}\ket{01}-\frac{\sqrt{3}}{2}\ket{10}. \nonumber 
\end{align}
So in this case, the prepared state between Alice-Charlie becomes (see Eq.$~$(\ref{c2}))
\begin{align}
    \rho_{AC,\eta_i}= (\Lambda_{ADC}\otimes \Lambda_{dep})\ketbra{\eta_i}{\eta_i}, \quad \quad i=0,1,2,3,
\end{align}
which occurs with probability $q_i =\tr\left[ (\ket{\eta_i}_{B_1B_2}\bra{\eta_i}\otimes \mathbb{I})~\eta_{AB_1B_2C}\right]=\frac{1}{4}$. Although the average entanglement cost of such a basis is strictly less than the average cost of a MES basis. However, the interesting fact is that even if the cost is low, measurement in such a basis is more efficient for teleportation channel distribution than doing PVM in any MES, which means 
\begin{align}
    F^*_{AC,\lbrace \eta_i \rbrace}= \sum_{i=0}^3 q_i ~F^*(\rho_{AC,\eta_i}) > F^*(\rho_{AC,\Phi_0})
\end{align}
is satisfied for at least some values of $p$ and $\lambda$. We already have the feasible solutions (see Eq.$~$(\ref{c10})) for $\lambda=\frac{2}{5}$ while Bob performs PVM in any MES. So again we fix $\lambda=\frac{2}{5}$ and again find the feasible solutions for $F^*_{AC,\lbrace \eta_i \rbrace}$ as 
\begin{align}
     F^*(\rho_{AC,\eta_0})=F^*(\rho_{AC,\eta_2})&= \dfrac{21~p^2+238~p+381}{80(9+7p)}>\dfrac{1}{2}, \quad \quad \quad \text{if,} \quad 0<p<1 \nonumber \\
     \text{whereas,} \nonumber \\
     F^*(\rho_{AC,\eta_1})=F^*(\rho_{AC,\eta_3})&=  \dfrac{1}{240}(127-63~p)>\dfrac{1}{2}, \quad \quad \quad \quad \text{if,} \quad 0\leq p<\dfrac{1}{9} \nonumber \\
     &= \dfrac{1}{2}, \quad \quad \quad \quad ~~~~~~~~~~~~~~~~~~~~~~~~~~~ \text{if,} \quad \dfrac{1}{9}\leq p<1. 
    \end{align}
Note here that these solutions are consistent because each of them does not exceed the already established upper bound 
\begin{align}
    F^*(\rho_{AC,\eta_i})\leq \dfrac{1+C(\rho_{AC,\eta_i})}{2}, \quad \quad \quad \quad \forall i \quad \text{and for,} \quad 0\leq p<1, \nonumber
\end{align}
where $C(\rho_{AC,\eta_i})$ is the concurrence of the two-qubit state $\rho_{AC,\eta_i}$ and has the following expression, 
\begin{align}
    & C(\rho_{AC,\eta_0})=C(\rho_{AC,\eta_2})= \max \left \lbrace 0, ~~\dfrac{1}{20}(\sqrt{7~p+9}-4)\sqrt{3(1-p)}\right \rbrace  \nonumber \\
    & C(\rho_{AC,\eta_1})=C(\rho_{AC,\eta_3})= \max \left \lbrace 0,~~\dfrac{1}{20}(4-3\sqrt{7~p+1})\sqrt{3(1-p)}\right \rbrace , \nonumber 
\end{align}
for the entire range $0\leq p<1$. Therefore, within the range $\frac{1}{3}<p<1$, the average {\it fully entangled fraction} has the expression 
\begin{align}
    F^*_{AC,\lbrace \eta_i \rbrace }& = \dfrac{1}{4} \sum_{i=0}^3 F^*(\rho_{AC,\eta_i}) \nonumber \\
    &= \dfrac{1}{2}\left[ F^*(\rho_{AC,\eta_0})+F^*(\rho_{AC,\eta_1})\right] \nonumber \\
    &= \dfrac{1}{2}\left( \dfrac{21~p^2+238~p+381}{80(9+7p)} +\dfrac{1}{2}\right)> F^*(\rho_{AC,\Phi_0}) \label{c16}
\end{align}
is satisfied if $p>\frac{1}{3}$. Thus Eq.$~$(\ref{c16}) shows that PVM in partially entangled states is more efficient than PVM in any choice of MES basis while optimizing {\it fully entangled fraction} between Alice-Charlie. 
\end{proof}
\end{widetext}

\end{document}